\documentclass[a4paper]{amsart}
\usepackage{amssymb}
\usepackage{mathrsfs}
\usepackage{bm}

\usepackage{young}
\usepackage[enableskew,vcentermath]{youngtab}

\newtheorem{theorem}{Theorem}[section]
\newtheorem{lem}[theorem]{Lemma}
\newtheorem{cor}[theorem]{Corollary}
\newtheorem{prop}[theorem]{Proposition}
\theoremstyle{definition}
\newtheorem{defi}[theorem]{Definition}
\newtheorem{example}[theorem]{Example}

\theoremstyle{remark}
\newtheorem{rem}[theorem]{Remark}

\newcommand{\Glie}{\mathfrak{g}}   % Lie algebra
\newcommand{\Gaff}{\widehat{\mathfrak{g}}}   %% affine Lie algebra
\newcommand{\Gafft}{\widehat{\mathfrak{g}'}} %% affine Lie algebra transposed
\newcommand{\BP}{\mathbf{P}}       % weight lattice
\newcommand{\BQ}{\mathbf{Q}}       % root lattice
\newcommand{\BGG}{\mathcal{O}}   % category O
\newcommand{\CP}{\mathfrak{P}}    % weight
\newcommand{\lCP}{\widehat{\mathfrak{P}}}  % l-weight
\newcommand{\lCQ}{\widehat{\mathcal{Q}}}  % l-root lattice
\newcommand{\Bm}{\mathbf{m}}     % example l-weight

\newcommand{\CE}{\mathcal{E}}         % character
\newcommand{\lCE}{\mathcal{E}_{\ell}}  % q-character
\newcommand{\lwt}{\mathrm{wt}_{\ell}}  % set of l-weights
\newcommand{\wt}{\mathrm{wt}}         % set of weights
\newcommand{\BR}{\mathbf{R}}     % classifying sets
\newcommand{\qc}{\chi_q}      % q-character
\newcommand{\nqc}{\widetilde{\chi}_q}    % normalized q-character

\newcommand{\CB}{\mathcal{B}}    % tableaux
\newcommand{\Sm}{\mathbb{S}}     % antipode
\newcommand{\CF}{\mathscr{F}}      % inductive system
\newcommand{\SF}{\mathcal{F}}
\newcommand{\SG}{\mathcal{G}}
\newcommand{\CG}{\mathscr{G}}     % module map difference

\newcommand{\BC}{\mathbb{C}}  	% complex numbers
\newcommand{\BZ}{\mathbb{Z}}      % integers
\newcommand{\super}{\mathbb{Z}_2}   % parity
\newcommand{\odd}{\overline{1}}     % odd
\newcommand{\even}{\overline{0}}   % even
\newcommand{\End}{\mathrm{End}} % endomorphism
\newcommand{\ev}{\mathrm{ev}}    % evaluation
\newcommand{\Tr}{\mathrm{Tr}}    % trace

\newcommand{\Bf}{\mathbf{f}}
\newcommand{\Bg}{\mathbf{g}}

\newcommand{\Bn}{\mathbf{n}}     % 2-finite T
\newcommand{\BV}{\mathbf{V}}       % vector representation
\newcommand{\CW}{\mathscr{W}}    % asymptotic repr
\newcommand{\WW}{\mathcal{W}}    % Weyl module
\newcommand{\BW}{\mathbb{W}}      % asymptotic repr normalized
\newcommand{\CN}{\mathcal{N}}      % asymptotic T
\newcommand{\Bd}{\mathbf{d}}       % Demazure T
\newcommand{\aBw}{\bm{\omega}}      % asmptotic w
\newcommand{\CS}{\mathcal{S}}       % cyclicity criteria

\allowdisplaybreaks

\begin{document}
\title[Length-two representations]{Length-two representations of quantum affine superalgebras and Baxter operators}
\author{Huafeng Zhang}
\address{Laboratoire Paul Painlev\'e, Universit\'e Lille 1, 59655 Villeneuve d'Ascq, France }
\email{huafeng.zhang@math.univ-lille1.fr}
\begin{abstract}
Associated to quantum affine general linear Lie superalgebras are two families of short exact sequences of representations whose first and third terms are irreducible: the Baxter TQ relations involving infinite-dimensional representations; the extended T-systems of Kirillov--Reshetikhin modules. We make use of these representations over the {\it full} quantum affine superalgebra to define Baxter operators as transfer matrices for the quantum integrable model and to deduce Bethe Ansatz Equations, under genericity conditions. 
\end{abstract}
\maketitle
\setcounter{tocdepth}{1}
\tableofcontents
\section*{Introduction}
Fix $\Glie := \mathfrak{gl}(M|N)$ a general linear Lie superalgebra and $q$ a non-zero complex number which is not a root of unity. Let $U_q(\Gaff)$ be the associated quantum affine superalgebra \cite{Y}. This is a Hopf superalgebra neither commutative nor co-commutative, and it can be seen as a $q$-deformation of the universal enveloping algebra of the affine Lie superalgebra of central charge zero $\Gaff := \Glie \otimes \BC[t,t^{-1}]$. 

In this paper we study a tensor category of (finite- and infinite-dimensional) representations of $U_q(\Gaff)$. Its Grothendieck ring turns out to be commutative as is common in Lie Theory. We produce various identities of isomorphism classes of representations, and interpret them as functional relations of transfer matrices in the quantum integrable system attached to $U_q(\Gaff)$, the XXZ spin chain. 

\medskip

\noindent {\bf 1. Baxter operators.} In an exactly solvable model a common problem is to find the spectrum of a family $T(z)$ of commuting endomorphisms of a vector space $V$ depending on a complex spectral parameter $z$, called transfer matrices. The Bethe Ansatz method, initiated by H. Bethe, gives explicit eigenvectors and eigenfunctions of $T(z)$ in terms of solutions to a system of algebraic equations, the Bethe Ansatz equations (BAE). Typical examples are the Heisenberg spin chain and the ice model.

 In \cite{Baxter72}, for the 6-vertex model R. Baxter related $T(z)$ to another family of commuting endomorphisms $Q(z)$ on $V$ by the relation:
$$ \mathrm{TQ\ relation}: \quad\quad T(z) = a(z) \frac{Q(zq^2)}{Q(z)} + d(z) \frac{Q(zq^{-2})}{Q(z)}. $$
Here $a(z), d(z)$ are scalar functions and $q$ is the parameter of the model. $Q(z)$ is a polynomial in $z$, called Baxter operator. The cancellation of poles at the right-hand side becomes Bethe Ansatz equations for the roots of $Q(z)$. Similar operator equation holds for the 8-vertex model \cite{Baxter72}, where the Bethe Ansatz method fails.

Within the framework of Quantum Inverse Scattering Method, the transfer matrix $T(z)$ is defined in terms of  representations of a quantum group $\mathbf{U}$. Let $\mathcal{R}(z) \in \mathbf{U}^{\otimes 2}$  be the universal R-matrix with spectral parameter $z$ and let $V, W$ be two representations of $\mathbf{U}$. Then $t_W(z) := \mathrm{tr}_W (\mathcal{R}(z)_{W\otimes V})$ forms a commuting family of endomorphisms on $V$, thanks to the quasi-triangularity of $(\mathbf{U}, \mathcal{R}(z))$. As examples, the transfer matrix for the 6-vertex model (resp. XXX spin chain) comes from tensor products of two-dimensional irreducible representations of the affine quantum group $U_q(\widehat{\mathfrak{sl}_2})$ (resp. Yangian $Y_{\hbar}(\mathfrak{sl}_2)$), while the face-type model of Andrews--Baxter--Forrester, which is equivalent to the 8-vertex model by a vertex-IRF correspondence, requires Felder's elliptic quantum group $E_{\tau,\eta}(\mathfrak{sl}_2)$ \cite{Felder,FV2}. 

The representation meaning of the $Q(z)$ was understood in the pioneer work of Bazhanov--Lukyanov--Zamolodchikov \cite{BazhanovLukyanovZamolodchikov1997} for $U_q(\widehat{\mathfrak{sl}_2})$, and extended to an arbitrary non-twisted affine quantum group $U_q(\widehat{\mathfrak{a}})$ of a finite-dimensional simple Lie algebra $\mathfrak{a}$ in the recent work of Frenkel--Hernandez \cite{FH}. One observes that the first tensor factor of $\mathcal{R}(z)$ lies in a Borel subalgebra $U_q(\mathfrak{b})$ of $U_q(\widehat{\mathfrak{a}})$, so the above transfer-matrix construction makes sense for $U_q(\mathfrak{b})$-modules. Notably the Baxter operators $Q(z)$ are transfer matrices of $L_{i,a}^+$, the  {\it positive prefundamental modules} over $U_q(\mathfrak{b})$, for $i$ a Dynkin node of  $\mathfrak{a}$ and $a \in \BC^{\times}$.
 The $L_{i,a}^+$ are irreducible objects of a category $\BGG_{\mathrm{HJ}}$ of $U_q(\mathfrak{b})$-modules introduced by Hernandez--Jimbo \cite{HJ}.  

Making use of the prefundamental modules, Frenkel--Hernandez \cite{FH} solved a conjecture of Frenkel--Reshetikhin \cite{FR} on the spectra of the quantum integrable system, which connects eigenvalues of transfer matrices $t_W(z)$, for $W$ finite-dimensional $U_q(\widehat{\mathfrak{a}})$-modules, with polynomials arising as eigenvalues of the Baxter operators.

The two-term TQ relations, as a tool to derive Bethe Ansatz Equations for the roots of Baxter polynomials, are consequences of identities in the Grothendieck ring $K_0(\BGG_{\mathrm{HJ}})$ of category $\BGG_{\mathrm{HJ}}$ \cite{FH,HL,Jimbo1,Jimbo2,FH2}. Such identities are also examples of cluster mutations of Fomin--Zelevinsky \cite{HL}.

In the elliptic case, the triangular structure of $\mathcal{R}(z)$ is less clear as there is not yet a formulation of Borel subalgebras. Still the eigenvalues of $T(z)$ admit TQ relations  by a Bethe Ansatz in \cite{FV2}. In a joint work with G. Felder \cite{FZ}, we were able to construct elliptic Baxter operator $Q(z)$ for $E_{\tau,\eta}(\mathfrak{sl}_2)$ as a transfer matrix of certain infinite-dimensional representations over the full elliptic quantum group. 

Then a natural question is whether the Baxter operators can always be realized from representations of the full quantum group (of type Yangian, affine, or elliptic). Inspired by \cite{FZ}, in the present paper we provide a partial answer for the quantum affine superalgebra $U_q(\Gaff)$, based on the {\it asymptotic representations} which we introduced in a previous work \cite{Z5}. 

Let us mention the appearance of quantum affine superalgebras and Yangians in other supersymmetric integrable models like deformed Hubbard model and anti de Sitter/conformal field theory correspondences; see \cite{Beisert,Beisert1} and references therein.

Compared to the intense works on affine quantum groups (see the reviews \cite{CH,L}), the representation theory of $U_q(\Gaff)$ is still less understood as the super case poses one essential difficulty, the smallness of Weyl group symmetry. 

\medskip

\noindent {\bf 2. Asymptotic representations.} Before stating the main results of this paper, let us recall from \cite{Z5} the asymptotic modules over $U_q(\Gaff)$.

Let $I_0 := \{1,2,\cdots,M+N-1\}$ be the set of Dynkin nodes of the Lie superalgebra $\Glie$. There are $U_q(\Gaff)$-valued power series $\phi_i^{\pm}(z)$ in $z^{\pm 1}$ for $i \in I_0$ whose coefficients mutually commute; they can be viewed as $q$-analogs of $A \otimes t^{\pm n} \in \Gaff$
with $A$ being a diagonal matrix in $\Glie$ and $n$ a positive integer. Algebra $U_q(\Gaff)$  admits a triangular decomposition whose Cartan part is generated by the $\phi_i^{\pm}(z)$. The highest weight representation theory built from this decomposition is suitable for the classification of finite-dimensional irreducible representations \cite{Z1} in terms of rational functions. 

Fix a Dynkin node $i \in I_0$ and a spectral parameter $a \in \BC^{\times}$.  To each positive integer $k$ is attached a {\it Kirillov--Reshetikhin module}. It is a finite-dimensional irreducible $U_q(\Gaff)$-module generated by a highest weight vector $\omega$ such that
$$ \phi_j^{\pm}(z) \omega = \omega \quad \mathrm{if}\ j \neq i,\quad \phi_i^{\pm}(z) \omega = \frac{q_i^k - za q_i^{-k}}{1-za} \omega.  $$
Here $q_i = q$ for $i \leq M$ and $q_i = q^{-1}$ for $i > M$.
In \cite{Z5}, we made an ``analytic continuation" by taking $q_i^k$ to be a fixed $c \in \BC^{\times}$ as $k \rightarrow \infty$ to obtain a $U_q(\Gaff)$-module $\CW_{c,a}^{(i)}$. This is what we call asymptotic module. It  is a modification of the limit construction of prefundamental modules over Borel subalgebras in \cite{BazhanovLukyanovZamolodchikov1997,HJ}.  

We defined in \cite{Z5} a category $\BGG_{\Glie}$ of representations of $U_q(\Gaff)$ by imposing the standard weight condition as for Kac--Moody algebras \cite{Kac} and dropping integrability condition \cite{H2,MY}. It contains the $\CW_{c,a}^{(i)}$ and all the finite-dimensional $U_q(\Gaff)$-modules. Category $\BGG_{\Glie}$ is monoidal and abelian. \footnote{In the main text we also study category $\BGG$ of representations of a Borel subalgebra of $U_q(\Gaff)$, which admits prefundamental modules as in \cite{HJ}; see Definition \ref{def: category BGG}. Here $\BGG_{\Glie}$ is the full subcategory of $\BGG$ consisting of $U_q(\Gaff)$-modules.}

\medskip

\noindent {\bf 3. Main results.} We prove the following property of Grothendieck ring $K_0(\BGG_{\Glie})$: 
\begin{itemize}
\item[(i)] If $\CW$ is an asymptotic module, then there exist three modules $D, S', S''$ in category $\BGG_{\Glie}$ such that $[D] [\CW] = [S'] + [S'']$ and $S',S''$ are tensor products of asymptotic modules; see Theorem \ref{thm: TQ asymptotic}.
\end{itemize}

Consider the XXZ spin chain of $U_q(\Gaff)$. For $i \in I_0$, we define the Baxter operator $Q_i(u)$ to be the transfer matrix of $\CW_{u,1}^{(i)}$ evaluated at $1$ (Definition \ref{def: Baxter operators}), as in the elliptic case \cite{FZ}. To justify the definition, we prove the following facts.
\begin{itemize}
\item[(ii)] If $V$ is a  finite-dimensional $U_q(\Gaff)$-module, then  $t_V(z^{-2})$ is a sum of monomials of the $d(z) \frac{Q_i(zac)}{Q_i(za)}$ where $i \in I_0, a,c \in \BC^{\times}$, and the $d(z)$ are scalar functions, the number of terms being $\dim V$; see Corollary \ref{cor: generalized TQ relations}.
\item[(iii)] Each $Q_i(z)$ satisfies a two-term TQ relation; see Equation \eqref{equ: Baxter TQ}.
\end{itemize}
Note that (ii) reduces the transfer matrix of an {\it arbitrary} finite-dimensional $U_q(\Gaff)$ to the {\it finite} set $\{Q_i(u)\ |\ i \in I_0\}$ up to scalar functions. It forms generalized Baxter TQ relations in the sense of Frenkel--Hernandez \cite{FH}.

\medskip

\noindent {\bf 4. Proofs.} 
This requires the $q$-character map of Frenkel--Reshetikhin \cite{FR}, which is an injective ring homomorphism from the Grothendieck ring  $K_0(\BGG_{\Glie})$ to a commutative ring of $I_0$-tuples of rational functions with parity (Proposition \ref{prop: q-char ring}). 

The $q$-character of an asymptotic module is fairly easy thanks to its limit construction in \cite{Z5}. We obtain a separation of variable identity (SOV, Lemma \ref{lem: separation of variables}),
$$  [\CW_{c,1}^{(i)}] [ \CW_{1,a^2}^{(i)}] = [\CW_{ca,a^2}^{(i)}] [ \CW_{a^{-1},1}^{(i)}] \in K_0(\BGG_{\Glie}). $$
This identity puts the parameters $c,a \in \BC^{\times}$ in $\CW_{c,a}^{(i)}$ at an equal role. It categorifies  
$$\frac{c-zc^{-1}}{1-z}  \times \frac{1-za^2}{1-za^2} = \frac{ca-zc^{-1}a}{1-za^2} \times \frac{a^{-1}-za}{1-z}. $$
In \cite{Z5} we established generalized TQ relations in category $\BGG_{\Glie}$, which together with SOV proves (ii). 
Similarly (iii) follows from (i) and SOV.

Along the proof of (i) we obtain results of independent interest:
\begin{itemize}
\item $q$-character formulas of four families of finite-dimensional irreducible $U_q(\Gaff)$-modules, including all the Kirillov--Reshetikhin modules (Theorem \ref{thm: q-char MA}); 
\item a criteria for a tensor product of Kirillov--Reshetikhin modules to admit an irreducible head (i.e. of highest weight, Theorem \ref{thm: tensor KR even}); 
\item short exact sequences  of tensor products of Kirillov--Reshetikhin modules (Theorem \ref{thm: Demazure T}). 
\end{itemize}
The third point includes the T-system \cite{Nakajima,H,Tsuboi} as a special case.

\medskip

%The proof of Theorem \ref{thm: TQ asymptotic} requires an extension of $\BGG_{\Glie}$ to a category $\BGG$ of representations of a Borel subalgebra $Y_q(\Glie)$ of $U_q(\Gaff)$. Category $\BGG$ contains positive/negative prefundamental modules $L_{i,a}^{\pm}$. The asymptotic limits \cite{HJ,Z5} connect the $U_q(\Gaff)$-modules $\CW_{c,a}^{(i)}$ with the $Y_q(\Glie)$-modules $L_{i,a}^-$. TQ relations are derived for the $L_{i,a}^{\pm}$  (Theorem \ref{thm: TQ} and Corollary \ref{cor: TQ negative}) similarly as \cite{HL,Jimbo2}. These relations are translated into the $\CW_{c,a}^{(i)}$ by a duality functor $\SG^*$ of categories $\BGG$ (Lemma \ref{lem: duality by permutation}).  A crucial step is the asymptotic construction of $\CN_{c,a}^{(i)}$ in Section \ref{sec: asym}. The logic goes as follows (for simplicity we drop the Dynkin node $i$ and spin/spectral parameters $a,c$):
%$$ N^+L^+\eqref{equ: TQ} \xrightarrow[\text{duality}]{\SG^*} N^-L^- \eqref{equ: TQ negative} \xrightarrow[\text{replacement}]{\text{asymptotic}}  \CN\CW\eqref{equ: negative TQ asy} \xrightarrow[\text{duality}]{\SG^*} M \CW \eqref{equ: positive TQ asy}. $$
%Along the way, 
%When $s = 0$, Theorem \ref{thm: Demazure T} reduces to the T-system \cite{Nakajima,H,Tsuboi}.
% 
%\medskip

\noindent {\bf 5. Perspectives.}
We expect that our main results (i)--(iii) have analogy in elliptic quantum groups $E_{\tau,\hbar}(\mathfrak{a})$, based on twistor theory relating affine quantum groups to elliptic quantum groups \cite{twist,GTL,K2}. For $\mathfrak{a} = \mathfrak{sl}_N$ this has been verified in \cite{FZ,Z6}.
For $\mathfrak{a}$ of general type, a category of $E_{\tau,\hbar}(\mathfrak{a})$-modules was studied in \cite{GTL1} with well-behaved $q$-character theory, although its tensor product structure is unclear. 

It is possible to adapt the arguments to the case of Yangians (not necessarily of type A) in view of \cite{GTL}. One could avoid {\it degenerate Yangian} \cite{B,B2,QYangian}, whose prefundamental representations lead to Baxter operators but do not carry natural action of the ordinary Yangian. \cite[Appendix]{FZ} discussed the $\mathfrak{gl}_2$ case. The Yangian of centrally extended $\mathfrak{psl}(2|2)$ \cite{Beisert1} is of special interest in AdS/CFT. We do not know of any representation category $\BGG$ with well-behaved highest weight theory, yet there are limit constructions of infinite-dimensional representations \cite{Tor}.

For twisted quantum affine algebras $\mathbf{U}$, there are conjectural TQ relations in category $\BGG_{\mathrm{HJ}}$ \cite{FH2}. One may ask for such relations in terms of $\mathbf{U}$-modules. This is interesting from another point of view: the correspondence between twisted quantum affine algebras and non-twisted quantum affine superalgebras \cite{Zr1,W}. (This is different from Langlands duality in that the Cartan matrices for these algebras are identical.) A typical example is the equivalence \cite{W} of categories $\BGG_{\mathrm{int}}$ of integrable representations over $U_q(A_{2n}^{(2)})$ and $U_q(\widehat{\mathfrak{osp}(1|2n)})$. Let us mention an earlier work of Z. Tsuboi \cite{Tsuboi1} on Bethe Ansatz Equations  for orthosymplectic Lie superalgebras, the representation theory meaning of which is to be understood. One should need the Drinfeld second realization of quantum affine superalgebras \cite{Zr2}.

\medskip

The paper is structured as follows. In Section \ref{sec: basics} we review the quantum affine superalgebra $U_q(\Gaff)$ and its Borel subalgebra $Y_q(\Glie)$, and study the basic properties of category $\BGG$ of $Y_q(\Glie)$-modules. Section \ref{sec: asy TQ} presents the main result (i). In Section \ref{sec: Baxter}, for the $U_q(\Gaff)$ XXZ spin chain, we construct Baxter operators from the $\CW_{c,a}^{(i)}$ and derive Bethe Ansatz Equations from (i). 

The two basics ingredients are: the $q$-character formulas in terms of Young tableaux, proved in Section \ref{sec: q-char}; cyclicity of tensor products of Kirillov--Reshetikhin modules studied in Section \ref{sec: cyclicity}. The $q$-characters already lead to TQ relations of positive prefundamental modules over $Y_q(\Glie)$ in Sections \ref{sec: TQ}--\ref{sec: proof TQ}. The proof of (i) is completed in Section \ref{sec: asym} upon realizing $D$ as a suitable asymptotic limit.

The extended T-systems of Kirillov--Reshetikhin modules are proved in Section \ref{sec: proof T}. Although they are not needed in the proof of the main theorem, we include them here as applications of $q$-characters and cyclicity.

\medskip

 {\bf Acknowledgments.} The author thanks Vyjayanthi Chari, Giovanni Felder, David Hernandez and Marc Rosso for enlightening discussions, and the anonymous referee for valuable comments. This work was supported by the National Center of Competence in Research SwissMAP---The Mathematics of Physics of the Swiss National Science Foundation.

\section{Basics on quantum affine superalgebras} \label{sec: basics}
Fix $M,N \in \BZ_{>0}$. In this section we collect basic facts on the quantum affine superalgebra associated with the general linear Lie superalgebra $\Glie := \mathfrak{gl}(M|N)$ and its representations. The main references are \cite{Z3,Z4,Z5}, some of whose results are modified to be coherent with the non-graded quantum affine algebras.

Set $\kappa := M+N,\ I := \{1,2,\cdots,\kappa \}$ and $I_0 := I \setminus \{\kappa\}$. Let $\super$ denote the ring $\BZ/2\BZ = \{\even,\odd\}$. The {\it weight lattice} $\BP$ is the abelian group freely generated by the $\epsilon_i$ for $i \in I$. Let $|?|$ be the morphism of additive groups $\BP \longrightarrow \super$ such that 
$$ |\epsilon_1| = |\epsilon_2| = \cdots = |\epsilon_M| = \even,\quad |\epsilon_{M+1}| = |\epsilon_{M+2}| = \cdots = |\epsilon_{M+N}| = \odd. $$
$\BP$ is equipped with a symmetric bilinear form $(,): \BP \times \BP \longrightarrow \BZ$,
$$ (\epsilon_i,\epsilon_j) = \delta_{ij} (-1)^{|\epsilon_i|} \quad \mathrm{where}\   (-1)^{\even} := 1,\ (-1)^{\odd} := -1. $$
Define $\alpha_i := \epsilon_i - \epsilon_{i+1}$ for $i \in I_0$, and the {\it root lattice} $\BQ$ to be the subgroup of $\BP$ generated by the $\alpha_i$. Set $q_{l} := q^{(\epsilon_l,\epsilon_l)}$ and $q_{ij} := q^{(\alpha_i,\alpha_j)}$ for $i,j \in I_0$ and $l \in I$. 

If $W$ is a vector superspace and $w \in W$ is a $\super$-homogeneous vector, then by abuse of language let $|w| \in \super$ denote the parity of $w$. (It is not to be confused with the absolute value $|n|$ of an integer $n$.)

Let $\BV$ be the vector superspace with basis $(v_i)_{i\in I}$ and parity $|v_i| := |\epsilon_i|$. Define the elementary matrices $E_{ij} \in \mathrm{End}(\BV)$ by $E_{ij} v_k = \delta_{jk} v_i$ for $i,j,k \in I$. They form a basis of the vector superspace $\End(\BV)$ and $|E_{ij}| = |\epsilon_i| + |\epsilon_j|$.
\subsection{Quantum superalgebras}\label{subsec: quantum superalgebra} 
Recall the Perk--Schultz matrix \cite{Perk-Schultz}
\begin{align*}  
R(z,w) &=  \sum\limits_{i\in I}(zq_i - wq_i^{-1}) E_{ii} \otimes E_{ii}  + (z-w) \sum\limits_{i \neq j} E_{ii} \otimes E_{jj} \\
&\quad  + z \sum\limits_{i<j} (q_i-q_i^{-1}) E_{ji} \otimes E_{ij} + w \sum\limits_{i<j}(q_j-q_j^{-1})  E_{ij} \otimes E_{ji}.  
\end{align*}
It is well-known that $R(z,w)$ satisfies the quantum Yang--Baxter equation:
$$R_{12}(z_1,z_2)R_{13}(z_1,z_3) R_{23}(z_2,z_3) = R_{23}(z_2,z_3)R_{13}(z_1,z_3)R_{12}(z_1,z_2) \in \End (\BV)^{\otimes 3}. $$
The convention for the tensor subscripts is as usual. Let $n \geq 2$ and $A_1,A_2,\cdots,A_n$ be unital superalgebras. Let $1 \leq i < j \leq n$. If $x \in A_i$ and $y \in A_j$, then 
\begin{displaymath}
(x\otimes y)_{ij} := (\otimes_{k=1}^{i-1} 1_{A_k}) \otimes x \otimes (\otimes_{k=i+1}^{j-1} 1_{A_k}) \otimes y \otimes (\otimes_{k=j+1}^n 1_{A_k}) \in \otimes_{k=1}^n A_k.
\end{displaymath}
Now we can define the quantum affine superalgebra associated to $\Glie$.
\begin{defi}\cite[Section 3.1]{Z3}      \label{def: quantum affine superalgebras}
$U_q(\Gaff)$ is the superalgebra with presentation: 
\begin{itemize}
\item[(R1)] RTT-generators $s_{ij}^{(n)}, t_{ij}^{(n)}$ of parity $|\epsilon_i|+|\epsilon_j|$ for $i,j \in I$ and $n \in \BZ_{\geq 0}$;
\item[(R2)] RTT-relations in $U_q(\Gaff) \otimes (\End (\BV)^{\otimes 2})[[z,z^{-1},w,w^{-1}]]$
\begin{eqnarray*}
&& R_{23}(z,w) T_{12}(z) T_{13}(w) = T_{13}(w) T_{12}(z) R_{23}(z,w),  \\
&& R_{23}(z,w) S_{12}(z) S_{13}(w) = S_{13}(w) S_{12}(z) R_{23}(z,w),    \\
&& R_{23}(z,w) T_{12}(z) S_{13}(w) = S_{13}(w) T_{12}(z) R_{23}(z,w); 
\end{eqnarray*}
\item[(R3)] $t_{ij}^{(0)} = s_{ji}^{(0)} = 0$ and $s_{kk}^{(0)} t_{kk}^{(0)} = 1$ for $i,j,k \in I$ and $i<j$.
\end{itemize}
$T(z)  \in U_q(\Gaff)\otimes \End (\BV)[[z^{-1}]]$ and $S(z) \in U_q(\Gaff)\otimes \End (\BV)[[z]]$ are power series
\begin{align*}
&T(z) = \sum_{ij} t_{ij}(z) \otimes E_{ij},\quad t_{ij}(z) = \sum_{n \in \BZ_{\geq 0}} t_{ij}^{(n)} z^{-n}, \\
& S(z) = \sum_{ij} s_{ij}(z) \otimes E_{ij},\quad s_{ij}(z) = \sum_{n \in \BZ_{\geq 0}} s_{ij}^{(n)} z^{n}.
\end{align*}
The {\it Borel subalgebra} $Y_q(\Glie)$, also called $q$-Yangian \footnote{This is because the algebra $Y_q(\Glie)$ admits an $RTT=TTR$ type presentation, as does the ordinary Yangian $Y(\Glie)$. Here $q$ is a parameter of $R$. }, is the subalgebra of $U_q(\Gaff)$ generated by the $s_{ij}^{(n)}$ and $(s_{ii}^{(0)})^{-1}$. The finite-type quantum supergroup $U_q(\Glie)$ is the subalgebra of $U_q(\Gaff)$ generated by the $s_{ij}^{(0)}$ and $t_{ij}^{(0)}$.
\end{defi}
$U_q(\Gaff)$ has a Hopf superalgebra structure with counit $\varepsilon: U_q(\Gaff) \longrightarrow \BC$ defined by $\varepsilon(s_{ij}^{(n)}) = \varepsilon(t_{ij}^{(n)}) = \delta_{ij}\delta_{n0}$, and coproduct $\Delta: U_q(\Gaff) \longrightarrow U_q(\Gaff)^{\otimes 2}$:
\begin{eqnarray*}   
&&\Delta (s_{ij}^{(n)}) = \sum_{m=0}^n \sum_{k \in I} \epsilon_{ijk}  s_{ik}^{(m)} \otimes s_{kj}^{(n-m)}, \quad  \Delta (t_{ij}^{(n)}) = \sum_{m=0}^n \sum_{k \in I} \epsilon_{ijk} t_{ik}^{(m)} \otimes t_{kj}^{(n-m)}. 
\end{eqnarray*}
Here $\epsilon_{ijk} := (-1)^{|E_{ik}||E_{kj}|}$. The antipode $\Sm: U_q(\Gaff) \longrightarrow U_q(\Gaff)$ is determined by
\begin{eqnarray*}  
&&(\Sm \otimes \mathrm{Id})(S(z)) = S(z)^{-1}, \quad (\Sm \otimes \mathrm{Id})(T(z)) = T(z)^{-1}.    
\end{eqnarray*}
$S(z)^{-1}$ and $T(z)^{-1}$ are well-defined owing to Definition \ref{def: quantum affine superalgebras} (R3). Notice that $Y_q(\Glie)$ and $U_q(\Glie)$ are sub-Hopf-superalgebras of $U_q(\Gaff)$.

We shall need $U_{q^{-1}}(\Gaff)$, whose RTT generators are denoted by $\overline{s}_{ij}^{(n)}, \overline{t}_{ij}^{(n)}$. 

Recall the following are isomorphisms of Hopf superalgebras ($a \in \BC^{\times}$):
\begin{eqnarray}    
& \Phi_a: U_q(\Gaff) \longrightarrow U_q(\Gaff),\quad  & s_{ij}^{(n)} \mapsto a^n s_{ij}^{(n)},\ t_{ij}^{(n)} \mapsto a^{-n} t_{ij}^{(n)},   \label{iso:Z-grading} \\
& \Psi: U_q(\Gaff) \longrightarrow U_q(\Gaff)^{\mathrm{cop}},\quad & s_{ij}^{(n)} \mapsto \varepsilon_{ji}t_{ji}^{(n)},\quad t_{ij}^{(n)} \mapsto \varepsilon_{ji}s_{ji}^{(n)},   \label{iso:transposition}   \\
& h: U_{q^{-1}}(\Gaff) \longrightarrow U_q(\Gaff)^{\mathrm{cop}},\quad &\overline{S}(z) \mapsto S(z)^{-1}, \quad \overline{T}(z) \mapsto T(z)^{-1}. \label{iso:involution}
\end{eqnarray}
Here $\varepsilon_{ij} := (-1)^{|\epsilon_i|+|\epsilon_i||\epsilon_j|}$ and $A^{\mathrm{cop}}$ of a Hopf superalgebra $A$ takes the same underlying superalgebra but the twisted coproduct $\Delta^{\mathrm{cop}} := c_{A,A} \Delta$, with $c_{A,A}: x \otimes y \mapsto (-1)^{|x||y|}y \otimes x$ the graded permutation, and antipode $\Sm^{-1}$. There are superalgebra morphisms for $p(z) \in \BC[[z]]^{\times}, p_1(z) \in \BC[[z^{-1}]]^{\times}$ with $p(0)p_1(\infty) = 1$:
\begin{eqnarray}
&& \ev_a^+: U_q(\Gaff) \longrightarrow U_q(\Glie), \ s_{ij}(z) \mapsto \frac{s_{ij}^{(0)} - za t_{ij}^{(0)}}{1-za},\ t_{ij}(z) \mapsto \frac{t_{ij}^{(0)} - z^{-1}a^{-1}s_{ij}^{(0)}}{1-z^{-1}a^{-1}},  \label{homo: evaluation}  \\
&& \phi_{[p,p_1]}: U_q(\Gaff) \longrightarrow U_q(\Gaff),\ s_{ij}(z) \mapsto p(z)s_{ij}(z),\ t_{ij}(z) \mapsto p_1(z) t_{ij}(z).  \label{homo:power series}
\end{eqnarray}
$\Phi_a, h, \ev_a^+, \phi_{[p,p_1]}$ restrict to $Y_q(\Glie)$ or $Y_q(\Glie')$, denoted by $\Phi_a,h,\ev_a^+, \phi_p$. Let 
$ \overline{\ev}_a^+: U_{q^{-1}}(\Gaff) \longrightarrow U_{q^{-1}}(\Glie)$ be the corresponding morphisms when replacing $q$ by $q^{-1}$. This gives rise to (notice that $h(U_{q^{-1}}(\Glie)) = U_q(\Glie)$):
\begin{eqnarray} 
&&\ev_a^-: U_{q}(\Gaff) \longrightarrow U_q(\Glie),\quad \ev_a^- = h \circ \overline{\ev}^+_a \circ h^{-1}. \label{homo:evaluation dec}
\end{eqnarray}

$U_q(\Gaff)$ is $\BQ$-graded: $x \in U_q(\Gaff)$ is of weight $\lambda \in \BQ$ if $s_{ii}^{(0)} x  = q^{(\lambda,\epsilon_i)} x s_{ii}^{(0)}$ for all $i \in I$. For example $s_{ij}^{(n)}$ and $t_{ij}^{(n)}$ are of weight $\epsilon_i-\epsilon_j$ \cite[(3.14)]{Z3}. Let $U_q(\Gaff)_{\lambda}$ be the weight space of weight $\lambda$. The $\BQ$-grading restricts to $Y_q(\Glie)$ and $U_q(\Glie)$.

We recall the {\it Drinfeld second realization} of $U_q(\Gaff)$ from \cite[Section 3.1.4]{Z3}. Write
\begin{equation*} 
\begin{cases}
S(z) = (\sum\limits_{i<j} e_{ij}^+(z) \otimes E_{ij} + 1) (\sum\limits_l K_l^+(z) \otimes E_{ll} )(\sum\limits_{i<j} f_{ji}^+(z) \otimes E_{ji} + 1), \\
T(z) = (\sum\limits_{i<j} e_{ij}^-(z) \otimes E_{ij} + 1) (\sum\limits_l K_l^-(z) \otimes E_{ll} )(\sum\limits_{i<j} f_{ji}^-(z) \otimes E_{ji} + 1),
\end{cases}  
\end{equation*}
as invertible power series in $z^{\pm 1}$ over $U_q(\Gaff)\otimes \End (\BV)$; the subscripts $i,j,l \in I$. Notice that $K_{\kappa}^+(z) = s_{\kappa\kappa}(z)$. For $i \in I_0,\ j \in I$ let us define $\tau_i, \theta_j$:
\begin{align}
&\tau_i := q^{M-N+1-i} \quad \mathrm{for}\ 1\leq i \leq M, \quad \tau_{M+l} := q^{l+1-N} \quad \mathrm{for}\ 1\leq l < N, \label{equ: spectral shifts} \\
&\theta_j := q^{2(M-N+1-j)} \quad \mathrm{for}\ 1\leq j \leq M,\quad \theta_{M+l} := q^{2(l-N)}\quad \mathrm{for}\ 1\leq l \leq N. \label{equ: theta}
\end{align} 
 The Drinfeld loop generators are defined by generating series: let $i \in I_0$,
\begin{align*}
x_i^+(z) &= \sum_{n\in \BZ} x_{i,n}^+ z^n := \frac{e_{i,i+1}^+(z\tau_i)-e_{i,i+1}^-(z\tau_i)}{q_i-q_i^{-1}} \in U_q(\Gaff)[[z,z^{-1}]], \\
x_i^-(z) &= \sum_{n\in \BZ} x_{i,n}^- z^n := \frac{f_{i+1,i}^-(z\tau_i)-f_{i+1,i}^+(z\tau_i)}{q_i^{-1}-q_i} \in U_q(\Gaff)[[z,z^{-1}]], \\
\phi_i^{\pm}(z) &= \sum_{n\geq 0} \phi_{i,\pm n}^{\pm} z^{\pm n} := K_i^{\pm}(z\tau_i)K_{i+1}^{\pm}(z\tau_i)^{-1} \in U_q(\Gaff)[[z^{\pm 1}]].
\end{align*}
From Gauss decomposition we have $K_l^+(z), \phi_i^+(z) \in Y_q(\Glie)[[z]]$ for $l \in I$ and $i \in I_0$.
\begin{rem} \label{rem: two Gauss decompositions}
In \cite[Section 3.1.4]{Z3} a different Gauss decomposition of $S(z),T(z)$ was considered ($f$ always ahead of $e$). If $\overline{X_i}^{\pm}(z), \overline{K}_l^{\pm}(z)$ with $i \in I_0,\ l \in I$ denote the Drinfeld generating series of $U_{q^{-1}}(\Gaff)$ in {\it loc. cit.}, then 
$$ h(\overline{K}_l^{\pm}(z)) = K_l^{\pm}(z)^{-1},\quad h(\overline{X}_i^{\pm}(z)) = \pm (q_i^{-1}-q_i) x_{i}^{\pm}(z\tau_i^{-1}). $$
Let us rewrite \cite[Theorem 3.5]{Z3} in terms of the $x_i^{\pm}(z),\phi_i^{\pm}(z), K_l^{\pm}(z)$. First, the coefficients of these series generate the whole algebra $U_q(\Gaff)$. Second, for $i,j \in I_0,\ l,l' \in I$ and $\eta,\eta' \in \{\pm\}$ we have: (recall $q_{ij} = q^{(\alpha_i,\alpha_j)}$)
\begin{gather*}
K_l^{\eta}(z) K_{l'}^{\eta'}(w) = K_{l'}^{\eta'}(w) K_l^{\eta}(z),\quad K_l^+(0) K_l^-(\infty) = 1, \\
K_{M+N}^{\eta}(z) x_i^{\pm}(w) = \left(\frac{zq^{-1}-wq}{z-w}\right)^{\pm \delta_{i+1,M+N}} x_i^{\pm}(w) K_{M+N}^{\eta}(z), \\
\phi_i^{\eta}(z) x_j^{\pm}(w) = \frac{z-wq_{ij}^{\pm 1}}{zq_{ij}^{\pm 1}-w} x_j^{\pm}(w) \phi_i^{\eta}(z), \\
[x_i^+(z),x_j^-(w)] =\delta_{ij} \frac{\phi_i^+(z) -  \phi_i^-(w)}{q_i-q_i^{-1}} \delta(\frac{z}{w}), \\
(zq_{ij}^{\pm 1}-w) x_i^{\pm}(z)x_j^{\pm}(w) = (z-wq_{ij}^{\pm 1}) x_j^{\pm}(w)x_i^{\pm}(z)\quad \mathrm{if}\ (i,j) \neq (M,M), \\
[x_i^{\pm}(z_1),[x_i^{\pm}(z_2), x_j^{\pm}(w)]_q]_{q^{-1}} + \{z_1 \leftrightarrow z_2  \} = 0 \quad \mathrm{if}\ (i \neq M,\ |j-i|=1), \\
x_M^{\pm}(z) x_M^{\pm}(w)= - x_M^{\pm}(w)x_M^{\pm}(z), \quad
x_i^{\pm}(z) x_j^{\pm}(w) = x_j^{\pm}(w)x_i^{\pm}(z) \quad \mathrm{if}\ |i-j| > 1,
\end{gather*}
together with the degree 4 oscillator relation when $M,N > 1$:
\begin{gather*}
[[[x_{M-1}^{\pm}(u),x_M^{\pm}(z_1)]_q, x_{M+1}^{\pm}(v)]_{q^{-1}}, x_M^{\pm}(z_2)] + \{z_1 \leftrightarrow z_2\} = 0.
\end{gather*}
Here $[x,y]_a := xy - a (-1)^{|x||y|} y x$ for $x,y \in U_q(\Gaff)$ and $a \in \BC$. These relations are coherent with the Drinfeld second realization of quantum affine algebras (e.g. \cite[Section 3.2]{H2}) and superalgebras \cite[Theorem 8.5.1]{Y}. For $i \in I_0\setminus \{M\}$, the subalgebra of $U_q(\Gaff)$ generated by $(x_{i,n}^{\pm},\phi_{i,n}^{\pm})_{n \in \BZ}$ is a quotient algebra of $U_{q_i}(\widehat{\mathfrak{sl}_2})$.
\end{rem}

Let $\BQ^+ := \oplus_{i\in I_0} \BZ_{\geq 0} \alpha_i \subset \BP$ and $\BQ^- := - \BQ^+$. By \cite[Proposition 3.6]{Z3}:
\begin{align}
 \Delta (K_i^{\pm}(z)) \in K_i^{\pm}(z) \otimes K_i^{\pm}(z) + & \sum_{0\neq \alpha \in \BQ^+} U_q(\Gaff)_{-\alpha} \otimes U_q(\Gaff)_{\alpha} [[z^{\pm 1}]], \label{coproduct: affine Cartan}  \\
  \Delta (x_i^+(z)) \in x_i^+(z) \otimes 1 + & \sum_{0\neq \alpha \in \BQ^+} U_q(\Gaff)_{\alpha_i-\alpha} \otimes U_q(\Gaff)_{\alpha}[[z,z^{-1}]], \label{coproduct: positive}  \\
  \Delta (x_i^-(z)) \in 1 \otimes x_i^-(z) + &\sum_{0\neq \alpha \in \BQ^+} U_q(\Gaff)_{-\alpha} \otimes U_q(\Gaff)_{\alpha-\alpha_i}[[z,z^{-1}]].
\end{align}
The coproduct shares the same triangular property as \cite[Lemma 1]{FR}. 

\subsection{Category $\BGG$} \label{subsec: O}
 We first recall the notion of weights from \cite[Section 6]{Z5}. Define 
$$\CP := (\BC^{\times})^I \times \super,\quad \lCP := (\BC[[z]]^{\times})^I \times \super.  $$
The multiplicative group structure on $\BC^{\times}, \BC[[z]]^{\times}$ and the {\it additive} group structure on the ring $\super$ make $\CP, \lCP$ into multiplicative abelian groups. $\CP$ is naturally a subgroup of $\lCP$, and $\BC[[z]]^{\times} \longrightarrow \BC^{\times}, f(z) \mapsto f(0)$ induces a projection $\varpi: \lCP \longrightarrow \CP$. There is an injective homomorphism of abelian groups (see also \cite[Section 3.1]{Jimbo2})
\begin{equation}  \label{equ: embed weight lattice}
q^?: \BP \longrightarrow \CP,\quad \lambda  \mapsto q^{\lambda} := ((q^{(\epsilon_i,\lambda)})_{i\in I}; |\lambda|).
\end{equation}
Elements of $\lCP$ will usually be denoted by $\Bf,\Bg,\cdots$, or $\Bf(z), \Bg(z), \cdots$ when their dependence on $z$ is needed. For instance, if $\Bf = ((f_i(z))_{i\in I};s) \in \lCP$, then for $a \in \BC^{\times}$ we have $\Bf(za) = ((f_i(za))_{i\in I};s) \in \lCP$. We view $h(z) \in \BC[[z]]^{\times}$ as the element $(h(z),\cdots,h(z);\even) \in \lCP$, which makes $\BC[[z]]^{\times}$ a subgroup of $\lCP$.

Let $V$ be a $Y_q(\Glie)$-module. For $p = ((p_i)_{i\in I};s) \in \CP$, define 
$$ V_p := \{ v \in V_s \ |\ s_{ii}^{(0)} v = p_i v \ \ \mathrm{for}\ i \in I \}. $$
If $V_p \neq 0$, then $p$ is called a {\it weight} of $V$, and $V_p$ the weight space of weight $p$. Let $\wt(V)$ denote the set of weights of $V$. We have $s_{ij}^{(n)} V_p \subseteq V_{q^{\epsilon_i-\epsilon_j}p}$ for $p \in \wt(V)$. Similarly, for $\Bf = ((f_i(z))_{i\in I};s) \in \lCP$ define
$$ V_{\Bf} := \{ v \in V_s\ |\ \exists d \in \BZ_{>0} \ \mathrm{such\ that}\ \ (K_{i}^+(z)-f_i(z))^d v = 0\ \mathrm{for}\  i \in I \}. $$
If $V_{\Bf} \neq 0$, then $\Bf$ is an $\ell$-weight of $V$, and $V_{\Bf}$ the $\ell$-weight space of $\ell$-weight $\Bf$. Let $\lwt(V)$ be the set of $\ell$-weights of $V$.

{\it One should be aware that in \cite[Section 6]{Z5} the definition of $\ell$-weight spaces involves different Drinfeld generators. Nevertheless making use of Remark \ref{rem: two Gauss decompositions} and the involution $h$, we can translate all the the results concerning $Y_{q^{-1}}(\Glie)$- and $U_{q^{-1}}(\Glie)$-modules in \cite{Z5}, so as to obtain parallel results on $Y_q(\Glie)$- and $U_q(\Gaff)$-modules. }
\begin{example} \label{example one-dim}
To $\Bf = h(z)p \in \lCP$ with $h(z) \in 1 + z\BC[[z]]$ and $p = ((p_i)_{i\in I};s) \in \CP$ is associated a representation of $Y_q(\Glie)$ on the one-dimensional vector superspace $\BC_s := \BC \mathbf{1}$ of parity $s = |\mathbf{1}|$, defined by $ s_{ij}(z) \mathbf{1} = \delta_{ij} h(z) p_i \mathbf{1}$. Let $\BC_{\Bf}$ denote this $Y_q(\Glie)$-module. We have $\{\Bf\} = \lwt(\BC_{\Bf})$ and $\{p\} = \wt(\BC_{\Bf})$.
\end{example}
\begin{defi} \label{def: category BGG}  \cite[Definition 6.3]{Z5}
A $Y_q(\Glie)$-module $V$ is in category $\BGG$ if:
\begin{itemize}
\item[(i)] $V$ has a weight space decomposition $V = \oplus_{p\in \CP} V_p$;
\item[(ii)] $\dim V_p < \infty$ for all $p \in \CP$;
\item[(iii)] there exist $\mu_1,\mu_2,\cdots,\mu_d \in \CP$ such that $\mathrm{wt}(V) \subseteq \cup_{j=1}^d (q^{\BQ^-}\mu_j )$.  
\end{itemize}
\end{defi}

Let $V$ be a $Y_q(\Glie)$-module in category $\BGG$. A non-zero $\omega \in V$ is called a {\it highest $\ell$-weight} vector if it belongs to $V_{\Bf}$ for certain $\Bf = ((f_i(z))_{i\in I};s) \in \lCP$ and it is annihilated by the $s_{ij}(z)$ for $i<j$. Necessarily $K_i^+(z) \omega = f_i(z) \omega$. Call $V$ a highest $\ell$-weight module if it is generated as a $Y_q(\Glie)$-module by a highest $\ell$-weight vector $\omega$, in which case $\omega$ is unique up to scalar multiple and its $\ell$-weight is called the highest $\ell$-weight of $V$. Lowest $\ell$-weight vector/module is defined similarly by replacing the condition $i < j$ with $i > j$.

In Example \ref{example one-dim} the vector $\mathbf{1} \in \BC_{\Bf}$ is both of highest and of lowest $\ell$-weight.

\noindent {\bf Attention!} If $\omega$ is a lowest $\ell$-weight vector of $\ell$-weight $\Bf = ((f_i(z))_{i\in I};s)$, then we have $s_{ii}(z) \omega = f_i(z) \omega$ for $i \in I$; see also \cite[Section 6]{Z5}. This is not necessarily true if ``lowest" is replaced by ``highest".

Let $\BR$ be the subset of $\lCP$ consisting of the $\Bf = ((f_i(z))_{i\in I};s)$ such that $\frac{f_i(z)}{f_{i+1}(z)}$ is the Taylor expansion at $z=0$ of a rational function for $i\in I_0$. 

\begin{lem} \label{lem: simple O} \cite[Lemma 6.8 \& Proposition 6.10]{Z5} Let $\Bf = ((f_i(z))_{i\in I};s) \in \BR$.
\begin{itemize}
\item[(1)]In category $\BGG$ there exists a  unique irreducible highest $\ell$-weight module $L(\Bf)$ of highest $\ell$-weight $\Bf$ up to isomorphism. The $L(\Bg)$ for $\Bg \in \BR$ form the set of irreducible objects (two-by-two non-isomorphic) of category $\BGG$.
\item[(2)] $\dim L(\Bf) = 1$ if and only if $\frac{f_i(z)}{f_{i+1}(z)} \in \BC^{\times}$ for $i \in I_0$, i.e. $\Bf \in \BC[[z]]^{\times} \CP $.
\item[(3)] $\dim L(\Bf) < \infty$ if and only if for $i \in I_0 \setminus \{M\}$ there exist $P_i(z) \in 1+z\BC[z]$ and $c_i \in \BC^{\times}$ such that $\frac{f_i(z)}{f_{i+1}(z)} = c_i\frac{P_i(zq_i^{-1})}{P_i(zq_i)}$.
\item[(4)] $L(\Bf)$ can be extended to a $U_q(\Gaff)$-module if and only if $\frac{f_i(z)}{f_{i+1}(z)}$ is a product of the $c \frac{1-zac^{-2}}{1-za}$ with $a,c \in \BC^{\times}$ for $i \in I_0$. 
\end{itemize}
\end{lem}

Based on (4), let $\BR_U$ be the subset of $\BR$ consisting of $\Bf = ((f_i(z))_{i\in I}; s)$ such that for $i\in I$, the rational function $f_i(z)$ is a product of the $c\frac{1-zac^{-2}}{1-za}$ with $a,c \in \BC^{\times}$. For $\Bf \in \BR_U$, the $Y_q(\Glie)$-module $L(\Bf)$ is extended {\it uniquely} to a $U_q(\Gaff)$-module by
$$ K_i^+(z) \omega = f_i(z) \omega = K_i^-(z) \omega \quad \mathrm{for}\ i \in I. $$
Here $\omega$ is a highest $\ell$-weight vector, and in the second identity one views $f_i(z) \in \BC[[z^{-1}]]$ by taking the its Taylor expansion of at $z = \infty$. We continue to let $L(\Bf)$ denote the irreducible $U_q(\Gaff)$-module thus obtained for $\Bf \in \BR_U$.

\begin{example}    \label{ss: l-weights}
For $i \in I_0$ and $a \in \BC^{\times}$  define the {\it prefundamental $\ell$-weight} $\Psi_{i,a} \in \BR$, the {\it fundamental weight} $\varpi_i \in \BP$, and $[a]_i \in \BR$ by:
\begin{gather*}  
\begin{tabular}{|c|c|c|} 
  \hline
  &$i \leq M$ & $i > M$  \\
  \hline
 $\Psi_{i,a} $ & $(\underbrace{h(z),\cdots, h(z)}_{i}, \underbrace{1,\cdots, 1}_{\kappa-i};\even )$ & $(\underbrace{1,\cdots,1}_i, \underbrace{h(z)^{-1}, \cdots, h(z)^{-1}}_{\kappa-i}; \even)$ \\
  \hline
  $[a]_i$ & $(\underbrace{a,\cdots, a}_{i}, \underbrace{1,\cdots, 1}_{\kappa-i};\even )$ & $(\underbrace{1,\cdots,1}_i, \underbrace{a^{-1}, \cdots, a^{-1}}_{\kappa-i}; \even)$ \\
  \hline
  $\varpi_i$ &$\epsilon_1+\epsilon_2+\cdots + \epsilon_i$ &  $-\epsilon_{i+1} - \epsilon_{i+2} - \cdots - \epsilon_{\kappa}$ \\
  \hline
\end{tabular} 
\end{gather*}
where $h(z) = 1-za\tau_i^{-1}$. For $i,j \in I_0$ let us write $i \sim j$ if $|i-j| = 1$. Define
$$ a_{ij} := a^{(\alpha_i,\alpha_j)},\quad \hat{q}_i = q_i \quad \mathrm{if}\ i \neq M,\quad  \hat{q}_M = q^{-1}. $$
Let us introduce the following elements of $\BR$ for $c \in \BC^{\times}$ and $m \in \BZ_{>0}$:
\begin{gather*}
\Bn_{i,a}^+ := \frac{\Psi_{i,aq_i^{-2}}}{\Psi_{i,a}}  \prod_{j\in I_0: j\sim i} \Psi_{j,aq_{ij}^{-1}},\quad \Bn_{i,a}^- := \frac{\Psi_{i,a}}{ \Psi_{i,a\hat{q}_i^2}} \prod_{j\in I_0: j \sim i} \Psi_{j,aq_{ij}}^{-1}, \\
\aBw_{c,a}^{(i)} := [c]_i \frac{\Psi_{i,ac^{-2}}}{\Psi_{i,a}},\quad \varpi_{m,a}^{(i)} := q^{m\varpi_i} \frac{\Psi_{i,aq_i^{1-2m}}}{\Psi_{i,aq_i}},\quad Y_{i,a} := q^{\varpi_i} \frac{\Psi_{i,aq_i^{-1}}}{\Psi_{i,aq_i}}, \\
\Bn_{c,a}^{(i)} :=   \aBw_{\hat{q}_i,a\hat{q}_i^2}^{(i)} \prod_{j\in I_0: j \sim i} \aBw_{c_{ij}^{-1},aq_{ij}}^{(j)},\quad \Bm_{c,a}^{(i)} := \aBw_{q_i,a}^{(i)} \prod_{j\in I_0: j \sim i}\aBw_{c_{ij}^{-1},aq_{ij}^{-1}c_{ij}^{-2}}^{(j)}, \\
A_{i,a} := (\underbrace{1,\cdots,1}_{i-1},  q_i \frac{1-za\tau_iq^{-1}\theta_i^{-1} q_i^{-1}}{1-za\tau_iq^{-1}\theta_i^{-1} q_i}, q_{i+1}^{-1} \frac{1-za\tau_iq^{-1}\theta_i^{-1} q_iq_{i+1}^2}{1-za\tau_iq^{-1}\theta_i^{-1} q_i}  \underbrace{1,\cdots,1}_{\kappa - i-1} ;|\alpha_i|). 
\end{gather*}
The irreducible $Y_q(\Glie)$-modules $L_{i,a}^{\pm} := L(\Psi_{i,a}^{\pm 1})$ are called positive/negative {\it prefundamental modules}. If $\omega$ is a highest $\ell$-weight vector of $L_{i,a}^{+}$, then 
$$\phi_j^+(z) \omega = \omega \quad \mathrm{for}\ j \neq i,\quad \phi_i^+(z) \omega = (1-za)\omega. $$
So $\Psi_{i,a}$  is a super analog of \cite[(3.16)]{HJ}. Define the irreducible $Y_q(\Glie)$-modules:
$$N_{i,a}^{\pm} := L(\Bn_{i,a}^{\pm}),\quad M_{c,a}^{(i)} := L(\Bm_{c,a}^{(i)}),\quad W_{m,a}^{(i)} := L(\varpi_{m,a}^{(i)}). $$
Call $W_{m,a}^{(i)}$ a {\it Kirillov--Reshetikhin module} (KR module). By Lemma \ref{lem: simple O}, the $M, W$ are $U_q(\Gaff)$-modules with $W$ finite-dimensional. (In Sections \ref{sec: asym}--\ref{sec: proof T}  $N_{m,a}^{(i)}$ will denote the irreducible module $L(\Bm_{q^m,a}^{(i)})$ for $m \in \BZ_{>0}$, so here we do not use $N_{c,a}^{(i)}$.) 
\end{example}

\begin{rem}
Later in Sections \ref{sec: cyclicity}--\ref{sec: asym} we work with $U_q(\Gaff)$-modules in category $\BGG$. Such a module $V$ is called a highest $\ell$-weight $U_q(\Gaff)$-module in 
\cite[Section 1.2]{Z4}  if there exists a non-zero $\super$-homogeneous vector $\omega$ such that $V = U_q(\Gaff) \omega$ and 
$$ s_{ij}^{(n)} \omega = t_{ij}^{(n)} \omega = 0, \quad s_{ll}^{(n)} \omega \in \BC \omega \ni  t_{ll}^{(n)} \omega \quad \mathrm{for}\ i < j. $$
Indeed $V$ is of highest $\ell$-weight as a $U_q(\Gaff)$-module if and only it is of highest $\ell$-weight as a $Y_q(\Glie)$-module. (The ``if" part comes from weight grading, while the ``only if" part from the Drinfeld relations in Remark \ref{rem: two Gauss decompositions}.) It also follows that $V$ is an irreducible $U_q(\Gaff)$-module if and only if it is an irreducible $Y_q(\Glie)$-module, as in \cite[Proposition 3.5]{HJ}. Therefore when we say $V$ is of highest $\ell$-weight or irreducible, we make no reference to $Y_q(\Glie)$ or $U_q(\Gaff)$.
\end{rem}

As in \cite[Section 3.2]{HL}, let $\lCE$ be the set of formal sums $\sum_{\Bf \in \lCP} c_{\Bf} \Bf$ with integer coefficients $c_{\Bf} \in \BZ$ such that $\oplus_{\Bf\in \lCP} \BC_{\Bf}^{\oplus |c_{\Bf}|}$ is an object of category $\BGG$. It is a ring: addition is the usual one of formal sums; multiplication is induced by that of $\lCP$. (One views $\lCE$ as a completion of the group ring $\BZ[\lCP]$.)

For $V$ an object of category $\BGG$, its weight space decomposition can be refined to an $\ell$-weight decomposition because of condition (ii) in Definition \ref{def: category BGG}.  Following \cite{FR} we define its $q$-character and classical character
\begin{equation}  \label{def: q-char}
 \chi_q(V) = \sum_{\Bf \in \mathrm{wt}_{\ell}(V)} \dim (V_{\Bf}) \Bf ,\quad \chi(V) = \sum_{p \in \mathrm{wt}(V)} \dim (V_p) p \in \lCE. 
\end{equation}
In Example \ref{example one-dim} we have $\chi_q(\BC_{\Bf}) = \Bf$ and $\chi(\BC_{\Bf}) = \varpi(\Bf)$.

We shall need the completed Grothendieck group $K_0(\BGG)$. Its definition is the same as that in \cite[Section 3.2]{HL}: elements are formal sums $\sum_{\Bf \in \BR} c_{\Bf}[L(\Bf)]$ with integer coefficients $c_{\Bf} \in \BZ$ such that $\oplus_{\Bf\in \BR} L(\Bf)^{\oplus |c_{\Bf}|}$ is in category $\BGG$; addition is the usual one of formal sums. For $\Bf \in \BR$ and $V$ in category $\BGG$, the multiplicity of the irreducible module $L(\Bf)$ in $V$ is well-defined due to Definition \ref{def: category BGG}, as in the case of Kac--Moody algebras \cite[Section 9.6]{Kac}; it is denoted by $m_{L(\Bf),V} \in \BZ_{\geq 0}$.  Necessarily $[V] := \sum_{\Bf \in \BR} m_{L(\Bf),V} [L(\Bf)] \in K_0(\BGG)$. In the case $V = L(\Bf)$ the right-hand side is simply $[L(\Bf)]$ because $m_{L(\Bg),L(\Bf)} = \delta_{\Bg\Bf}$ for $\Bg \in \BR$.

 Make $K_0(\BGG)$ into a ring by $[V] [W] := [V \otimes W]$. Equation \eqref{def: q-char} extends uniquely to morphisms of additive groups $\chi_q: K_0(\BGG) \longrightarrow \lCE$ and $\chi: K_0(\BGG) \longrightarrow \lCE$, called $q$-character map and character map respectively. As in \cite[Theorem 3]{FR}, we have
\begin{prop} \label{prop: q-char ring} \cite[Corollary 6.9]{Z5}
The $q$-character map $\chi_q$ is an injective morphism of rings. Consequently the ring $K_0(\BGG)$ is commutative.
\end{prop}
The tensor product $L(\Bf) \otimes L(\Bg)$ contains an irreducible sub-quotient $L(\Bf\Bg)$ for $\Bf, \Bg \in \BR$. Let us define the {\it normalized $q$-character} $\nqc(L(\Bf)) := \Bf^{-1}\qc(L(\Bf))$.

For $V,W$ in category $\BGG$, write $V \simeq W$ if there is a one-dimensional module $D$ in category $\BGG$ such that $V \cong W \otimes D$ as $Y_q(\Glie)$-modules.  By Lemma \ref{lem: simple O} (2) and Proposition \ref{prop: q-char ring} we have $L(\Bf) \simeq L(\Bg)$ if and only if $ \Bg^{-1}\Bf \in \BC[[z]]^{\times} \CP$, in which case the normalized $q$-characters of $L(\Bf)$ and $L(\Bg)$ are identical and we write $\Bf \equiv \Bg$.

As an example, for the {\it generalized simple root} $A_{i,a} \in \BR_U$ we have
\begin{equation}  \label{equ: A Psi}
A_{i,a} \equiv \frac{\Psi_{i,aq_i^{-2}}}{\Psi_{i,a\hat{q}_i^2}} \prod_{j\in I_0: j \sim i} \frac{\Psi_{j,aq_{ij}^{-1}}}{\Psi_{j,aq_{ij}}}.
\end{equation}

\subsection{Category $\BGG'$}
As in \cite[Section 1]{Z4}, let $\mathfrak{gl}(N|M) =: \Glie'$ be another Lie superalgebra, which is not to be confused with the derived algebra of $\Glie$. Define the Hopf superalgebras $U_q(\Gafft),Y_q(\Glie'), U_q(\Glie')$ in the same way as for $U_q(\Gaff),Y_q(\Glie),U_q(\Glie)$ in Section \ref{subsec: quantum superalgebra}, except that $M,N$ are interchanged. We start from the same weight/root lattices $\BP,\BQ$ and $\CP,\lCP$ but with different parity map $|?|': \BP \longrightarrow \super$: 
$$ |\epsilon_1|' = |\epsilon_2|' = \cdots = |\epsilon_N|' = \even,\quad |\epsilon_{N+1}|' = |\epsilon_{N+2}|' = \cdots = |\epsilon_{N+M}|' = \odd, $$
bilinear form $(\epsilon_i,\epsilon_j)' = \delta_{ij} (-1)^{|\epsilon_i|'}$, and embedding $q'^{\lambda} := ((q^{(\lambda,\epsilon_i)'})_{i\in I}; |\lambda|')$ of $\BP$ in $\CP$. One defines category $\BGG'$ of $Y_q(\Glie')$-modules as in Section \ref{subsec: O}. Let us summarize the modifications of notations related to $\Glie'$ to be used later on:
\begin{equation}  \label{tab: comparison}
\begin{tabular}{|c|c|c|}
  \hline
 $\Glie,\ U_q(\Glie),\ Y_q(\Glie),\ U_q(\Gaff)$ & $\Glie',\ U_q(\Glie'),\ Y_q(\Glie'),\ U_q(\Gafft)$ & algebras  \\
   \hline 
   $s_{ij}^{(n)},\ t_{ij}^{(n)},\ q_i,\ q_{ij},\ \tau_i,\ \theta_j$ & $s_{ij}'^{(n)},\ t_{ij}'^{(n)},\ q_i',\ q_{ij}',\ \tau_i',\ \theta_j'$ & RTT \\
  \hline
  $x_i^{\pm}(z),\ K_i^{\pm}(z),\ \phi_{i}^{\pm}(z)$ & $x_i'^{\pm}(z),\ K_i'^{\pm}(z),\ \phi_{i}'^{\pm}(z)$ & currents \\
  \hline
  $\BGG,\ L(\Bf),\  L_{i,a}^{\pm},\ N_{i,a}^{\pm},\  W_{m,a}^{(i)}$ & $\BGG',\ L'(\Bf),\ L_{i,a}'^{\pm},\ N_{i,a}'^{\pm},\ W_{m,a}'^{(i)} $ & categories \\
  \hline
\end{tabular}
\end{equation}
In case $M=N$ one can simply remove all the primes in the table.

 For $i,j \in I$, set $\widehat{i} := \kappa+1-i$ and $\varepsilon_{ij}' := (-1)^{|\epsilon_i|'+|\epsilon_i|'|\epsilon_j|'}$. Then
 \begin{equation}
   \SF: U_q(\Gafft) \longrightarrow U_q(\Gaff)^{\mathrm{cop}},\quad s_{ij}'^{(n)} \mapsto \varepsilon_{ji}' s_{\widehat{j}\widehat{i}}^{(n)},\quad t_{ij}'^{(n)} \mapsto \varepsilon_{ji}' t_{\widehat{j}\widehat{i}}^{(n)}. \label{iso:symmetry permut} 
 \end{equation}
defines a Hopf superalgebra isomorphism. Let $\overline{\SF}: U_{q^{-1}}(\Gafft) \longrightarrow U_{q^{-1}}(\Gaff)^{\mathrm{cop}}$ and $h': U_{q^{-1}}(\Gafft) \longrightarrow U_q(\Gafft)^{\mathrm{cop}}$ be analogs of Equations \eqref{iso:symmetry permut} and \eqref{iso:involution}. They induce
\begin{equation}
\SG: U_q(\Gafft) \longrightarrow U_q(\Gaff)^{\mathrm{cop}},\quad \SG := h \circ \overline{\SF} \circ h'^{-1}\label{iso:dual permit} 
\end{equation}
a Hopf superalgebra isomorphism which restricts to $\SG: Y_q(\Glie') \longrightarrow Y_q(\Glie)$.

\begin{lem}\label{lem: duality by permutation}
The pullback by $\SG$ is an anti-equivalence of monoidal categories $\SG^*: \BGG \longrightarrow \BGG'$. If $\Bf = (f_1(z), f_2(z),\cdots, f_{\kappa}(z);s) \in \BR$, then as $Y_q(\Glie')$-modules
$$  \SG^*( L(\Bf)) \cong L'(f_{\kappa}(z), f_{\kappa-1}(z),\cdots, f_1(z);s).  $$
In particular, $\SG^*(L_{i,a}^{\pm}) \simeq L_{M+N-i,aq^{N-M}}'^{\mp}$ for $1\leq i < M+N$.
\end{lem}
\begin{proof}
Let $V$ be a $Y_q(\Glie)$-module in category $\BGG$. If $p \in \CP$, then $ V_p = (\SG^*V)_{p'}$ where 
$ p' = ((p_{\widehat{i}})_{i\in I}; s)$, and so $V_{q^{n\alpha_i}p} = (\SG^*V)_{q'^{n\alpha_{\kappa-i}}p'}$ for $i \in I_0$ and $n \in \BZ$. This implies that $\SG^*V$ is in category $\BGG'$. The first statement is now clear.

Let $V = L(\Bf)$ and let $\omega \in V$ be a highest $\ell$-weight vector. In $h^* V$ we have 
$$ \overline{K}_l^+(z) h^*\omega = f_l(z)^{-1} h^*\omega,\quad  \overline{s}_{ij}(z) h^* \omega = 0 \quad \mathrm{for}\ i,j,l\in I\ \mathrm{with}\ i < j. $$
From the Gauss decomposition of $h^{-1}(S(z))$ we get $\overline{s}_{ll}(z) h^*\omega = \overline{K}_l^+(z) h^*\omega$. Similar identities hold when replacing $h^*\omega$ by $\overline{\SF}^*h^*\omega$. This implies:
\begin{align*}
\overline{K}_i'^+(z) \overline{\SF}^* h^* \omega &= \overline{s}_{ii}'(z) \overline{\SF}^* h^* \omega = \overline{\SF}^* \left(\overline{s}_{\widehat{i},\widehat{i}}(z) h^* \omega\right) \\
& = \overline{\SF}^* \left(\overline{K}_{\widehat{i}}^+(z) h^* \omega\right) = f_{\widehat{i}}(z)^{-1} \overline{\SF}^* h^* \omega, \\
K_i'^+(z) \SG^* \omega &=  K_i'^+(z)(h'^{-1})^* \overline{\SF}^* h^* \omega  = (h'^{-1})^* \left(\overline{K}_i'^+(z)^{-1}\overline{\SF}^* h^* \omega \right) \\
& = f_{\widehat{i}}(z) (h'^{-1})^* \overline{\SF}^* h^* \omega  = f_{\widehat{i}}(z) \SG^* \omega,
\end{align*}
leading to the second statement; here the $\overline{s}_{ii}'(z),\ \overline{K}_i'^+(z)$ denote the RTT generators and Drinfeld generators of $U_{q^{-1}}(\Gafft)$ arising from \cite{Z3}; see Remark \ref{rem: two Gauss decompositions}. The last statement is a comparison of highest $\ell$-weights based on $\tau_{M+N-i}' = \tau_i q^{N-M}$.
\end{proof}
$\SG^*$ can be viewed as a categorification of the duality function of Grothendieck rings in \cite[Theorem 5.17]{HL}. We shall make extensive use of it: to change the signature of the $L_{i,a}^{\pm}$; to pass from Dynkin nodes $i \leq M$ to $i \geq M$.
\section{Tableau-sum formulas of $q$-characters} \label{sec: q-char}
We compute $\chi_q(L(\Bm))$ for $\Bm \in \BR_U$ coming from Young diagrams.

\begin{defi} \label{def: Young} \cite[Section 4.2]{BKK} 
 $\mathcal{P}$ is the set of $\lambda = \sum_i \lambda_i \epsilon_i \in \BP$ such that:
\begin{itemize}
\item we have $\lambda_1 \geq \lambda_2 \geq  \cdots \geq \lambda_M \geq 0$ and $\lambda_{M+1} \geq \lambda_{M+2} \geq \cdots \geq \lambda_{\kappa} \geq 0$;
\item if $\lambda_{M+j} > 0$ for some $1\leq j \leq N$, then $\lambda_M \geq j$.
\end{itemize}
To $\lambda \in \mathcal{P}$ we attach a subset $Y_+^{\lambda}$ of $\BZ_{>0}^2$ consisting of $(k,l)$ such that: $l \leq \lambda_k$ for $1\leq k \leq M$; if $k > M$ then $l \leq N$ and $k \leq M + \lambda_{M+l}$. Let $\CB_+(\lambda)$ be the set of functions $T: Y_+^{\lambda} \longrightarrow I$ such that: 
\begin{itemize}
\item $T(k,l) \leq T(k',l')$ if $k \leq k',\ l \leq l'$ and $(k,l), (k',l') \in Y_+^{\lambda}$;
\item $T(k,l) < T(k+1,l)$ if $(k,l), (k+1, l) \in Y_+^{\lambda}$ and $T(k,l) \leq M$;
\item $T(k,l) < T(k,l+1)$ if $(k,l), (k,l+1) \in Y_+^{\lambda}$ and $T(k,l) > M$.
\end{itemize}
Let $Y_-^{\lambda} = -Y_+^{\lambda} \subset \BZ_{<0}^2$ and define $\CB_-(\lambda)$ as the set of functions $Y_-^{\lambda} \longrightarrow I$ satisfying the above three conditions with $Y_+^{\lambda}$ replaced by $Y_-^{\lambda}$.
\end{defi}

We view $Y_{+}^{\lambda}, Y_-^{\lambda}$ as Young diagrams at the southeast and northwest positions respectively, so that $(k,l) \in Y_{\pm}^{\lambda}$ correspond to the box at row $\pm k$ and column $\pm l$. For example, take $\Glie = \mathfrak{gl}(2|2)$ and $\lambda = 4\epsilon_1+2\epsilon_2+ 2\epsilon_3+\epsilon_4 \in \mathcal{P}$:
$$ Y_+^{\lambda} = \young(~~~~,~~,~~,~),\quad Y_-^{\lambda} = \young(:::~,::~~,::~~,~~~~) $$

%Recall the definitions of $\tau_i,\theta_j$ for $i,j \in I$ and $i < \kappa$ in Equations \eqref{equ: spectral shifts}--\eqref{equ: theta}.
\begin{defi} \label{def: tableau}
Let $i \in I_0,\ j \in I$ and $a \in \BC^{\times}$. Define the $\ell$-weights in $\BR_U$:
\begin{gather*}
\boxed{j}_a := (\underbrace{1,\cdots,1}_{j-1},  q_j \frac{1-za\theta_j^{-1} q_j^{-1}}{1-za\theta_j^{-1} q_j}, \underbrace{1,\cdots,1}_{\kappa - j} ;|\epsilon_j|), 
\end{gather*}
Define the $\boxed{j}_a^*, \boxed{j}_a'$ inductively by $\boxed{1}_a^* := \boxed{1}_{a\theta_1}^{-1},\ \boxed{\kappa}_a' := \boxed{\kappa}_a^{-1}$ and 
$$\boxed{i+1}_a^* := \boxed{i}_a^* A_{i,a\tau_iq^{-1}},\quad  \boxed{i+1}_a' =: \boxed{i}_a' A_{i,a\tau_iq^{-1}}. $$
Call $a$ the {\it spectral parameter} of the boxes $\boxed{j}_a,\ \boxed{j}_a^*,\ \boxed{j}_a'$.
\end{defi}
One checks that $A_{i,a} = \boxed{i}_{a\tau_i q^{-1}}\boxed{i+1}_{a\tau_iq^{-1}}^{-1}$ using $\theta_{i+1} = \theta_i q_i^{-1}q_{i+1}^{-1}$.
\begin{example}
If $\Glie := \mathfrak{gl}(2|3)$, then $\tau_1 = q^{-1}$ and (compare with \cite[Section 5.4.1]{FR})
\begin{align*}
&\boxed{1}_a \xrightarrow{A_{1,aq^2}^{-1}} \boxed{2}_a \xrightarrow{A_{2,aq^3}^{-1}} \boxed{3}_a \xrightarrow{A_{3,aq^2}^{-1}} \boxed{4}_a \xrightarrow{A_{4,aq}^{-1}} \boxed{5}_a, \\
&\boxed{1}_a^*\xrightarrow{A_{1,aq^{-2}}} \boxed{2}_a^* \xrightarrow{A_{2,aq^{-3}}} \boxed{3}_a^* \xrightarrow{A_{3,aq^{-2}}} \boxed{4}_a^* \xrightarrow{A_{4,aq^{-1}}} \boxed{5}_a^*.
\end{align*} 
\end{example}

To $p = ((p_i)_{i\in I};s) \in \CP$ is associated a unique irreducible $U_q(\Glie)$-module $V_q(p)$, which is generated by a vector $v$ of parity $s$ subject to the following relations:
$$ s_{ii}^{(0)} v = p_i v,\quad s_{jk}^{(0)} v = 0 \quad \mathrm{for}\ i,j,k \in I\ \mathrm{with}\ j < k. $$
For $\lambda \in \BP$, set $V_q(\lambda) := V_q(q^{\lambda})$. (It was denoted by $V(\lambda)$ in \cite[Section 3.3]{BKK}.)

For $\lambda \in \mathcal{P}$, the $U_q(\Glie)$-module $V_q(\lambda)$ is finite-dimensional \cite[Section 3.3]{BKK}; its dual space $V_q^*(\lambda) := \mathrm{Hom}_{\BC}(V_q(\lambda),\BC)$ is equipped with a $U_q(\Glie)$-module structure:
\begin{displaymath}
\langle x \varphi, v \rangle := (-1)^{|\varphi||x|} \langle \varphi, \Sm (x) v \rangle \quad \mathrm{for}\ x \in U_q(\Glie),\ \varphi \in V_q^*(\lambda),\ v \in V_q(\lambda).  
\end{displaymath} 
\begin{theorem} \label{thm: q-char MA}
Let $a \in \BC^{\times}$ and $\lambda \in \mathcal{P}$. Let $V_q^{\pm}(\lambda;a),\ V_q^{\pm *}(\lambda;a)$ be the pullbacks of the $U_q(\Glie)$-modules $V_q(\lambda),\ V_q^*(\lambda)$ by $\ev_a^{\pm}$ respectively. Then we have 
\begin{align}
\chi_q\left(V_q^{+}(\lambda;a) \right) &= \sum_{T \in \mathcal{B}_{-}(\lambda)} \prod_{(i,j) \in Y_{-}^{\lambda}} \boxed{T(i,j)}_{aq^{2(j-i)+1}}, \label{for: iMA} \\
\chi_q\left(V_q^{+*}(\lambda;a)\right) &= \sum_{T \in \mathcal{B}_{-}(\lambda)} \prod_{(i,j)\in Y_{-}^{\lambda}} \boxed{T(i,j)}_{aq^{2(i-j)+1}}^*,  \label{for: dual iMA} \\
\chi_q\left(V_q^{-}(\lambda;a) \right) &= \sum_{T \in \mathcal{B}_{+}(\lambda)} \prod_{(i,j) \in Y_{+}^{\lambda}} \boxed{T(i,j)}_{aq^{2(j-i+M-N)+1}}, \label{for: dMA} \\
\chi_q\left(V_q^{-*}(\lambda;a)\right) &= \sum_{T \in \mathcal{B}_{+}(\lambda)} \prod_{(i,j)\in Y_{+}^{\lambda}} \boxed{T(i,j)}_{aq^{2(i-j)+1}}'. \label{for: dual dMA}
\end{align}
In particular, $V_q^{\pm}(\lambda;a)$ and $V_q^{\pm *}(\lambda;a)$ have multiplicity free $q$-characters.
\end{theorem}
\begin{rem} \label{rem: BKK}
Applying $\varpi: \lCP \longrightarrow \CP$ to Equation \eqref{for: dMA} recovers the character formula of $V_q(\lambda)$ in \cite[Theorem 5.1]{BKK}. 
% in \cite[Proposition 7.5]{Z5}, which proves Equation \eqref{for: dMA} by Remark \ref{rem: two Gauss decompositions}. (One replaces $\mathfrak{X}_{i,aq^m}$ in \cite[Definition 6.2]{Z5} with .)
\end{rem}
 We shall prove Equations \eqref{for: iMA}--\eqref{for: dual iMA}; the idea is similar to \cite[Lemma 4.7]{FM}. The proof of Equation \eqref{for: dMA}--\eqref{for: dual dMA} is parallel and will be omitted. 
 %\footnote{Equation \eqref{for: dMA} is Proposition 7.5 in a preliminary version of \cite{Z5}, arXiv:1410.0837v2. One identifies the $U_{q^{-1}}(\Gaff)$-modules $V(\lambda;a)$ and the $\ell$-weights $\mathfrak{X}_{i,aq^m}$ therein with $h^*(V_q^-(\lambda;a))$ and $\boxed{i}_{aq^{-m+2(M-N)}}$ respectively in the present situation. }

For $i \in I$, let $U_q^{\geq i}(\Gaff)$ (resp. $U_q^{\geq i}(\Glie)$) be the subalgebra of $U_q(\Gaff)$ generated by the $s_{jk}^{(n)}, t_{jk}^{(n)}$ (resp. for $n = 0$) with $j,k \geq i$. Define
\begin{equation} \label{for: Jucys-Murphy}
C_i(z) := \prod_{j\geq i} K_j^+(z \theta_j)^{(\epsilon_j,\epsilon_j)} \in Y_q(\Glie)[[z]]. 
\end{equation}
The coefficients of $C_i(z)$ are central elements of $U_q^{\geq i}(\Gaff)$; see \cite[Proposition 6.1]{Z5}. 
\begin{lem} \label{lem: JC dual action}
Let $i,l \in I$. The spectra of $C_i(z)$ on $\ell$-weight spaces of $\ell$-weights $\boxed{l}_a, \boxed{l}_a^*$ are $( q\frac{1-zaq^{-1}}{1-zaq})^{\delta_{i\leq l}}$ and $(q^{-1} \frac{1-zat_iq}{1-zat_iq^{-1}})^{\delta_{i\leq l}}$ respectively,
where $t_1 = \theta_1$ and $t_i = \tau_{i-1}^2 q^{-2}$ for $i > 1$. Moreover $\boxed{l}_a^* = \frac{(1-zaq^{-3})(1-zaq)}{(1-zaq^{-1})^2} \boxed{l}_a'$.
\end{lem}
\begin{proof}
The $\boxed{l}$-case is from Definition \ref{def: tableau}. In particular the $A_{j,b}$ for $j \neq i-1$ do not contribute to the spectra of $C_i(z)$.  The $\boxed{l}^*$-case is now clear from $\boxed{l}_a^* = \boxed{1}_{a\theta_1}^{-1}A_{1,a\tau_1q^{-1}} A_{2,a\tau_2q^{-1}} \cdots A_{l-1,a\tau_{l-1}q^{-1}}$. To compare $\boxed{l}^*$ with $\boxed{l}'$ one may assume $l = \kappa$ by Definition \ref{def: tableau}; the spectrum of $C_i(z)$ associated to the $\ell$-weight $\boxed{\kappa}_a'$ is $q^{-1}\frac{1-zaq}{1-zaq^{-1}}$, leading to the last identity.
\end{proof}

Let $S$ be $V_q^+(\lambda;a)$ or $V_q^{+*}(\lambda;a)$. If $\mu \in \BP$ and $v \in S$ are such that $s_{ii}^{(0)} v = q^{(\mu,\epsilon_i)} v$ for all $i \in I$, then $|v| = |\mu|$. To compute the $q$-character of $S$, it is enough to determine the action of the $C_i(z)$ since it in turn implies the parity.

Let $S_1$ be an irreducible sub-$U_q^{\geq i}(\Glie)$-module of $S$ and $0 \neq v_1 \in S_1, \mu \in \BP$ with
$$ t_{jk}^{(0)} v_1 = 0,\quad s_{ll}^{(0)} v_1= q^{(\mu,\epsilon_l)} v_1 \quad \mathrm{for}\ j,k,l \in I,\ j > k. $$
Call $\mu$ the lowest weight of $S_1$. By Schur Lemma and Gauss decomposition,
\begin{equation} \label{for: eigen JM}
C_i(z) v = \prod_{j\geq i} \left( \frac{q^{(\mu,\epsilon_j)}-za\theta_j q^{-(\mu,\epsilon_j)}}{1-za\theta_j} \right)^{(\epsilon_j,\epsilon_j)} v\quad \mathrm{for}\ v \in S_1.
\end{equation}
The strategy is to find  all such triples $(i, S_1, \mu)$.

 Following Table \eqref{tab: comparison} and Definition \ref{def: Young}, define for $\Glie'$ the similar objects 
 $$\mathcal{P}'\subset \BP,\quad Y_{\pm}'^{\lambda} \subset \BZ^2,\quad \CB_{\pm}'(\lambda),\  V_q'(\lambda),\ V_q'^*(\lambda)$$
  with $(M,N)$ replaced by $(N,M)$. The transpose of Young diagrams induces a bijection $\mathcal{P} \longrightarrow \mathcal{P}', \lambda \mapsto \lambda^{\sharp}$ such that $(k,l) \in Y_+^{\lambda}$ if and only if $(l,k) \in Y_+'^{\lambda^{\sharp}}$.
\begin{lem} \label{lem: transpose of Young diagrams}
Let $\lambda \in \mathcal{P}$.
\begin{itemize}
\item[(1)] As $U_q(\Glie')$-modules $\SF^*\left(V_q(\lambda)\right) \cong V_q'^*(\lambda^{\sharp})$ and $\SF^*\left(V_q^*(\lambda)\right) \cong V_q'(\lambda^{\sharp})$.
\item[(2)] If $T \in \CB_-(\lambda)$, then $T'(k,l) := M+N+1- T(-l,-k)$ defines an element $T' \in \CB_+'(\lambda^{\sharp})$. Moreover $T\mapsto T'$ is a bijection $\CB_-(\lambda) \longrightarrow \CB_+'(\lambda^{\sharp})$.
\end{itemize}
\end{lem}
\begin{proof}
(2) is a lengthy but straightforward check by Definition \ref{def: tableau}. For (1), it suffices to establish the second isomorphism since $\SF$ respects Hopf superalgebra structures. Let $\mu$ be the lowest weight of $V_q(\lambda)$ and define 
\begin{gather*}
 r_i := \sharp \{j \in \BZ_{>0}\ |\ (i,j) \in Y_+^{\lambda} \},\quad c_j := \sharp \{i \in \BZ_{>0}\ |\ (i,j) \in Y_+^{\lambda} \};  \\
 r_i' := \max (r_i - N, 0),\quad c_j'  := \max (c_j - M, 0).
\end{gather*} 
Then from \cite[(4.1)--(4.2)]{BKK} we have
$$\lambda = \sum_{i=1}^M r_i \epsilon_i + \sum_{j=1}^N c_j' \epsilon_{M+j}, \quad \mu = \sum_{i=1}^M r_{M+1-i}' \epsilon_i + \sum_{j=1}^N c_{N+1-j} \epsilon_{M+j}.$$
If $v$ is a lowest weight vector of $V_q(\lambda)$, then $V_q^*(\lambda)$ contains a highest weight vector $v^*$ of weight $-\mu$, and $\SF^*(v^*) \in \SF^*\left(V_q^*(\lambda)\right)$ is a highest weight vector of weight $$c_1\epsilon_1+c_2\epsilon_2 + \cdots + c_N \epsilon_N + r_1'\epsilon_{N+1} + r_2'\epsilon_{N+2} + \cdots + r_M'\epsilon_{M+N}, $$
which is exactly $\lambda^{\sharp}$, leading to the desired isomorphism.
\end{proof}
For $i \in I$ let $U_q^{\leq i}(\Glie') := \SF^{-1}(U_q^{\geq \kappa+1-i}(\Glie))$; it is the subalgebra of $U_q(\Glie')$ generated by the $s_{jk}'^{(0)}, t_{jk}'^{(0)}$ with $j,k \leq i$. 
To decompose $V_q(\lambda)$ (resp. $V_q^*(\lambda))$ with respect to lowest weights along the ascending chain of subalgebras of $U_q(\Glie)$
$$U_q^{\geq \kappa}(\Glie) \subset U_q^{\geq \kappa - 1}(\Glie) \subset \cdots \subset U_q^{\geq 2}(\Glie) \subset U_q^{\geq 1}(\Glie) = U_q(\Glie) $$
is to decompose $V_q'(\lambda^{\sharp})$ with respect to highest (resp. lowest) weights along
$$ U_q^{\leq 1}(\Glie') \subset U_q^{\leq 2}(\Glie') \subset \cdots \subset U_q^{\leq \kappa-1}(\Glie') \subset U_q^{\leq \kappa}(\Glie') = U_q(\Glie'), $$
\begin{rem} \label{rem: GT}
By \cite{BKK}, $V_q'(\lambda^{\sharp})$ is an irreducible submodule of a tensor power of $V_q'(\epsilon_1)$, and  all such tensor powers are semi-simple $U_q(\Glie')$-modules. So the decomposition for $V_q'(\lambda^{\sharp})$ is equivalent to that for the character formula in Remark \ref{rem: BKK}, and then to the branching rule of $\Glie'$-modules in \cite[Section 5]{BR}. We reformulate the latter in terms of $\CB_+'(\lambda^{\sharp})$, equivalently $\CB_-(\lambda)$ by Lemma \ref{lem: transpose of Young diagrams}, as follows. 
\begin{itemize}
\item[(1)] $V_q(\lambda)$ admits a basis $(v_T: T\in \CB_-(\lambda))$ such that $v_T$ is contained in an irreducible sub-$U_q^{\geq i}(\Glie)$-module of lowest weight $\mu_{T}^{\geq i}$ for $i \in I$.
\item[(2)] $V_q^*(\lambda)$ admits a basis $(w_T: T\in \CB_-(\lambda))$ such that $w_T$ is contained in an irreducible sub-$U_q^{\geq i}(\Glie)$-module of lowest weight $-\nu_{T}^{\geq i}$ for $i \in I$.
\end{itemize}
 $\mu_T^{\geq i}$ and $\nu_T^{\geq i}$ are defined as follows. Set $Y_{T}^{\geq i} := \{(k,l) \in Y_-^{\lambda} \ | \ T(k,l) \geq i\}$ and
$$ r_k := \sharp \{l \in \BZ\ |\  (-k,-l) \in Y_{T}^{\geq i}  \},\quad c_l := \sharp\{k \in \BZ\ |\ (-k,-l) \in Y_{T}^{\geq i}\}.  $$
If $i > M$, then $\begin{cases}
\mu_{T}^{\geq i} = c_1 \epsilon_{M+N} + c_2 \epsilon_{M+N-1} + \cdots + c_{M+N+1-i} \epsilon_{i}, \\
\nu_T^{\geq i} = c_1 \epsilon_i +c_2 \epsilon_{i+1} + \cdots + c_{M+N+1-i}\epsilon_{M+N}.
\end{cases}$
If $i \leq M$, then
\begin{align*}
 \mu_{T}^{\geq i} &= c_1 \epsilon_{M+N} + c_2 \epsilon_{M+N-1} + \cdots + c_{N} \epsilon_{M+1} + r_1' \epsilon_M + r_2' \epsilon_{M-1} + \cdots + r_{M+1-i}' \epsilon_i, \\
 \nu_T^{\geq i} &= r_1\epsilon_i + r_2\epsilon_{i+1} + \cdots + r_{M+1-i}\epsilon_M + c_1' \epsilon_{M+1} + c_2' \epsilon_{M+2} + \cdots + c_N' \epsilon_{M+N},
\end{align*}
where $r_k' := \max(r_k-N,0)$ and $c_l' := \max(c_l - M+i-1,0)$.
\end{rem}
\begin{example} \label{example: tableaux permutation}
To illustrate Lemma \ref{lem: transpose of Young diagrams} (2) and Remark \ref{rem: GT}, let $\Glie = \mathfrak{gl}(2|3)$ and $\lambda = 4 \epsilon_1 + 2 \epsilon_2 + \epsilon_3 \in \mathcal{P}$. We represent elements in $\CB_-(\lambda)$ and $\CB_+'(\lambda^{\sharp})$ by Young tableaux of shapes $\lambda, \lambda^{\sharp}$ respectively. Let $T \in \CB_-(\lambda)$ be such that
$$ \CB_-(4 \epsilon_1 + 2 \epsilon_2 + \epsilon_3) \ni T = \young(:::1,::22,1345) \mapsto \young(145,24:,3::,5::) = T' \in \CB_+'(3 \epsilon_1 + 2 \epsilon_2 + \epsilon_3 + \epsilon_4). $$
The Young diagrams $Y_T^{\geq i}$ with descending order on $5 \geq i \geq 1$ become:
$$ \young(~),\quad \young(~~),\quad \young(~~~),\quad \young(:~~,~~~),\quad \young(:::~,::~~,~~~~). $$ 
Correspondingly, the pairs $(\mu_T^{\geq i}, \nu_T^{\geq i})$ from $i = 5$ to $i=1$ are:
\begin{gather*}
(\epsilon_5,\ \epsilon_5),\quad (\epsilon_4+\epsilon_5,\ \epsilon_4+\epsilon_5), \quad (\epsilon_3+\epsilon_4+\epsilon_5,\ \epsilon_3+\epsilon_4+ \epsilon_5), \\
(\epsilon_3+2\epsilon_4+2\epsilon_5,\ 3\epsilon_2+\epsilon_3+\epsilon_4),\quad (\epsilon_2+\epsilon_3+2\epsilon_4+3\epsilon_5,\ 4\epsilon_1+2\epsilon_2+\epsilon_3).
\end{gather*}
\end{example}
\noindent {\bf Proof of Equations \eqref{for: iMA}--\eqref{for: dual iMA}. }
Let us define $g_i(z),g_i^*(z) \in \BC[[z]]^{\times}$ for $i \in I$:
$$ g_i(z) := \prod_{(k,l) \in Y_T^{\geq i}} q\frac{1-zaq^{2(l-k)}}{1-zaq^{2(l-k+1)}},\quad g_i^*(z) := \prod_{(k,l) \in Y_T^{\geq i}} \left(q^{-1} \frac{1-zat_iq^{2(k-l+1)}}{1-zat_iq^{2(k-l)}} \right). $$
By Lemma \ref{lem: JC dual action}, it suffices to prove that: for $i \in I$,
$$ C_i(z) v_T =  g_i(z) v_T \quad \mathrm{in}\ V_q^+(\lambda;a),\quad C_i(z)w_T =  g^*(z)w_T\quad \mathrm{in}\ V_q^{+*}(\lambda;a). $$
This is divided into two cases: $i > M$ or $i \leq M$.

Assume  $i > M$. Then $T(-k,-l) \geq i$ if and only if $1\leq l \leq M+N-i+1$ and $1\leq k \leq c_l$. It follows from Equation \eqref{equ: theta} that 
\begin{align*}
g_i(z) &= \prod_{l=1}^{M+N-i+1} \prod_{k=1}^{c_l} q\frac{1-zaq^{2(k-l)}}{1-zaq^{2(k-l+1)}} = \prod_{l=1}^{M+N-i+1}  \frac{1-zaq^{2(1-l)}}{q^{-c_l}-zaq^{2(1-l)+c_l}}  \\
&= \prod_{j=i}^{M+N} \left(\frac{q^{(\mu_T^{\geq i},\epsilon_j)} - za\theta_j q^{-(\mu_{T}^{\geq i}, \epsilon_j)}}{1-za\theta_j}\right)^{(\epsilon_j,\epsilon_j)}, \\
g_i^*(z) &= \prod_{l=1}^{M+N-i+1} \prod_{k=1}^{c_l} q^{-1}\frac{1-zat_iq^{2(l-k+1)}}{1-zat_iq^{2(l-k)}} = \prod_{l=1}^{M+N-i+1} \frac{1-zat_iq^{2l}}{q^{c_l} -zat_iq^{2l-c_i}} \\
&=\prod_{j=i}^{M+N}\left( \frac{q^{-(\nu_T^{\geq i},\epsilon_i)} - za\theta_j q^{(\mu_T^{\geq i},\epsilon_i)} }{1-za\theta_j} \right)^{(\epsilon_j,\epsilon_j)}.
\end{align*}
Here in the last equation we used $t_i q^{2l} = \tau_{i-1}^2q^{2l-2} = \theta_iq^{2l-2} = \theta_{i+l-1}$.

Assume $i \leq M$. Then $T(-k,-l) \geq i$ if and only if $(1\leq l \leq N,\ 1\leq k \leq c_l)$ or $(1\leq k \leq M+1-i,\ N+1\leq l \leq N+r_k')$. This gives
\begin{align*}
g_i(z) &=\left( \prod_{l=1}^N \prod_{k=1}^{c_l}q\frac{1-zaq^{2(k-l)}}{1-zaq^{2(k-l+1)}}\right) \times \left(\prod_{k=1}^{M+1-i} \prod_{l=1}^{r_k'} q\frac{1-zaq^{2(k-l-N)}}{1-zaq^{2(k-l-N+1)}}  \right) \\
&= \prod_{j= i}^{M+N}\left(\frac{q^{(\mu_T^{\geq i},\epsilon_j)} - za\theta_j q^{-(\mu_{T}^{\geq i}, \epsilon_j)}}{1-za\theta_j}\right)^{(\epsilon_j,\epsilon_j)}.
\end{align*}
Notice that $T(-k,-l) \geq i$ if and only if $(1\leq k \leq M+1-i,\ 1\leq l \leq r_k)$ or $(1\leq l\leq N,\ M-i+2 \leq k \leq M-i+1+c_l')$. This gives
\begin{align*}
g_i^*(z) &= \left( \prod_{k=1}^{M+1-i} \prod_{l=1}^{r_k} q^{-1} \frac{1-zat_iq^{2(l-k+1)}}{1-zat_iq^{2(l-k)}} \right) \left( \prod_{l=1}^N \prod_{k=1}^{c_l'} q^{-1}\frac{1-zat_iq^{2(l-k-M+i)}}{1-zat_iq^{2(l-k-M+i-1)}} \right)  \\
&=\left( \prod_{k=1}^{M+1-i} \frac{q^{-r_k} - zat_iq^{2(1-k+r_k)}}{1-zat_iq^{2(1-k)}} \right) \left(\prod_{l=1}^N \frac{1-zat_iq^{2(l-1-M+i)}}{q^{c_l'}-zat_iq^{2(l-1-M+i-c_l')}} \right) \\
&=\prod_{j=i}^{M+N}\left( \frac{q^{-(\nu_T^{\geq i},\epsilon_i)} - za\theta_j q^{(\mu_T^{\geq i},\epsilon_i)} }{1-za\theta_j} \right)^{(\epsilon_j,\epsilon_j)}.
\end{align*} 
The last identity comes from $t_i q^{2(l-1-M+i)} = \theta_{M+l}$ and $t_i q^{2(1-k)} = \theta_{i+k-1}$.

In both cases, $g_i(z)$ and $g_i^*(z)$ become Equation \eqref{for: eigen JM} with $\mu = \mu_T^{\geq i}$ and $-\nu_T^{\geq i}$ respectively, and this completes the proof of Equations \eqref{for: iMA}--\eqref{for: dual iMA}. \qed

%\begin{itemize}
%\item[(ii)] $Y_{i,a} := q^{\varpi_i} \Psi_{i,aq_i^{-1}}\Psi_{i,aq_i}^{-1} \in \BR_U$ where $\varpi_i = \epsilon_1+\epsilon_2+\cdots +\epsilon_i$ if $i \leq M$ and $\varpi_i = -(\epsilon_{i+1}+\epsilon_{i+2}+\cdots + \epsilon_{\kappa})$ if $i > M$.
%\item[(iii)] $\varpi_{m,a}^{(i)} := Y_{i,a}Y_{i,aq_i^{-2}} \cdots Y_{i,aq_i^{2-2m}} \in \BR_U$.
%\item[(iv)]  and $W_{m,a}^{(i)} := L(\varpi_{m,a}^{(i)})$ {\it Kirillov--Reshetikhin module} over $U_q(\Gaff)$.
%\end{itemize}

Let $\lCQ^-$ be the submonoid of $\BR$ generated by the $A_{i,a}^{-1}$ with $i \in I_0$ and $a \in \BC^{\times}$.
\begin{cor}  \label{cor: KR evaluation}
Let $i \in I_0,\ a \in \BC^{\times}$ and $m \in \BC^{\times}$. We have 
\begin{align}
W_{m,a}^{(i)} &\cong V_q^+(m\varpi_i; aq^{M-N-i}) \cong V_q^-(m\varpi_i; aq^{N-M+i-2m})\ \mathrm{if}\ i \leq M, \label{equ: KR +} \\
W_{m,a}^{(i)} & \cong V_q^{-*}(\lambda_m^{(i)}; aq^{M+N-2-i})  \simeq V_q^{+*}(\lambda_m^{(i)};aq^{i-M-N+2m-2})\ \mathrm{if}\ i > M. \label{equ: KR -}
\end{align}
Here for $i > M$, the Young diagram of $\lambda_m^{(i)} \in \mathcal{P}$ is a rectangle with $m$ rows and $\kappa-i$ columns. An $\ell$-weight of $W_{m,a}^{(i)}$ different from $\varpi_{m,a}^{(i)}$ must belong to $\varpi_{m,a}^{(i)}A_{i,aq_i}^{-1}\lCQ^-$.
\end{cor}
\begin{proof}
Assume $i\leq M$. The Young diagram $Y^{m\varpi_i}_-$ is a rectangle with $i$ rows and $m$ columns. Let $H \in \CB_-(m\varpi_i)$ be such that $H(-k,-l) = i+1-k$ for $1\leq k \leq i$. Then $v_H \in V_q^+(m\varpi_i;a\tau_iq^{-1})$ in Remark \ref{rem: GT} is a highest $\ell$-weight vector of $\ell$-weight
$$ m_H =  \prod_{l=1}^m\prod_{k=1}^i \boxed{k}_{a\tau_iq^{2(i+1-k-l)}} = \prod_{l=1}^m Y_{i,aq^{2-2l}} = \varpi_{m,a}^{(i)}. $$
Here we used $\prod_{k=1}^i \boxed{k}_{a\tau_iq^{2(i+1-k-l)}} = Y_{i,aq^{2-2l}}$ and $\theta_i = \tau_i^2 = q^{2(M-N+1-i)}$ for $1\leq i \leq M$, based on Example \ref{ss: l-weights}. This proves the first isomorphism of \eqref{equ: KR +}; the second one is a consequence of Equations \eqref{for: iMA} and \eqref{for: dMA}. If $T \in \CB_-(m\varpi_i)$ and $T \neq H$, then $T(-k,-l) \geq i+1-k$ and $T(-1,-1) > i$. The $\ell$-weight property of $W_{m,a}^{(i)}$ follows from Definition \ref{def: tableau} and Equation \eqref{for: iMA}:
$$m_T m_H^{-1} \in \boxed{i+1}_{a\tau_i}\boxed{i}_{a\tau_i}^{-1} \lCQ^- = A_{i,aq}^{-1} \lCQ^-.$$

Assume $i>M$. Let $v$ be the highest $\ell$-weight of $V_q^{-*}(\lambda_m^{(i)};b)$. By Equation \eqref{homo:evaluation dec}, 
\begin{align*}
K_p^+(z) v = v \quad \mathrm{for}\ p \leq i,\quad K_p^+(z) v = \frac{1-zb}{q^{-m}-zbq^m} v \quad \quad \mathrm{for}\ p > i.
\end{align*}
$v$ is of $\ell$-weight $\varpi_{m,b\tau_jq}^{(j)}$, proving the first isomorphism of \eqref{equ: KR -}.Since $\boxed{l}_a^* \equiv \boxed{l}_a'$ for $l \in I$, the second isomorphism of \eqref{equ: KR -} is deduced from Equations \eqref{for: dual iMA} and \eqref{for: dual dMA}. let $H \in \CB_+(\lambda_m^{(i)})$ be such that $H(k,l) = i+l$ for $1\leq l \leq M+N-i$. The monomial $m_H'$ associated to $H$ in Equation \eqref{for: dual dMA} is the highest $\ell$-weight. If $T \in \CB_+(\lambda_m^i)$ and $T \neq H$, then $T(k,l) \leq i+l$ and $T(1,1) \leq i$. By Definition \ref{def: tableau} and Equation \eqref{for: dual dMA}:
$$ m_T' m_H'^{-1} \in \boxed{i}'_{a\tau_i^{-1}}\boxed{i+1}_{a\tau_i^{-1}}'^{-1} \lCQ^- = A_{i,aq^{-1}}^{-1}\lCQ^-, $$
proving the $\ell$-weight property of $W_{m,a}^{(i)}$.
\end{proof}
The $\ell$-weight property is similar to \cite[Lemma 4.4]{H}; $W_{m,a}^{(i)}$ in \cite{H} is $W_{m,aq_i^{2m-2}}^{(i)}$ here. Let $\varpi_{m,a}^{(M-)} := \prod_{l=1}^m Y_{M,aq^{2l-2}}^{-1}$ and $W_{m,a}^{(M-)} := L(\varpi_{m,a}^{(M-)})$. Similarly we have
\begin{equation} \label{equ: KR odd}
W_{m,a}^{(M-)} \cong V_q^{-*}(\lambda_m;aq^{N-2}) \simeq V_q^{+*}(\lambda_m;aq^{2m-2-N}).
\end{equation}
where $\lambda_m \in \mathcal{P}$ is such that its Young diagram is a rectangle with $m$ rows and $N$ columns. If $\varpi_{m,a}^{(M-)} \Bn \in \lwt(W_{m,a}^{(M-)})$ and $\Bn \neq 1$, then $\Bn \in A_{M,aq^{-1}}^{-1}\lCQ^-$.

%The second isomorphisms in Equations \eqref{equ: KR -}--\eqref{equ: KR odd} will be made explicit in the proof of Lemma \ref{lem: polyn asym}.

% For $i,j \in I_0$ write $i \sim j$ if $j = i \pm 1$. Recall $q_{ij} = q^{(\alpha_i,\alpha_j)}$.
%
%\begin{lem}  \label{lem: Y and A-variables}
%Let $i \in I_0 \setminus \{M,M+1\}$ and $a \in \BC^{\times}$. In $\BR$ we have 
%$$ A_{i,a} \equiv Y_{i,aq}Y_{i,aq^{-1}} \prod_{j\in I_0: j\sim i} Y_{j,a}^{-1},\quad A_{M,a} \equiv \prod_{j\in I_0: j\sim M} Y_{j,a}^{-1}. $$
%If $N > 2$, then $A_{M+1,a} \equiv Y_{M+1,aq}Y_{M+1,aq^{-1}} Y_{M,a} Y_{M+2,a}^{-1}$. If $N = 2$, then $A_{M+1,a} \equiv Y_{M+1,aq}Y_{M+1,aq^{-1}} Y_{M,a}$. In terms of the $\Psi$, for $i\in I_0$ we have:
%
%where 
%\end{lem}
%\begin{proof}
%Straightforward based on Section \ref{ss: l-weights}.
%\end{proof}
\section{Length-two representations}  \label{sec: TQ}
A $Y_q(\Glie)$-module $V$ in category $\BGG$ is {\it of length-two} if it admits a Jordan--H\"older series of length two, namely, it fits in a short exact sequence $0 \rightarrow S \rightarrow V \rightarrow S' \rightarrow 0$ in category $\BGG$ such that both $S'$ and $S'$ are irreducible. We shall simply write such a sequence as $S \hookrightarrow V \twoheadrightarrow S'$.

In this section we describe length-two modules by tensor products.

For $i \in I_0, a \in \BC^{\times}, m \in \BZ_{>0}$ and $s \in \BZ_{\geq 0}$, let us define $\Bd_{m,a}^{(i,s)} \in \BR_U$ to be
\begin{gather*}
\varpi_{m,aq_i^{2m+1}}^{(i)} \varpi_{m+s,aq_i^{2m-1}}^{(i)} \prod_{l=1}^m A_{i,aq_i^{2l}}^{-1}  \ \mathrm{if}\ i \neq M, \quad
 \varpi_{s,aq^{-1}}^{(M)} \prod_{j\in I_0: j \sim M} \varpi_{m,aq_j^{2m}}^{(j)} \ \mathrm{if}\ i = M. 
\end{gather*}
Let $D_{m,a}^{(i,s)} := L(\Bd_{m,a}^{(i,s)})$ be the irreducible $U_q(\Gaff)$-module.
\begin{rem} \label{rem: comparison with HJ}
Let us rewrite $\Bd_{m,a}^{(i,s)}$ in terms of the $\Psi$ using Equation \eqref{equ: A Psi}:
\begin{gather*} 
 \Bd_{m,a}^{(i,s)} \equiv  \frac{\Psi_{i,aq_i^{-2s}}}{\Psi_{i,a}} \prod_{j\in I_0: j \sim i} \frac{\Psi_{j,aq_{ij}^{-1}}}{\Psi_{j,aq_{ij}^{-2m-1}}}. 
\end{gather*}
In the non-graded case $N=0$, we can identify $\Bn_{i,a}^+$ with $\Psi$ in \cite[(6.13)]{HL} and $\mathfrak{m}_{i,a}^{(2)}$ in \cite[(6.2)]{Jimbo2}, $\Bd_{m,a}^{(i,s)}$ with $\widetilde{\Psi}_i^{(-s,2m-1)}$ in \cite[Section 4.3]{FH2}. Notice that $\Bd_{m,a}^{(i,s)}$ satisfies the condition of \lq\lq minimal affinization by parts\rq\rq\ in \cite[Theorem 2]{ChariPrime}.
\end{rem}

\begin{theorem} \label{thm: TQ}
Let $i \in I_0$ and $a \in \BC^{\times}$. The $Y_q(\Glie)$-module $N_{i,a}^+ \otimes L_{i,a}^+$ has a Jordan--H\"{o}lder series of length two and in the Grothendieck ring $K_0(\BGG)$:
\begin{equation} \label{equ: TQ}
[N_{i,a}^+ \otimes L_{i,a}^+] =  [L_{i,aq_i^{-2}}^+] \prod_{j\in I_0:j\sim i}[L_{j,aq_{ij}^{-1}}^+] + [D] [L_{i,a\hat{q}_i^2}^+] \prod_{j\in I_0: j\sim i} [L_{j,aq_{ij}}^+].
\end{equation}
Here $D = L(\Bn_{i,a}^+\Psi_{i,a}\Psi_{i,a\hat{q}_i^{2}}^{-1} A_{i,a}^{-1}\prod_{j\sim i} \Psi_{j,aq_{ij}}^{-1})$ is one-dimensional.
\end{theorem}
When $i = M$, the two monomials at the right-hand side of Equation \eqref{equ: TQ} has a common factor $[L_{M,aq^{-2}}^+]$. This is a special feature of quantum affine superalgebras.
\begin{theorem} \label{thm: Demazure T}
Let $i \in I_0 \setminus \{M\},\ a \in \BC^{\times}$ and $m,s \in \BZ_{>0}$. There are short exact sequences of $U_q(\Gaff)$-modules whose first and third terms are irreducible: 
\begin{gather*}
D_{m,a}^{(i,s)} \hookrightarrow W_{m,aq_i^{2m+1}}^{(i)} \otimes W_{m+s,aq_i^{2m-1}}^{(i)} \twoheadrightarrow W_{m+s+1,aq_i^{2m+1}}^{(i)} \otimes W_{m-1,aq_i^{2m-1}}^{(i)}, \label{equ: Demazure resolution} \\
 D_{m+s,aq_i^{-2s}}^{(i,0)} \otimes W_{m,aq_i^{2m+1}}^{(i)} \hookrightarrow W_{m+s,aq_i^{2m+1}}^{(i)} \otimes D_{m,a}^{(i,s)} \twoheadrightarrow W_{m+s+1,aq_i^{2m+1}}^{(i)} \otimes D_{m,a}^{(i,s-1)}. \label{equ: Demazure TQ}
\end{gather*}
\end{theorem}
The assumption $i \neq M$ is necessary because $\dim W_{m,a}^{(M)} = 2^{MN}$ for $m \geq N$. 
 Equation \eqref{equ: TQ} corresponds to \cite[(6.14)]{HL} and \cite[Proposition 6.8]{Jimbo2}, and  can be thought of as a two-term Baxter TQ relation for $Y_q(\Glie)$. The exact sequences of Theorem \ref{thm: Demazure T} are extended T-systems \cite{Nakajima,H}, the initial case $s = 0$ being the T-system in \cite{Tsuboi}; see Theorem \ref{thm: Tsuboi}.

The proof of Theorem \ref{thm: TQ}, given in Section \ref{sec: proof TQ}, is similar to \cite[(6.14)]{HL}, based on $q$-characters. Theorem \ref{thm: Demazure T} is more involved and requires cyclicity of tensor products of KR modules; its proof is postponed to Section \ref{sec: proof T}. 

We make crucial use of the idea that $D_{m,a}^{(i,s)}$ admits an injective resolution by tensor products of KR modules of the same Dynkin node for $i \neq M$.

%Let us refine Corollary \ref{cor: KR evaluation} as in \cite[Lemma 5.5]{H} and \cite[Theorem 3.2]{Nakajima}.
\begin{lem} \label{lem: KR l-weights refined}
Let $m \in \BZ_{>0},\ a \in \BC^{\times}$ and $i \in I_0 \setminus \{M\}$. If $\varpi_{m,a}^{(i)} \Bn \in \lwt(W_{m,a}^{(i)})$ and $\Bn \neq 1$, then either $\Bn = A_{i,aq_i}^{-1} A_{i,aq_i^{-1}}^{-1} \cdots A_{i,aq_i^{3-2l}}^{-1}$ for some $1\leq l \leq m$, or $\Bn$ belongs to $A_{i,aq_i}^{-1}A_{j,aq_i^2}^{-1} \lCQ^-$ where $j \in I_0$ and $j \sim i$.
\end{lem}
\begin{proof}
We only consider the case $i < M$; the other case is similar. Let us  be in the situation of the proof of Corollary \ref{cor: KR evaluation}. By Equation \eqref{for: iMA}, $\Bn = m_Tm_H^{-1}$ for a unique $T \in \CB_-(m\varpi_i)$ with $T(-1,-1) > i$ and $T(-k,-l) \geq i+1-k$. If $T(-1,-1) > i+1$, then using $\tau_{i+1} = q^{-1}\tau_i$ we obtain
$$ m_Tm_H^{-1} \in \boxed{i+2}_{a\tau_i}\boxed{i}_{a\tau_i}^{-1} \lCQ^- = A_{i,aq}^{-1}A_{i+1,aq^2}^{-1}\lCQ^-. $$
If $T(-2,-1) > i-1$, then together with $T(-1,-1) > i$ we have
$$ m_Tm_H^{-1} \in \boxed{i+1}_{a\tau_i}\boxed{i}_{a\tau_i}^{-1} \boxed{i}_{a\tau_iq^2}\boxed{i-1}_{a\tau_iq^2}^{-1}\lCQ^- = A_{i,aq}^{-1} A_{i-1,aq^2}^{-1} \lCQ^-. $$
Suppose $T(-1,-1) = i+1$ and $T(-2,-1) = i-1$. There exists $1\leq l \leq m$ such that the only difference between $T,H$ is at $(-1,-j)$ with $1\leq j \leq l$, and 
$$ m_Tm_H^{-1} = \prod_{j=1}^l \boxed{i+1}_{a\tau_iq^{2-2j}}\boxed{i}_{a\tau_iq^{2-2j}}^{-1} = \prod_{j=1}^l A_{i,aq^{3-2j}}^{-1}.  $$
This completes the proof of the lemma.
\end{proof}
\begin{cor} \label{cor: Demazure l-weights}
Let $m,s \in \BZ_{>0},\ a \in \BC^{\times}$ and $i \in I_0 \setminus \{M\}$.
\begin{itemize}
\item[(1)] For $1\leq l \leq s$, we have $\Bd_{m,a}^{(i,s)} A_{i,a}^{-1}A_{i,aq_i^{-2}}^{-1} \cdots A_{i,aq_i^{2-2l}}^{-1} \in \lwt(D_{m,a}^{(i,s)})$ and its associated $\ell$-weight space is one-dimensional.
\item[(2)] If $\Bd_{m,a}^{(i,s)}\Bn \in \lwt(D_{m,a}^{(i,s)})$ is not of the form of (1) and $\Bn \neq 1$, then $\Bn \in \{A_{j,aq_i^{2m+1}}^{-1},\ A_{i,aq_i^{2m+2}}^{-1}\ |\ j \in I_0,\ j \sim i \}\lCQ^-$. 
\end{itemize}
\end{cor}
\begin{proof}
For non-graded quantum affine algebras this corollary is \cite[Lemma 4.8]{FH2}, the proof utilized a delicate elimination theorem of $\ell$-weights \cite[Theorem 5.1]{H3}. Here our proof is a weaker version of elimination based on the restriction to the diagram subalgebra $U_i$ of $U_q(\Gaff)$ generated by $(x_{i,n}^{\pm}, \phi_{i,n}^{\pm})_{n\in \BZ}$. By Remark \ref{rem: two Gauss decompositions}, the algebra $U_i$ is a quotient of $U_{q_i}(\widehat{\mathfrak{sl}_2})$.

Set $T := W_{m,aq_i^{2m+1}}^{(i)} \otimes W_{m+s,aq_i^{2m-1}}^{(i)}$ and $S := L(\varpi_{m,aq_i^{2m+1}}^{(i)} \varpi_{m+s,aq_i^{2m-1}}^{(i)})$.  Then $S$ is a sub-quotient of $T$. Let $\lambda := (2m+s)\varpi_i$. By Corollary \ref{cor: KR evaluation}, 
\begin{itemize}
\item[(A)] $\dim T_{q^{\lambda-k\alpha_i}} = \min(m+1, k+1)$ for $0\leq k \leq m+s$.
\end{itemize}

1. Let $v_0 \in S$ be a highest $\ell$-weight vector and $S^i := U_i v_0 \subseteq S$.  Viewed as a $U_{q_i}(\widehat{\mathfrak{sl}_2})$-module, $S^i$ is of highest $\ell$-weight \cite[Section 2]{FR}
$$ \mathbf{m}_i := (Y_{aq_i^{2m+1}} Y_{aq_i^{2m-1}} \cdots Y_{aq_i^3}) (Y_{aq_i^{2m-1}} Y_{i,aq_i^{2m-3}} \cdots Y_{aq_i^{1-2s}}). $$
 $S^i$ is spanned by the $x_{i,n_1}^- x_{i,n_2}^- \cdots x_{i,n_k}^- v_0$. If $w \in S^i$ is annihilated by the $x_{i,n}^+$, then $x_{j,n}^+ w = 0 \in S$ for all $j \in I_0 \setminus \{i\}$ (because $[x_{j,n}^+,x_{i,k}^-] = 0$) and $w \in \BC v_0$. The $U_{q_i}(\widehat{\mathfrak{sl}_2})$-module $S^i$ is irreducible and has a factorization \cite[Theorem 4.8]{CP}: 
$$ S^i \cong L^{i}(Y_{aq_i^{2m+1}} Y_{aq_i^{2m-1}} \cdots Y_{aq_i^{1-2s}}) \otimes L^i (Y_{aq_i^{2m-1}} Y_{aq_i^{2m-3}} \cdots Y_{aq_i^3}), $$
where $L^i(\mathbf{n})$ denotes the irreducible $U_{q_i}(\widehat{\mathfrak{sl}_2})$-module of highest $\ell$-weight $\mathbf{n}$ (for $\Bn$ a product of the $Y_b$). For $ k \in \BZ_{>0}$, let $V_k \subseteq S^i$ be the subspace spanned by the $x_{i,n_1}^- x_{i,n_2}^- \cdots x_{i,n_k}^- v_0$ with $n_l \in \BZ$ for $1\leq l \leq k$. Then $V_k = S_{q^{\lambda-k\alpha_i}}$. Based on the $q$-character of $S^i$ with respect to the spectra of $\phi_i^+(z)$ in \cite[Section 4.1]{FR}, for $-1\leq l < s$ we have:
\begin{itemize}
\item[(B)] $\dim S_{q^{\lambda-k\alpha_i}} = \min(m,k+1)$ for $1\leq k \leq m+s$;
\item[(C)] $\mathbf{m}_i \prod_{t=-l}^{m} (Y_{aq_i^{2t+1}}^{-1} Y_{aq_i^{2t-1}}^{-1})$ is not an $\ell$-weight of the $U_{q_i}(\widehat{\mathfrak{sl}_2})$-module $S^i$.
\end{itemize}

2. By (A)--(B),  $\{\Bn \in \lwt(T)\setminus \lwt(S)\ |\ \varpi(\Bn) = \lambda-(m+l)\alpha_i \} = \{\Bn_l\}$ for $0\leq l \leq s$, the multiplicity of $\Bn_l$ in $\chi_q(T) - \chi_q(S)$ is one, and $L(\Bn_0)$ is a sub-quotient of $T$. Comparing the spectra of $\phi_i^+(z)$ by (C) and Lemma \ref{lem: KR l-weights refined}, we obtain: $\Bn_0 = \Bd_{m,a}^{(i,s)}$ and $\Bn_l = \Bd_{m,a}^{(i,s)} A_{i,a}^{-1}A_{i,aq_i^{-2}}^{-1} \cdots A_{i,aq_i^{2-2l}}^{-1}$. Part (2) follows by viewing $D_{m,a}^{(i,s)}$ as a sub-quotient of $T$. If  $(D_{m,a}^{(i,s)})_{q^{\lambda-(m+l)\alpha_i}} \neq 0$ for $1\leq l \leq s$, then necessarily $\Bn_l \in \lwt(D_{m,a}^{(i,s)})$ and its $\ell$-weight space is one-dimensional, proving (1). 

3. Let $w_0 \in D_{m,a}^{(i,s)}$ be a highest $\ell$-weight vector. Then $x_{i,0}^+ w_0 = 0$  and $\phi_{i,0}^+ w_0 = q_i^s w_0$. Since the triple $(x_{i,0}^+,\ x_{i,0}^-,\ \phi_{i,0}^+)$ generates a quotient algebra of $U_{q_i}(\mathfrak{sl}_2)$, we have $(D_{m,a}^{(i,s)})_{q^{\lambda-(m+l)\alpha_i}} \ni (x_{i,0}^-)^l w_0 \neq 0$ for $1\leq l \leq s$. 
\end{proof}

The case $i = M$ is distinguished since $U_M$ is not related to $U_q(\widehat{\mathfrak{sl}_2})$.

\begin{cor} \label{cor: Demazure odd l-weights}
Let $m,s \in \BZ_{>0}$ and $a \in \BC^{\times}$. 
\begin{itemize}
\item[(1)] $\Bd_{m,a}^{(M,s)} A_{M,a}^{-1} \in \lwt(D_{m,a}^{(M,s)})$ and the $\ell$-weight space is one-dimensional.
\item[(2)] $(\Bd_{m,a}^{(M,s)})^{-1}\lwt(D_{m,a}^{(M,s)}) \subset \left(\{A_{j,aq_j^{2m+1}}^{-1}\ |\ j \in I_0,\ j \sim M \}\lCQ^-\right)\cup \{1,A_{M,a}^{-1}\}$.
\end{itemize}
\end{cor}
\begin{proof}
Assume $M,N > 1$ without loss of generality. Let $\Bn \in (\Bd_{m,a}^{(M,s)})^{-1}\lwt(D_{m,a}^{(M,s)})$ with $\Bn \notin \{A_{M+1,aq^{-2m-1}}^{-1},\ A_{M-1,aq^{2m+1}}^{-1}\}\lCQ^-$ and $\Bn \neq 1$. 

Firstly, set $\lambda := s \varpi_M + m \varpi_{M-1}$. Then $\lambda \in \mathcal{P}$ and its Young diagram $Y_+^{\lambda}$ is formed of $(k,l)$ where either $(1\leq k < M,\  1\leq l \leq s+m)$ or $(k=M,\ 1\leq l \leq s)$. Consider the evaluation module $S := V_q^-(\lambda;aq^{N-2s-1})$. Let $H \in \CB_+(\lambda)$ be such that $H(k,l) = k$. The monomial $m_H$ attached to $H$ in Equation \eqref{for: dMA} is the highest $\ell$-weight of $S$. From the proof of Corollary \ref{cor: KR evaluation} we see that
$$ m_H = (Y_{M,aq^{1-2s}} \cdots Y_{M,aq^{-3}} Y_{M,aq^{-1}}) (Y_{M-1,aq^2} Y_{M-1,aq^4} \cdots Y_{M-1,aq^{2m}}). $$
In particular, the spectral parameters at the boxes $(M,s)$ and $(M-1,s+m)$ of $H$ are $a\tau_Mq^{-1}$ and $a\tau_{M-1}q^{2m}$ respectively. Let $T \in \CB_+(\lambda)$ and $T \neq H$. If $T(M-1,s+m) \geq M$, then by Definition \ref{def: tableau} and Equation \eqref{for: dMA},
$$ m_Tm_H^{-1} \in \boxed{M}_{a\tau_{M-1}q^{2m}}\boxed{M-1}_{a\tau_{M-1}q^{2m}}^{-1}  \lCQ^- = A_{M-1,aq^{2m+1}}^{-1}\lCQ^-. $$
If $T(M-1,s+m) < M$, then $T(k,l) = k$ for $k < M$ and by Equation \eqref{for: dMA}:
\begin{itemize}
\item[(i)] the $\ell$-weight space $S_{m_T}$ is also the one-dimensional weight space $S_{\varpi(m_T)}$;
\item[(ii)] $m_Tm_H^{-1} A_{M,a}$ is a product of the $A_{j,b}^{-1}$ with $j \geq M$; 
\item[(iii)] if $m_Tm_H^{-1} A_{M,a}$ is a product of the $A_{M,b}^{-1}$, then $m_Tm_H^{-1} A_{M,a} = 1$.
\end{itemize}
Here we used Definition \ref{def: Young} and $T(M,l) \geq M,\ T(M,s) > M$. 

Secondly, viewing $D_{m,a}^{(M,s)}$ as a sub-quotient of $S \otimes W_{m,aq^{-2m}}^{(M+1)}$ gives $\Bn = \Bn_1\Bn_2$ with $m_H \Bn_1 \in \lwt(S)$ and $\Bn_2 \varpi_{m,aq^{-2m}}^{(M+1)} \in \lwt(W_{m,aq^{-2m}}^{(M+1)})$. Since $\Bn \notin A_{M+1,aq^{-2m-1}}^{-1}\lCQ^-$, by Corollary \ref{cor: KR evaluation}, $\Bn_2 = 1$ and $m_H\Bn \in \lwt(S)$. Since $\Bn \notin A_{M-1,aq^{2m+1}}^{-1}\lCQ^-$, (ii)--(iii) hold by replacing $m_Tm_H^{-1}$ with $\Bn$, and $\dim(D_{m,a}^{M,s})_{\Bd_{m,a}^{(M,s)} \Bn} = 1$.

Thirdly, for $t \in \BZ_{>0}$, let $\mu_t \in \mathcal{P}$ be such that its Young diagram $Y_-^{\mu_t}$ is formed of $(-k,-l)$ where either $(1\leq l < N,\ 1\leq k \leq m+t)$ or $(l=N,\ 1\leq k \leq t)$. Consider the evaluation module $S_t := V_q^{+*}(\mu_t; aq^{2t-1-N})$. Let $H_t \in \CB_-(\mu_t)$ be such that $H_t(-k,-l) = M+N+1-l$. The monomial $m_{H_t}^*$ in Equation \eqref{for: dual iMA} is the highest $\ell$-weight of $S_t$ and by Corollary \ref{cor: KR evaluation} and Equation \eqref{equ: KR odd}:
$$ m_{H_t}^* \equiv \varpi_{m,aq^{-2m}}^{(M+1)} \varpi_{t,aq}^{(M-)}. $$
The spectral parameters at the boxes $(-t,-N)$ and $(-t-m,1-N)$ of $H_t$ are $a\tau_M^{-1}q$ and $a\tau_{M+1}^{-1}q^{-2m}$ respectively. Let $T \in \CB_-(\mu_t)$ and $T \neq H_t$. If $T(-t-m,1-N) < M+2$, then by Definition \ref{def: tableau} and Equation \eqref{for: dual iMA},
$$ m_T^* m_{H_t}^{*-1} \in \boxed{M+1}_{a\tau_{M+1}^{-1}q^{-2m}}^*\boxed{M+2}_{a\tau_{M+1}^{-1}q^{-2m}}^{*-1} \lCQ^- = A_{M+1,aq^{-2m-1}}^{-1}\lCQ^-.$$
If $T(-t-m,1-N) = M+2$, then $T(-k,-l) = M+N+1-l$ for $1\leq l < N$. Equation \eqref{for: dual iMA} implies that $m_T^*m_{H_t}^{*-1} A_{M,a}$ is a product of the $A_{j,b}^{-1}$ with $j \leq M$.

Lastly, viewing $D_{m,a}^{(M,s)}$ (after tensoring with a one-dimensional module) as a sub-quotient of 
$ S_t \otimes W_{t+s,aq^{2t-1}}^{(M)} \otimes W_{m,aq^{2m}}^{(M-1)}$
and choosing $t \in \BZ_{>0}$ so large that $\Bn \notin A_{M,aq^{2t}}^{-1}\lCQ^-$, we obtain $m_{H_t}^*\Bn \in \lwt(S_t)$, and so $\Bn A_{M,a}$ is a product of the $A_{j,b}^{-1}$ with $j \leq M$. From (ii)--(iii) it follows that $\Bn A_{M,a} = 1$.

It remains to show that $\Bd_{m,a}^{(M,s)} A_{M,a}^{-1} \in \lwt(D_{m,a}^{(M,s)})$. Indeed, as a $U_q(\Glie)$-module, $D_{m,a}^{(M,s)}$ has a highest weight vector of highest weight $q^{m\varpi_{M-1}+s\varpi_M+m\varpi_{M+1}}$, and so $q^{m\varpi_{M-1}+s\varpi_M+m\varpi_{M+1}-\alpha_M} \in \wt(D_{m,a}^{(M,s)})$. This means that there exists $\Bn \in (\Bd_{m,a}^{(M,s)})^{-1} \lwt(D_{m,a}^{(M,s)})$ with $\varpi(\Bn) = q^{-\alpha_M}$, which forces $\Bn = A_{M,a}^{-1}$.
\end{proof}

As an illustration, for $\Glie = \mathfrak{gl}(3|4)$ and $(m,s,t) = (2,3,1)$ we have
$$ H = \young(11111,22222,333) \in \CB_+(3\varpi_3+2\varpi_2),\quad H_t = \young(:567,:567,4567) \in \CB_-(3\varpi_3+\varpi_1). $$

%The idea of proof of Corollary \ref{cor: Demazure odd l-weights} works for $D_{m,a}^{(i,s)}$ with $i \neq M$. This leads to a stronger version of Corollary \ref{cor: Demazure l-weights} (2), where the last set can be replaced by $\{A_{j,aq_i^{2m+1}}^{-1}\ |\ j \in I_0,\ j \sim i\} \lCQ^-$.  We do not need this in the following.

\section{Proof of TQ relations: Theorem \ref{thm: TQ}}  \label{sec: proof TQ}
The crucial part in the proof is the irreducibility of arbitrary tensor products of positive prefundamental modules. In the case of quantum affine algebras this was proved in \cite[Theorem 4.11]{FH} and \cite[Lemma 5.1]{Jimbo2}. Our approach is similar to \cite{Jimbo2}, based on the duality functor $\SG^*$ in Lemma \ref{lem: duality by permutation}.

\begin{lem} \label{lem: positive pre char}
Let $a \in \BC^{\times}$ and $i \in I_0$. We have 
$$ \chi_q(L_{i,a}^+) = \Psi_{i,a} \times \chi(L_{i,a}^+). $$
\end{lem}
\begin{proof}
We can adapt the proof of \cite[Theorem 4.1]{FH}. Essentially we just need a weaker version of \cite[Lemma 4.5]{FH}: any $\ell$-weight of $W_{m,a}^{(i)}$ different from $\varpi_{m,a}^{(i)}$ belongs to $\varpi_{m,a}^{(i)} A_{i,aq_i}^{-1} \lCQ^-$, which is Corollary \ref{cor: KR evaluation}.  
\end{proof}
For negative prefundamental modules we recall the main results of \cite{Z5}.
\begin{lem}\cite[Lemma 6.7 \& Corollary 7.4]{Z5} \label{lem: negative pre char}
Let $a,c \in \BC^{\times}$ and $i \in I_0$. 
\begin{itemize}
\item[(i)] The $ \nqc(W_{m,aq_i^{-1}}^{(i)})$ for $m \in \BZ_{>0}$ are polynomials in $\BZ[A_{j,b}^{-1}]_{(j,b) \in I_0\times aq^{\BZ}}$, and as $m \rightarrow \infty$ they converge to a formal power series in $\BZ[[A_{j,b}^{-1}]]_{(j,b) \in I_0\times aq^{\BZ}}$, which is exactly the normalized $q$-character $\nqc(L_{i,a}^-)$. 
\item[(ii)] There exists a $U_q(\Gaff)$-module $\CW_{c,a}^{(i)}$ in category $\BGG$ such that
$$ \chi_q(\CW_{c,a}^{(i)}) = \aBw_{c,a}^{(i)} \times \nqc(L_{i,a}^-). $$
It is irreducible if $c \notin \pm q^{\BZ}$.
\end{itemize}
In particular, any $\ell$-weight of $L_{i,a}^-$ different from $\Psi_{i,a}$ belongs to  $\Psi_{i,a}A_{i,a}^{-1} \lCQ^-$.
\end{lem}
By \cite[Section 4]{Z5}, the $U_q(\Gaff)$-module $\CW_{c,a}^{(i)}$ is a \lq\lq generic asymptotic limit\rq\rq\ of the KR modules $W_{m,aq_i^{-1}}^{(i)}$; see also the proof of Lemma \ref{lem: polyn asym}.
\begin{cor} \label{cor: simplicity tensor product pre-fund}
Any tensor product of positive (resp. negative) prefundamental modules in category $\BGG$ is irreducible.
\end{cor}
\begin{proof}
In view of Lemmas \ref{lem: positive pre char}--\ref{lem: negative pre char}, the proof of \cite[Lemma 5.1]{Jimbo2} works here by replacing the duality of \cite[Lemma 3.5]{Jimbo2} with the functor $\SG^*$ in Lemma \ref{lem: duality by permutation}.
\end{proof}

\noindent {\bf Proof of Theorem \ref{thm: TQ}.}
In the non-graded case this was sketched in \cite[Section 6.1.3]{HL}. Here our proof is in the spirit of \cite[Lemma 4.8]{FH2}, by replacing the elimination theorem of $\ell$-weights therein with Corollaries \ref{cor: Demazure l-weights}--\ref{cor: Demazure odd l-weights}.

 Let $T := N_{i,a}^+ \otimes L_{i,a}^+$. We need to prove that $T$ has exactly two irreducible sub-quotient $S' := L(\Bn_{i,a}^+\Psi_{i,a})$ and $S'' := L(\Bn_{i,a}^+\Psi_{i,a} A_{i,a}^{-1})$ of multiplicity one, which implies Theorem \ref{thm: TQ} since $S'$ and $S''$ are irreducible tensor products of positive prefundamental modules with $D$. Clearly $S'$ is an irreducible sub-quotient of $T$, and $\chi_q(S') + \chi_q(S'') = \Bn_{i,a}^+\Psi_{i,a} (1+A_{i,a}^{-1}) \chi(L_{i,1}^+) \prod_{j\sim i} \chi(L_{j,1}^+)$ by Corollary \ref{cor: simplicity tensor product pre-fund}.

That $S''$ is a sub-quotient of $T$, i.e. $\chi_q(T)$ is bounded below by $\chi_q(S') + \chi_q(S'')$, is proved in the same way as in the first half of the comment after \cite[(6.13)]{HL}. For the reverse inequality, it suffices to show that $\chi_q(N_{i,a}^+)$ is bounded above by $\Bn_{i,a}^+ (1+A_{i,a}^{-1}) \prod_{j\sim i} \chi(L_{j,1}^+)$.

Assume $\Bn_{i,a}^+\Bn  \in \lwt(N_{i,a}^+)$ and $\Bn \neq 1$.
For $m \in \BZ_{>0}$ let $S_m := L(\Bn_{i,a}^+ (\Bd_{m,a}^{(i,1)})^{-1})$ and view $N_{i,a}^+$ as a sub-quotient of $D_{m,a}^{(i,1)} \otimes S_m$.  Write $$\Bn = \Bn_m' \Bn_m'',\quad\Bn_m' \Bd_{m,a}^{(i,1)} \in \lwt(D_{m,a}^{(i,1)}),\quad \Bn_m'' \Bn_{i,a}^+ (\Bd_{m,a}^{(i,1)})^{-1} \in \lwt(S_m).$$ By Remark \ref{rem: comparison with HJ}, we have $\Bn_{i,a}^+ (\Bd_{m,a}^{(i,1)})^{-1} \equiv \prod_{j\sim i}\Psi_{j,aq_{ij}^{-2m-1}}$.
 It follows from Corollary \ref{cor: simplicity tensor product pre-fund} that $\Bn_m'' \in q^{\BQ^-},\ \chi(S_m) = \prod_{j\sim i} \chi(L_{j,1}^+)$, and so $\Bn \in \lCQ^- q^{\BQ^-}$. 
 
 Choose $t \in \BZ_{>0}$ large enough so that $\Bn \in \lCQ_t^- q^{\BQ^-}$ where $\lCQ_t^-$ is the submonoid of $\lCQ$ generated by the $A_{j,aq^l}^{-1}$ with $-t < l < t$. Then for $m > t$, we must have $\Bn_m' \in  \{1,A_{i,a}^{-1}\}$  by Corollaries \ref{cor: Demazure l-weights}--\ref{cor: Demazure odd l-weights}. This implies that $\Bn_m''$ is uniquely determined by $\Bn$ and $\dim (N_{i,a}^+)_{\Bn} \leq \dim (S_m)_{\Bn_m''}$. As a consequence, the coefficient of any $\Bf\in \lCP$ in $\Bn_{i,a}^+(1+A_{i,a}^{-1}) \prod_{j\sim i} \chi(L_{j,1}^+) - \chi_q(N_{i,a}^+)$ is non-negative.  \hfill $\Box$

\section{Main result: asymptotic TQ relations} \label{sec: asy TQ}
We replace the $L, N$ in Equation \eqref{equ: TQ}  by $U_q(\Gaff)$-modules using the functor $\SG^*$. 
\begin{cor} \label{cor: TQ negative}
Let $i \in I_0$ and $a \in \BC^{\times}$. In the Grothendieck ring $K_0(\BGG)$:
\begin{equation}  \label{equ: TQ negative}
[N_{i,a}^-] [L_{i,a}^-] =  [L_{i,a\hat{q}_i^2}^-] \prod_{j\in I_0: j\sim i} [L_{j,aq_{ij}}^-] + [D][L_{i,aq_i^{-2}}^-] \prod_{j\in I_0:j\sim i}[L_{j,aq_{ij}^{-1}}^-]
\end{equation}
where $D= L(\Bn_{i,a}^- \Psi_{i,a}^{-1} A_{i,a}^{-1} \Psi_{i,aq_i^{-2}} \prod_{j\sim i} \Psi_{j,aq_{ij}^{-1}})$ is one-dimensional.
\end{cor}
\begin{proof}
Applying $\SG^{*-1}$ to Equation \eqref{equ: TQ} in $K_0(\BGG')$ gives \eqref{equ: TQ negative} by Lemma \ref{lem: duality by permutation}. Take $q$-characters in Equation \eqref{equ: TQ negative}. By Lemma \ref{lem: negative pre char}, $\Bn_{i,a}^-\Psi_{i,a}^{-1}A_{i,a}^{-1}$ appears at the left-hand side, but in none of the $\chi_q(L_{j,b}^-)$ at the right-hand side. This forces $\chi_q(D) \Psi_{i,aq_i^{-2}}^{-1} \prod_{j\sim i} \Psi_{j,aq_{ij}^{-1}}^{-1} = \Bn_{i,a}^-\Psi_{i,a}^{-1}A_{i,a}^{-1}$ and proves the second statement.
\end{proof}
Equation \eqref{equ: TQ negative} becomes \cite[Example 7.8]{HL} when $N = 0$.

\begin{prop} \label{prop: asymptotic T}
Let $i\in I_0$ and $a,c \in \BC^{\times}$. There exists a $U_q(\Gaff)$-module $\CN_{c,a}^{(i)}$ in category $\BGG$ whose $q$-character is 
$$\chi_q(\CN_{c,a}^{(i)}) = \Bn_{c,a}^{(i)} \times \nqc(N_{i,a}^-).  $$
If $c^2 \notin q^{\BZ}$, then $\CN_{c,a}^{(i)}$ is irreducible.
\end{prop}
The proof of this proposition will be given in Section \ref{sec: asym}. Assuming this proposition, we are able to prove the main result of the paper.
\begin{theorem} \label{thm: TQ asymptotic}
Let $i\in I_0$ and $a,c,d \in \BC^{\times}$. In the Grothendieck ring $K_0(\BGG)$:
\begin{multline} \label{equ: negative TQ asy}
\quad [\CN_{c,a}^{(i)}] [\CW_{d,a}^{(i)}] = [\CW_{d\hat{q}_i,a\hat{q}_i^2}^{(i)}] \prod_{j\in I_0: j\sim i} [\CW_{c_{ij}^{-1},aq_{ij}}^{(j)}] \\
+ [D_i^-][\CW_{dq_i^{-1},aq_i^{-2}}^{(i)}] \prod_{j\in I_0: j \sim i}[\CW_{c_{ij}^{-1}q_{ij}^{-1},aq_{ij}^{-1}}^{(j)}]
\end{multline}
where $D_i^- = L(\Bn_{c,a}^{(i)} \aBw_{d,a}^{(i)} A_{i,a}^{-1}(\aBw_{dq_i^{-1},aq_i^{-2}}^{(i)} \prod_{ j \sim i}\aBw_{c_{ij}^{-1}q_{ij}^{-1},aq_{ij}^{-1}}^{(j)})^{-1})$ is a one-dimensional $U_q(\Gaff)$-module.
If $c^2 \notin q^{\BZ}$, then in $K_0(\BGG)$
\begin{multline} \label{equ: positive TQ asy}
\quad [M_{c,a}^{(i)}] [\CW_{d,ad^2}^{(i)}] = [\CW_{dq_i,ad^2}^{(i)}] \prod_{j\in I_0: j\sim i} [\CW_{c_{ij}^{-1},aq_{ij}^{-1}c_{ij}^{-2}}^{(j)}] \\
+ [D_i][\CW_{d\hat{q}_i^{-1},ad^2}^{(i)}] \prod_{j\in I_0: j \sim i}[\CW_{c_{ij}^{-1}q_{ij}^{-1},aq_{ij}^{-1}c_{ij}^{-2}}^{(j)}]
\end{multline}
with $D_i= L( \Bm_{c,a}^{(i)} \aBw_{d,ad^2}^{(i)} A_{i,a}^{-1} ( \aBw_ {d\hat{q}_i^{-1},ad^2}^{(i)} \prod_{ j \sim i} \aBw_{c_{ij}^{-1}q_{ij}^{-1},aq_{ij}^{-1}c_{ij}^{-2}}^{(j)} )^{-1})$ one-dimensional.
\end{theorem}
The advantage of Equation \eqref{equ: positive TQ asy} over \eqref{equ: negative TQ asy} is that for fixed $j \in I_0$ the spectral parameter $a$ in $\CW_{c,a}^{(j)}$ is also fixed. This is crucial in deriving BAE in Section \ref{sec: Baxter}. 
\begin{proof}
$D_i^-$ is one-dimensional by the formulas in Example \ref{ss: l-weights}:
\begin{align*}
\Bn_{c,a}^{(i)} \aBw_{d,a}^{(i)} A_{i,a}^{-1} &\equiv \left(\frac{\Psi_{i,a}}{\Psi_{i,a\hat{q}_i^2}} \prod_{j \sim i} \frac{\Psi_{j,aq_{ij}c_{ij}^{2}}}{\Psi_{j,aq_{ij}}}\right) \times  \frac{\Psi_{i,ad^{-2}}}{\Psi_{i,a}} \times \left(\frac{\Psi_{i,aq_i^{-2}}}{\Psi_{i,a\hat{q}_i^2}} \prod_{j \sim i} \frac{\Psi_{j,aq_{ij}^{-1}}}{\Psi_{j,aq_{ij}}}\right)^{-1} \\
 &\equiv \frac{\Psi_{i,ad^{-2}}}{\Psi_{i,aq_i^{-2}}} \prod_{ j \sim i} \frac{\Psi_{j,aq_{ij}c_{ij}^{2}}}{\Psi_{j,aq_{ij}^{-1}}} \equiv \aBw_{dq_i^{-1},aq_i^{-2}}^{(i)} \prod_{j\in I_0: j \sim i}\aBw_{c_{ij}^{-1}q_{ij}^{-1},aq_{ij}^{-1}}^{(j)}.
\end{align*}
Dividing the $q$-characters of both sides of \eqref{equ: negative TQ asy} by $\Bn_{c,a}^{(i)} \aBw_{d,a}^{(i)}$, we obtain the normalized $q$-characters of \eqref{equ: TQ negative} by Lemma \ref{lem: negative pre char} and Proposition \ref{prop: asymptotic T}. This proves \eqref{equ: negative TQ asy}. For \eqref{equ: positive TQ asy}, let us assume first $d \notin \pm q^{\BZ}$.

 As in Table \eqref{tab: comparison}, let $\CN_{c,a}'^{(i)}, \CW_{c,a}'^{(i)}$ be the corresponding $U_q(\Gafft)$-modules in category $\BGG'$. Since $c^2, \pm d \notin q^{\BZ}$, by Lemma \ref{lem: negative pre char}, Proposition \ref{prop: asymptotic T} and Lemma \ref{lem: duality by permutation}, 
$ \SG^*(M_{c,a}^{(i)}) \simeq \CN_{c,aq^{N-M}}'^{(M+N-i)}$ and $\SG^*(\CW_{c,a}^{(i)}) \simeq \CW_{c^{-1},ac^{-2}q^{N-M}}'^{(M+N-i)}$ as irreducible $U_q(\Gafft)$-modules in category $\BGG'$.
Applying $\SG^{*-1}$ to \eqref{equ: negative TQ asy} in $K_0(\BGG')$ gives \eqref{equ: positive TQ asy}. The $\ell$-weight of $D_i$ is fixed similarly as in the proof of Corollary \ref{cor: TQ negative}. This implies 
$$ \chi_q(M_{c,a}^{(i)}) = \Bm_{c,a}^{(i)} (1+ A_{i,a}^{-1})  \prod_{j\in I_0: j \sim i} \nqc(L_{j,aq_{ij}^{-1}c_{ij}^{-2}}^-), $$
from which follows \eqref{equ: positive TQ asy} for arbitrary $c \in \BC^{\times}$.
\end{proof}
One can give an alternative proof to Equation \eqref{equ: positive TQ asy}, by slightly modifying that of Theorem \ref{thm: TQ}; see a closer situation in \cite[Theorem 6.1]{Z6}. This approach is independent of Theorem \ref{thm: Demazure T} and the results in Sections \ref{sec: cyclicity}--\ref{sec: proof T}.  

\section{Cyclicity of tensor products} \label{sec: cyclicity}
We provide a criteria for a tensor product of Kirillov--Reshetikhin modules to be of highest $\ell$-weight, which is needed to prove Theorem \ref{thm: Demazure T} and Proposition \ref{prop: asymptotic T}.

For $i,j \in \BZ_{>0}$ let us define the $q$-segment 
\begin{equation*} 
\CS(i,j) := \{ q^{-i-j+2r}\ |\ 0 \leq r < \min(i,j) \} \subset \BC^{\times}.
\end{equation*} 
It is $q^{j-i}\Sigma(i,j)^{-1}$ in \cite[Section 5]{Z4} and is symmetric in $i,j$.
\begin{theorem} \label{thm: tensor KR even}
Let $s \in \BZ_{>0}$. For $1\leq l \leq s$ let $1\leq i_l \leq M$ and $(m_l, a_l) \in \BZ_{>0}\times \BC^{\times}$. The $U_q(\Gaff)$-module $W_{m_1,a_1}^{(i_1)} \otimes W_{m_2,a_2}^{(i_2)} \otimes \cdots \otimes W_{m_s,a_s}^{(i_s)}$ is of highest $\ell$-weight if 
\begin{equation} \label{cond: cyclicity 1}
\frac{a_j}{a_k} \notin \bigcup_{p=1}^{m_j} q^{2p-2m_k} \CS(i_j,i_k)\quad \mathrm{for}\ 1\leq j < k \leq s.
\end{equation}
\end{theorem}

The idea is similar to \cite{Z3,Z4}, which in turn was inspired by \cite{Chari}, by restricting to diagram subalgebras. Let $A, B$ be Hopf superalgebras and let $\iota: A \longrightarrow B$ be a morphism of superalgebras. If $W$ is a $B$-module and $W'$ is a sub-$A$-module of the $A$-module $\iota^*(W)$, then let $\iota^{\bullet}(W')$ denote the $A$-module structure on $W'$.

For $1\leq p \leq 3$, define the quantum affine superalgebra $U_p$ with RTT generators $s_{ij;p}^{(n)}, t_{ij;p}^{(n)}$ and the superalgebra morphism $\iota_p: U_p \longrightarrow U_q(\Gaff)$ as follows:
$U_1 := U_q(\widehat{\mathfrak{gl}(1|1)}),\ U_2 := U_{q^{-1}}(\widehat{\mathfrak{gl}(1|1)})$ and $U_3 := U_q(\widehat{\mathfrak{gl}(M-1|N)})$, so that in $s_{ij;p}^{(n)}, t_{ij;p}^{(n)}$ we understand either $(1\leq i,j,p \leq 2)$ or $(1\leq i,j < M+N,\ p = 3)$;
\begin{eqnarray*}
&\iota_1: U_1 \longrightarrow U_q(\Gaff),\quad  & s_{ij;1}^{(n)} \mapsto s_{i'j'}^{(n)},\quad  t_{ij;1}^{(n)} \mapsto t_{i'j'}^{(n)}; \\
&\iota_2: U_2 \longrightarrow U_q(\Gaff),\quad  & s_{ij;2}^{(n)} \mapsto h(\overline{s}_{i'j'}^{(n)}),\quad  t_{ij;2}^{(n)} \mapsto  h(\overline{t}_{i'j'}^{(n)}); \\
&\iota_3: U_3 \longrightarrow U_q(\Gaff),\quad   & s_{ij;3}^{(n)} \mapsto s_{i+1,j+1}^{(n)},\quad  t_{ij;3}^{(n)} \mapsto t_{i+1,j+1}^{(n)}.
\end{eqnarray*}
Here $h$ is the involution in Equation \eqref{iso:involution} and $1' = 1,\ 2' = M+N$.

\begin{lem} \label{lem: Chari} \cite[Lemma 3.7]{Z3}
The tensor product of a lowest $\ell$-weight $U_q(\Gaff)$-module with a highest $\ell$-weight module is generated, as a $U_q(\Gaff)$-module, by a tensor product of a lowest $\ell$-weight vector with a highest $\ell$-weight vector.
\end{lem}

Let $1\leq p \leq 2$. We recall the notion of {\it Weyl module} over $U_p$  from \cite{Z4}. Let $f(z) \in \BC(z)$ be a product of the $c \frac{1-za}{1-zac^2}$ with $a,c \in \BC^{\times}$ and let $P(z) \in 1+z\BC[z]$ be such that $\frac{P(z)}{f(z)} \in \BC[z]$. The Weyl module $\WW_p(f;P)$ is the $U_p$-module generated by a highest $\ell$-weight vector $w$ of even parity such that 
$$ s_{11;p}(z) w = f(z) w = t_{11;p}(z) w,\quad s_{22;p}(z)w = w = t_{22;p}(z)w,$$
 and $\frac{P(z)}{f(z)} s_{21;p}(z) w$, as a formal power series in $z$ with coefficients in $\WW_p(f;P)$, is a polynomial in $z$ of degree $\leq \deg P$. Given another pair $(g,Q)$, if the polynomials $\frac{P(z)}{f(z)}$ and $Q(z)$ are co-prime, then $\WW_p(f;P) \otimes \WW_p(g;Q)$ is a quotient of $\WW_p(fg;PQ)$ and is of highest $\ell$-weight; see \cite[Proposition 14]{Z4}. 

\begin{example} \label{example: Weyl gl(1|1)}
 In the situation of Theorem \ref{thm: tensor KR even}, fix $v_l \in W_{m_l,a_l}^{(i_l)}$ a highest $\ell$-weight vector. Let $W_p$ be the sub-$U_p$-module of $\iota_p^*(\otimes_{l=1}^s W_{m_l,a_l}^{(i_l)})$ generated by $\otimes_{l=1}^sv_l$. Then $\iota_p^{\bullet}(W_p)$ is a quotient of the Weyl module $\WW_p$ for $1\leq p \leq 2$ where
 \begin{align*}
 & \WW_1 := \WW_1\left(\prod_{l=1}^s \frac{q^{m_l}-za_l q^{M-N-i_l-m_l}}{1-za_lq^{M-N-i_l}};\ \prod_{l=1}^s(1-za_lq^{M-N-i_l-2m_l})\right),\\
 & \WW_2 := \WW_2\left(\prod_{l=1}^s \frac{q^{-m_l}-za_lq^{N-M+i_l-m_l}}{1-za_lq^{N-M+i_l-2m_l}};\ \prod_{l=1}^s (1-za_lq^{N-M+i_l}) \right).
 \end{align*}
 $\iota_3^{\bullet}(W_3)$ is the tensor product $\otimes_{l=1}^s W_{m_l,a_l}^{3(i_l-1)}$ of KR modules over $U_3$.
The proof is the same as \cite[Lemmas 18 \& 20]{Z4}, based on Corollary \ref{cor: KR evaluation}.
\end{example}
 
For $p \in \BZ_{>0}$, let $\Glie_p := \mathfrak{gl}(1|p)$ and let $U_q(\Gaff_p)$ be the quantum affine superalgebra with RTT generators $s_{ij|p}^{(n)}, t_{ij|p}^{(n)}$ for $1\leq i,j \leq p+1$. Similarly $U_{q^{-1}}(\Gaff_p)$ with RTT generators $\overline{s}_{ij|p}^{(n)}, \overline{t}_{ij|p}^{(n)}$ and the involution $h_p: U_{q^{-1}}(\Gaff_p) \longrightarrow U_{q}(\Gaff_p)$ are defined. For $1\leq p \leq N$, the following extends uniquely to a superalgebra morphism 
\begin{equation}  \label{equ: auxiliary diagram subalg}
\vartheta_p: U_q(\Gaff_p) \longrightarrow U_q(\Gaff),\quad s_{ij|p}^{(n)} \mapsto s_{i'j'}^{(n)},\quad t_{ij|p}^{(n)} \mapsto t_{i'j'}^{(n)}
\end{equation}
where $1' =  1$ and $i' = M+N-p-1+i$ for $2\leq i \leq p+1$.

\begin{defi} \label{def: Weyl gl(1|p)}
Let $s \in \BZ_{>0}$ and $(m_l,a_l) \in \BZ_{>0} \times \BC^{\times} $ for $1\leq l \leq s$. The {\it Weyl module} $\WW^p(\prod_{l=1}^s \varpi_{m_l,a_l})$ is the $U_q(\Gaff_p)$-module generated by a highest $\ell$-weight vector $w$ of even parity such that for $2\leq j \leq p+1$,
\begin{align*}
 s_{11|p}(z) w &= w \prod_{l=1}^s \frac{q^{m_l}-za_lq^{-p-m_l}}{1-za_lq^{-p}} = t_{11|p}(z) w, \\
 h_p(\overline{s}_{11|p}(z)) w &= w \prod_{l=1}^s \frac{q^{-m_l}-za_lq^{p-m_l}}{1-za_lq^{p-2m_l}} = h_p(\overline{t}_{11|p}(z)) w, \\
 s_{jj|p}(z) w & = t_{jj|p}(z) w = h_p(\overline{s}_{jj|p}(z)) w = h_p(\overline{t}_{jj|p}(z)) w = w,
\end{align*}
and the following vector-valued polynomials in $z$ are of degree $\leq s$:
$$ \prod_{l=1}^s (1-za_lq^{-p}) \times s_{j1|p}(z) w,\quad \prod_{l=1}^s (1-za_lq^{p-2m_l}) \times h_p(\overline{s}_{j1|p}(z)) w. $$
Let $L^p(\prod_{l=1}^s \varpi_{m_l,a_l})$ denote its irreducible quotient of $\WW^p(\prod_{l=1}^s \varpi_{m_l,a_l})$.
\end{defi}

\begin{example} \label{example: Weyl gl(1|p)}
Let $1\leq p \leq N$. In Example \ref{example: Weyl gl(1|1)}, let $W^p$ be the sub-$U_q(\Gaff_p)$-module of $\vartheta_p^*(\otimes_{l=2}^s W_{m_l,a_l}^{(i_l)})$ generated by $\otimes_{l=2}^sv_l$. Then $\vartheta_p^{\bullet}(W^p)$ is a quotient of the Weyl module 
$\WW^p\left(\prod_{l=2}^s \varpi_{m_l,a_lq^{M-N-i_l+p}} \right)$ over $U_q(\Gaff_p)$.
\end{example}

\begin{example} \label{example: reduction to gl(1|p)}
Suppose $m_1 \leq N$ and take $p = m_1$. In $W_{m_1,a_1}^{(i_1)}$ there is a non-zero vector $v_1^{1}$ whose $\ell$-weight corresponds to the tableau $T_1^1 \in \CB_-(m_1\varpi_{i_1})$ such that: $T_1^1(-i_1,-j) = 1$ for $1\leq j \leq m_1$ and $T_1^1(-i,-j) = N+M-j+1$ for $1\leq i < i_1$ and $1\leq j \leq m_1$. Let $X$ be the sub-$U_q(\Gaff_{m_1})$-module of $\vartheta_{m_1}^*(W_{m_1,a_1}^{(i_1)})$ generated by $v_1^1$. By comparing the character formulas in Remark \ref{rem: BKK} we see that the $U_q(\Gaff_{m_1})$-module $\vartheta_{m_1}^{\bullet}(X) $  is irreducible and in terms of evaluation modules: 
\begin{align*}
 \vartheta_{m_1}^{\bullet}(X) &\cong V_q^+(m_1\epsilon_1+\sum_{j=1}^{m_1} (i_1-1)\epsilon_{j+1}; a_1q^{M-N-i_1})  \\
 &\simeq V_q^+((m_1+i_1-1)\epsilon_1;a_1q^{M-N+i_1-2}) \cong V_q^-((m_1+i_1-1)\epsilon_1;a_1q^{M-N-i_1}) \\
 &\cong L^{m_1}(\varpi_{m_1+i_1-1, a_1q^{M-N+i_1+m_1-2}}).
\end{align*}
Let $v_1^2 $ be a lowest $\ell$-weight vector of the $U_q(\Gaff_{m_1})$-module $\vartheta_{m_1}^{\bullet}(X)$. Then $v_1^2$ corresponds to the tableau $T_1^2 \in \CB_-(m_1\varpi_{i_1})$ such that $T_1^2(-i,-j) = N+M-j+1$ for $1\leq i \leq i_1$ and $´\leq j \leq m_1$; it is a lowest $\ell$-weight vector of the $U_q(\Gaff)$-module $W_{m_1,a_1}^{(i_1)}$.
Notice that $s_{ij}^{(n)} X = 0$ if $2\leq j \leq M+N-m_1$. Combining with Example \ref{example: Weyl gl(1|p)}, we observe that  $X \otimes W^{m_1}$ is stable by $\vartheta_{m_1}(U_q(\Gaff_{m_1}))$ and the identity map is an isomorphism of $U_q(\Gaff_{m_1})$-modules $\vartheta_{m_1}^{\bullet}(X\otimes W^{m_1}) \cong \vartheta_{m_1}^{\bullet}(X) \otimes \vartheta_{m_1}^{\bullet}(W^{m_1})$.
\end{example}
\begin{lem} \label{lem: cyclicity M=1}
Let $p, s \in \BZ_{>0}$ and let $(m_l,a_l) \in \BZ_{>0} \times \BC^{\times}$ for $1\leq l \leq s$. Assume $m_1 \geq p$. The $U_q(\Gaff_p)$-module $L^p(\varpi_{m_1,a_1}) \otimes \WW^p(\prod_{l=2}^s \varpi_{m_l,a_l})$ is of highest $\ell$-weight if $a_1 \neq a_l q^{2t-2m_l-2}$ for $2\leq l \leq s$ and $1\leq t \leq p$.
\end{lem}
\begin{proof}
By induction on $p$: for $p = 1$ we are led to consider the tensor product 
\begin{align*}
\WW_1\left(\frac{q^{m_1}-za_1q^{-1-m_1}}{1-za_1q^{-1}}; 1-za_1q^{-1-2m_1}\right) \otimes \\
\WW_1\left(\prod_{l=2}^s \frac{q^{m_l}-za_lq^{-1-m_l}}{1-za_lq^{-p}}; \prod_{l=2}^s (1-za_lq^{-1-2m_l})\right)
\end{align*}
of Weyl modules over $U_1 = U_q(\Gaff_1)$, which is of highest $\ell$-weight if $a_1 \neq a_lq^{-2m_l}$ for $2\leq l\leq s$. Assume therefore $p > 1$. In Equation  \eqref{equ: auxiliary diagram subalg} let us take $(p, M, N)$ to be $(p-1, 1, p)$. This defines a superalgebra morphism  
$$ \vartheta_{p-1}: U_q(\Gaff_{p-1}) \longrightarrow U_q(\Gaff_p),\quad s_{ij|p-1}^{(n)} \mapsto s_{i'j'|p}^{(n)},\quad t_{ij|p-1}^{(n)} \mapsto t_{i'j'|p}^{(n)}   $$ 
where $1' = 1$ and $i' = i +1$ for $1< i \leq p$.
Let $v_1, w$ be highest $\ell$-weight vectors of the $U_q(\Gaff_p)$-modules $L^p(\varpi_{m_1,a_1})$  and $\WW^p(\prod_{l=2}^s\varpi_{m_l,a_l})$ respectively. Set 
$$X_1 := \vartheta_{p-1}(U_q(\Gaff_{p-1})) v_1,\quad Y_1 := \vartheta_{p-1}(U_q(\Gaff_{p-1})) w. $$
Using evaluation modules over $U_q(\Gaff_p)$ we have by Corollary \ref{cor: KR evaluation} and Definition \ref{def: Weyl gl(1|p)}:
$$L^p(\varpi_{m_1,a_1}) \cong V_q^+(m_1\epsilon_1; a_1q^{-p}) \cong V_q^-(m_1\epsilon_1;a_1q^{p-2m_1}). $$
It follows that $s_{i2|p}^{(n)} X_1 = 0$ if $i \neq 2$. This implies that $X_1\otimes Y_1$ is stable by $\vartheta_{p-1}(U_q(\Gaff_{p-1}))$ and the identity map is an isomorphism of $U_q(\Gaff_{p-1})$-modules:
$$ \vartheta_{p-1}^{\bullet} (X_1\otimes Y_1) \cong \vartheta_{p-1}^{\bullet}(X_1) \otimes \vartheta_{p-1}^{\bullet}(Y_1). $$
As in Example \ref{example: reduction to gl(1|p)}, the $U_q(\Gaff_{p-1})$-module $\vartheta_{p-1}^{\bullet}(X_1)$ is irreducible and isomorphic to $L^{p-1}(\varpi_{m_1,a_1q^{-1}})$. By Definition \ref{def: Weyl gl(1|p)}, $\vartheta_{p-1}^{\bullet}(Y_1)$ is a quotient of the Weyl module
$\WW^{p-1}(\prod_{l=2}^s \varpi_{m_l,a_lq^{-1}})$. The induction hypothesis applied to $p-1$ shows that $L^{p-1}(\varpi_{m_1,a_1q^{-1}}) \otimes  \WW^{p-1}(\prod_{l=2}^s \varpi_{m_l,a_lq^{-1}})$ and so $\vartheta_{p-1}^{\bullet}(X_1) \otimes \vartheta_{p-1}^{\bullet}(Y_1)$
are of highest $\ell$-weight. Let $v_1'$ be the lowest $\ell$-weight vector of the $U_q(\Gaff_{p-1})$-module $\vartheta_{p-1}^{\bullet}(X_1)$; it corresponds to the tableau $T \in \CB_-(m_1\epsilon_1)$ such that $T(-1,-j) = p+2-j$ for $1\leq j \leq p-1$ and $T(-1,-j) = 1$ for $p\leq j \leq m_1$. We have
$$(*)\quad\quad v_1' \otimes w \in \vartheta_{p-1}(U_q(\Gaff_{p-1}))(v_1\otimes w) = X_1 \otimes Y_1. $$

Notice that $s_{ij;2}^{(n)} \mapsto h_p(\overline{s}_{ij|p}^{(n)})$ and $t_{ij;2}^{(n)} \mapsto h_p(\overline{t}_{ij|p}^{(n)})$ extend uniquely to a superalgebra morphism $\iota: U_2 \longrightarrow U_q(\Gaff_p)$. Let $X_2 := \iota(U_2) v_1'$ and $Y_2 := \iota(U_2) w$. The identification $L^p(\varpi_{m_1,a_1})\cong V_q^-(m_1\epsilon_1;a_1q^{p-2m_1})$ gives $X_2 := \BC v_1' + \BC v_1''$ where $v_1''$ is a lowest $\ell$-weight vector of $L^p(\varpi_{m_1,a_1})$. This implies $h_p(\overline{s}_{ij|p}^{(n)}) X_2 = 0$ if $i \notin \{1,2\}$, meaning that $X_2\otimes Y_2$ is stable by $\iota(U_2)$ and the graded permutation map is an isomorphism of $U_2$-modules $\iota^{\bullet}(X_2\otimes Y_2) \cong \iota^{\bullet}(Y_2) \otimes \iota^{\bullet}(X_2)$. By Definition \ref{def: Weyl gl(1|p)} the tensor product $\iota^{\bullet}(Y_2) \otimes \iota^{\bullet}(X_2)$ of $U_2$-modules is a quotient of
\begin{align*}
\WW_2\left(\prod_{l=2}^s \frac{q^{-m_l}-za_lq^{p-m_l}}{1-za_lq^{p-2m_l}}; \prod_{l=2}^s (1-za_lq^p) \right) \otimes \\
\WW_2\left(\frac{q^{-m_1+p-1}-za_1q^{-m_1+1}}{1-za_1q^{p-2m_1}}; 1-za_1q^{2-p} \right),
\end{align*}
which is of highest $\ell$-weight since $a_l q^{p-2m_l} \neq a_1q^{2-p}$ for $2\leq l \leq s$. The $U_2$-module $\iota^{\bullet}(X_2\otimes Y_2) $ is of highest $\ell$-weight and $ v_1'' \otimes w \in \iota(U_2)(v_1'\otimes w)$, which together with $(*)$ implies $v_1''\otimes w \in U_q(\Gaff_p)(v_1\otimes w)$. The $U_q(\Gaff_p)$-module $L^p(\varpi_{m_1,a_1})$ being generated by the lowest $\ell$-weight vector $v_1''$, we conclude by Lemma \ref{lem: Chari}.
\end{proof}
For $\mathfrak{gl}(1|3)$ we related the highest/lowest $\ell$-weight vectors of $L^3(\varpi_{5,a})$ by:
$$v_1= \young(11111) \xrightarrow{\vartheta_2: (134)q} v_1' = \young(11134) \xrightarrow{\iota: (12)q^{-1}} v_1''= \young(11234). $$

\noindent {\bf Proof of Theorem \ref{thm: tensor KR even}.} Let us assume first that $m_l \leq N$ for all $1\leq l \leq s$. 
We use a double induction on $(M, s)$ with Lemma \ref{lem: cyclicity M=1} being the initial case $M=1$. Under Condition \eqref{cond: cyclicity 1}, the induction hypothesis on $M$ applied to the tensor product of KR modules over $U_3$ in Example \ref{example: Weyl gl(1|1)} shows that $\iota_3^{\bullet}(W_3)$ is of highest $\ell$-weight and $ v_1^1 \otimes (\otimes_{l=2}^s v_l) \in \iota_3(U_3) (\otimes_{l=1}^s v_l)$.
It suffices to prove that the $U_q(\Gaff_{m_1})$-module $\vartheta_{m_1}^{\bullet}(X) \otimes \vartheta_{m_1}^{\bullet}(W^{m_1})$ in Example \ref{example: reduction to gl(1|p)} is of highest $\ell$-weight, from which follows $ v_1^2 \otimes (\otimes_{l=2}^s v_l) \in \vartheta_{m_1}(U_q(\Gaff_{m_1}))\iota_3(U_3) (\otimes_{l=1}^s v_l)$.
The $U_q(\Gaff)$-module $W_{m_1,a_1}^{(i_1)}$ being generated by the lowest $\ell$-weight vector $v_1^2$, we can use the second induction on $s$ and Lemma \ref{lem: Chari} to conclude. 
 
 By Examples \ref{example: Weyl gl(1|p)} and \ref{example: reduction to gl(1|p)}, $\vartheta_{m_1}^{\bullet}(X) \otimes \vartheta_{m_1}^{\bullet}(W^{m_1})$ is, up to tensor product by one-dimensional modules, a quotient of the $U_q(\Gaff_{m_1})$-module
$$L^{m_1}(\varpi_{m_1+i_1-1, a_1q^{M-N+i_1+m_1-2}}) \otimes \WW^{m_1}\left(\prod_{l=1}^s \varpi_{m_l,a_lq^{M-N-i_l+m_1}} \right), $$
which by Lemma \ref{lem: cyclicity M=1} is of highest $\ell$-weight if for $2\leq l \leq s$ and $1\leq t \leq m_1$:
$$ a_1q^{M-N+i_1+m_1-2} \neq a_lq^{M-N-i_l+m_1} \times q^{2t-2-2m_l},$$
namely, $a_1 \neq a_l q^{2t-2m_l-i_1-i_l}$. This is included in Condition \eqref{cond: cyclicity 1}.

Suppose $m_l > N$ for some $1\leq l \leq s$. Let $m := \max(m_l: 1\leq l \leq s)$ and let $U_4 := U_q(\widehat{\mathfrak{gl}(M|N+m)})$ be the quantum affine superalgebra with RTT generators $s_{ij;4}^{(n)}, t_{ij;4}^{(n)}$ for $1\leq i,j \leq M+N+m$. There is a unique superalgebra morphism
$$ \iota_4: U_q(\Gaff)  \longrightarrow U_4, \quad s_{ij}^{(n)} \mapsto s_{ij;4}^{(n)},\quad t_{ij}^{(n)} \mapsto t_{ij;4}^{(n)}. $$
Under Condition \eqref{cond: cyclicity 1}, the tensor product $\otimes_{l=1}^s W_{m_l,a_l}^{4(i_l)}$ of KR modules over $U_4$ is of highest $\ell$-weight. For $1\leq l \leq s$, let $X_l := \iota_4(U_q(\Gaff)) v_l$ where $v_l \in W_{m_l,a_l}^{4(i_l)}$ is a highest $\ell$-weight vector. Then a weight argument and Corollary \ref{cor: KR evaluation} show that 
$$ \iota_4(U_q(\Gaff)) (\otimes_{l=1}^s v_l) = \otimes_{l=1}^s X_l, $$
and as $U_q(\Gaff)$-modules $\iota_4^{\bullet}(\otimes_{l=1}^s X_l) \cong \otimes_{l=1}^s W_{m_l,a_lq^{-m}}^{(i_l)}$. This implies that the $U_q(\Gaff)$-module $\otimes_{l=1}^s W_{m_l,a_lq^{-m}}^{(i_l)}$ is of highest $\ell$-weight, proving the theorem. \qed

For $\mathfrak{gl}(3|6)$ we related the highest/lowest $\ell$-weight vectors of $W_{4,a}^{(3)}$ by:
$$ v_1 = \young(1111,2222,3333) \xrightarrow{\iota_3:(23456789)q} v_1^1 = \young(1111,6789,6789) \xrightarrow{\vartheta_4:(16789)q} v_1^2 = \young(6789,6789,6789). $$

For $\lambda \in \mathcal{P}$ and $a \in \BC^{\times}$ define the $U_{q^{-1}}(\Gaff)$-module $V_{q^{-1}}^{+}(\lambda;a)$ to be the pullback of the $U_{q^{-1}}(\Glie)$-module $V_{q^{-1}}(\lambda)$ by $\overline{\ev}_a^+$, as in Theorem \ref{thm: q-char MA}. By Equation \eqref{homo:evaluation dec}, 
$$  h^* \left( V_q^-(\lambda;a) \right) \cong V_{q^{-1}}^+(\lambda;a).  $$
\begin{cor} \label{cor: tensor KR even involution}
The tensor product in Theorem \ref{thm: tensor KR even} is of highest $\ell$-weight if 
\begin{equation} \label{cond: cyclicty 2}
\frac{a_j}{a_k} \notin \bigcup_{p=1}^{m_k} q^{2p-2m_k} \CS(i_j,i_k)\quad \mathrm{for}\ 1\leq j < k \leq s. 
\end{equation}
\end{cor}
\begin{proof}
The tensor product $T$ in Theorem \ref{thm: tensor KR even} is of highest $\ell$-weight if and only if so is the $U_{q^{-1}}(\Gaff)$-module $h^*(T)$. By Corollary \ref{cor: KR evaluation} we have 
$$ h^*(T) \cong \otimes_{l=s}^1 V_{q^{-1}}^+(m_l\varpi_{i_l}; a_lq^{N-M-2m_l+i_l}). $$
Applying Theorem \ref{thm: tensor KR even} to $U_{q^{-1}}(\Gaff)$, by viewing $W_{m,a}^{(i)}$ first as $V_q^+(m\varpi_i;aq^{M-N-i})$ and then as $V_{q^{-1}}^+(m\varpi_i;aq^{N-M+i})$, we have that $h^*(T)$ is of highest $\ell$-weight if
$$ \frac{a_k q^{-2m_k} }{a_j q^{-2m_j}}  \notin \bigcup_{p=1}^{m_k} q^{-2p+2m_j} \CS(i_k,i_j)^{-1} \quad \mathrm{for}\ 1\leq j < k \leq s. $$
This is Condition \eqref{cond: cyclicty 2} since $\CS(i_k,i_j) = \CS(i_j,i_k)$.
\end{proof}
 Let $V$ be a finite-dimensional $U_q(\Gaff)$-module. Its {\it twisted dual} is the dual space $\mathrm{Hom}_{\BC}(V,\BC) =: V^{\vee}$ endowed with the $U_q(\Gaff)$-module structure \cite[Section 6]{Z4}:
 $$\langle x \varphi, v \rangle := (-1)^{|\varphi||x|} \langle \varphi, \Sm \Psi(x) \rangle \quad \mathrm{for}\ x \in U_q(\Gaff),\ \varphi \in V^{\vee},\ v \in V. $$
By Equation \eqref{iso:transposition}, $(V\otimes W)^{\vee} \cong V^{\vee} \otimes W^{\vee}$ if $W$ is another finite-dimensional $U_q(\Gaff)$-module. $V$ is irreducible if and only if both $V$ and $V^{\vee}$ are of highest $\ell$-weight.

We recall the notion of {\it fundamental representations} from \cite{Z4}. Let $1\leq r \leq M$ and $1\leq s < N$. Define (compare \cite[Lemmas 5 \& 6]{Z4} with Corollary \ref{cor: KR evaluation})
\begin{equation} \label{def: fund rep}
V_{r,a}^+ := W_{1,aq^{N-M-r}}^{(r)},\quad V_{s,a}^- := W_{1,aq^{s+2}}^{(M+N-s)},\quad V_{N,a}^- := W_{1,aq^{N+2}}^{(M-)}.
\end{equation}
\begin{lem} \label{lem: twisted dual KR even}
Let $1\leq i \leq M <j<M+N$ and $(m,a)\in \BZ_{>0} \times \BC^{\times}$. We have:
$$(W_{m,a}^{(i)})^{\vee} \simeq W_{m,a^{-1}q^{2m}}^{(i)},\quad (W_{m,a}^{(j)})^{\vee} \simeq W_{m,a^{-1}q^{4-2m}}^{(j)},\quad (W_{m,a}^{(M-)})^{\vee} \simeq W_{m,a^{-1}q^{4-2m}}^{(M-)}.$$
\end{lem}
\begin{proof}
The twisted dual of a fundamental module is known \cite[Lemma 27]{Z4}:
$$  (V_{i,a}^+)^{\vee} \simeq V_{i,a^{-1}q^{2(M-N+i+1)}}^+,\quad (V_{M+N-j,a}^-)^{\vee} \simeq V_{M+N-j;a^{-1}q^{-2(M+N+1-j)}}^-. $$
By Equation \eqref{def: fund rep}, $(W_{1,a}^{(i)})^{\vee} \cong W_{1,a^{-1}q^2}^{(i)}$ and $(W_{1,a}^{(j)})^{\vee} \simeq W_{m,a^{-1}q^{2}}^{(j)}$. Viewing $W_{m,a}^{(i)}$ as the unique irreducible sub-quotient of $\otimes_{l=1}^m W_{1,aq_i^{2-2l}}^{(i)}$ of highest $\ell$-weight $\varpi_{m,a}^{(i)}$, and taking twisted duals, we obtain the desired formulas.
\end{proof}
\begin{cor} \label{cor: level 0 Demazure}
Let $1< i < M,\ a \in \BC^{\times}$ and $m\in \BZ_{>0}$. The $U_q(\Gaff)$-module $W_{m,a}^{(i-1)} \otimes W_{m,a}^{(i+1)}$ is irreducible.
\end{cor}
\begin{proof}
The tensor product and its twisted dual, which is $\simeq W_{m,a^{-1}q^{2m}}^{(i-1)} \otimes W_{m,a^{-1}q^{2m}}^{(i+1)}$ by Lemma \ref{lem: twisted dual KR even}, satisfy Condition \eqref{cond: cyclicty 2} and are of highest $\ell$-weight.
\end{proof}
The following special result on Dynkin node $M$ is needed in Section \ref{sec: asym}. 
\begin{lem} \label{lem: tensor fund} \cite{Z4}
For $m \in \BZ_{>0}$, the $U_q(\Gaff)$-module $V_{N,aq^{-3}}^- \otimes (\otimes_{l=1}^{m} V_{N-1,aq^{2l-1}}^-) \otimes (\otimes_{k=1}^{m}V_{M-1,aq^{2M-2k-1}}^+)$  is of highest $\ell$-weight.
Moreover for $1\leq k,l \leq m$ we have $ V_{N-1,aq^{2l-1}}^- \otimes  V_{M-1,aq^{2M-2k-1}}^+ \cong V_{M-1,aq^{2M-2k-1}}^+ \otimes V_{N-1,aq^{2l-1}}^-$.
\end{lem}
\begin{proof}
The first statement is induced from \cite[Theorem 15]{Z4} by the involution $h$ as in \cite[Remarks 3 \& 4]{Z4}, and the second is a particular case of \cite[Example 5]{Z4}.
\end{proof}

\section{Asymptotic representations} \label{sec: asym}
In this section we construct the $U_q(\Gaff)$-module $\CN_{c,a}^{(i)}$ of Proposition \ref{prop: asymptotic T} for $i \in I_0$ and $a,c \in \BC^{\times}$ from finite-dimensional representations.

For $m \in \BZ_{>0}$,  let $N_{m,a}^{(i)} := L(\Bn_{q^m,a}^{(i)})$ be the irreducible $U_q(\Gaff)$-module; it is finite-dimensional by Lemma \ref{lem: simple O} (3). Fix $v^m \in N_{m,a}^{(i)}$ to be a highest $\ell$-weight vector.

The main step is to construct an inductive system $(N_{m,a}^{(i)})_{m\in \BZ_{>0}}$ compatible with (normalized) $q$-characters, as in \cite[Section 4.2]{HJ} and \cite[Theorem 7.6]{HL}. We shall need the cyclicity results in Section \ref{sec: cyclicity} to adapt the arguments of \cite{HJ,HL}.
\begin{lem} \label{lem: asym N}
If $\Bn_{q^m,a}^{(i)} \Bm \in \lwt(N_{m,a}^{(i)})$, then $\Bn_{i,a}^- \Bm \in \lwt(N_{i,a}^-)$ and 
$$ \dim (N_{m,a}^{(i)})_{\Bn_{q^m,a}^{(i)} \Bm} \leq \dim (N_{i,a}^-)_{\Bn_{i,a}^- \Bm}. $$
\end{lem}
\begin{proof}
The first paragraph of the proof of \cite[Theorem 7.6]{HL} can be copied here, based on Lemma \ref{lem: positive pre char} and the fact that $\Bm$ is a product of the $A_{j,b}^{-1}$ with $j \in I_0$ and $b \in a q^{\BZ}$. For the latter fact, we realize $N_{m,a}^{(i)}$ as a tensor product of KR modules with one-dimensional modules and apply Corollary \ref{cor: KR evaluation}.
\end{proof}
\begin{lem} \label{lem: simple N}
Let $c \in \BC^{\times}$ be such that $c^2 \notin q^{\BZ}$. If $\Bn_{i,a}^- \Bm \in \lwt(N_{i,a}^-)$, then $\Bn_{c,a}^{(i)} \Bm \in \lwt(L(\Bn_{c,a}^{(i)}))$ and $\dim (N_{i,a}^-)_{\Bn_{i,a}^- \Bm} \leq \dim L(\Bn_{c,a}^{(i)})_{\Bn_{c,a}^{(i)} \Bm}$.
\end{lem}
\begin{proof}
From Example \ref{ss: l-weights} we obtain 
$$\Bn_{i,a}^- \equiv \Bn_{c,a}^{(i)} \prod_{j\sim i} \Psi_{j,aq_{ij}c_{ij}^2}^{-1}. $$
Viewing $N_{i,a}^-$ as a sub-quotient of $L(\Bn_{c,a}^{(i)}) \otimes (\otimes_{j\sim i} L_{j,aq_{ij}c_{ij}^2}^-) \otimes D$ with $D$ being a one-dimensional $Y_q(\Glie)$-module, we have $\Bm = \Bm' \prod_{j\sim i} \Bm^j$ with 
$$\Bn_{c,a}^{(i)} \Bm' \in \lwt(L(\Bn_{c,a}^{(i)})), \quad  \Psi_{j,aq_{ij}c_{ij}^2}^{-1}\Bm^j \in \lwt(L_{j,aq_{ij}c_{ij}^2}^-) \quad \mathrm{for}\ j \sim i. $$
By Corollary \ref{cor: TQ negative} and Lemmas \ref{lem: positive pre char}--\ref{lem: negative pre char} we have:
\begin{itemize} 
\item[(1)] $\Bm, \Bm' \in \lCQ^- q^{\BQ^-}$ and $\Bm$ is a monomial in the $A_{i',b}^{-1}$ with $i' \in I_0$ and $b \in aq^{\BZ}$;
\item[(2)] $\Bm^j$ is a monomial in the $A_{i',b'}^{-1}$ with $i' \in I_0$ and $b' \in \{ac^2,ac^{-2}\} q^{\BZ}$ for $j \sim i$.
\end{itemize}
Since $\{ac^2,ac^{-2}\} q^{\BZ}$ and $a q^{\BZ}$ do not intersect, $\Bm^j = 1$ and $\Bm' = \Bm$.
\end{proof}
%Note that similar $\ell$-weight arguments can be used to prove 
%$$ \nqc(L_{M,a}^-) = \lim\limits_{m\rightarrow \infty} \nqc(W_{m,aq^{-1}}^{(M)}) = \lim\limits_{m\rightarrow \infty} \nqc(W_{m,aq}^{(M-)}). $$

For $m_1, m_2 \in \BZ_{>0}$ with $m_1 < m_2$, let $Z_{i,a}^{m_1,m_2}$ be the irreducible $U_q(\Gaff)$-module of highest $\ell$-weight $\Bn_{q^{m_2},a}^{(i)} (\Bn_{q^{m_1},a}^{(i)})^{-1} = \prod_{j\sim i} \aBw_{q_{ij}^{m_1-m_2},aq_{ij}^{1+2m_1}}^{(j)}$; by Lemma \ref{lem: simple O} (3) it is finite-dimensional. Fix $v^{m_1,m_2} \in Z_{i,a}^{m_1,m_2}$ to be a highest $\ell$-weight vector.
\begin{lem}  \label{lem: cyclicity}
The $U_q(\Gaff)$-module $N_{m_1,a}^{(i)} \otimes Z_{i,a}^{m_1,m_2} \otimes Z_{i,a}^{m_2,m_3}$ is of highest $\ell$-weight for $0 < m_1 < m_2 < m_3$.
\end{lem}
\begin{proof}
We shall assume $1\leq i \leq M$. The case $M+1<i<M+N$ can be deduced from $1\leq i < M$ using $\SG^*$. (See typical arguments in the proof of Lemma \ref{lem: tensor product T1}.) 

Suppose $1\leq i < M$. By Corollary \ref{cor: level 0 Demazure}, $Z_{i,a}^{m_1,m_2} \simeq \otimes_{j\sim i} W_{m_2-m_1,aq^{-2m_1-2}}^{(j)}$.
The tensor product $W_{1,aq}^{(i)} \otimes (\otimes_{j\sim i}W_{m_1,aq^{-2}}^{(j)})$ satisfies Condition \eqref{cond: cyclicty 2} and is of highest $\ell$-weight. Its irreducible quotient is $\simeq N_{m_1,a}^{(i)}$. Next,
$$W_{1,aq}^{(i)} \otimes (\otimes_{j\sim i}W_{m_1,aq^{-2}}^{(j)} ) \otimes  (\otimes_{j\sim i} W_{m_2-m_1,aq^{-2m_1-2}}^{(j)}) \otimes (\otimes_{j\sim i} W_{m_3-m_2,aq^{-2m_2-2}}^{(j)}) $$
also satisfies Condition \eqref{cond: cyclicty 2}, and is of highest $\ell$-weight, implying that $N_{m_1,a}^{(i)} \otimes Z_{i,a}^{m_1,m_2} \otimes Z_{i,a}^{m_2,m_3}$ is of highest $\ell$-weight.

Suppose $i = M$. Consider the tensor product of fundamental modules:
$$ T := V_{N,aq^{-N-3}}^- \otimes (\otimes_{l=1}^{m_3} V_{N-1,aq^{-N+2l-1}}^-) \otimes (\otimes_{k=1}^{m_3}V_{M-1,aq^{-N+2M-2k-1}}^+). $$
By Lemma \ref{lem: tensor fund}, $T$ is of highest $\ell$-weight and 
\begin{align*}
T \cong &\ V_{N,aq^{-N-3}}^- \otimes (\otimes_{l=1}^{m_1} V_{N-1,aq^{-N+2l-1}}^-) \otimes (\otimes_{k=1}^{m_1}V_{M-1,aq^{-N+2M-2k-1}}^+) \otimes \\
 & (\otimes_{l=m_1+1}^{m_2} V_{N-1,aq^{-N+2l-1}}^-) \otimes (\otimes_{k=m_1+1}^{m_2}V_{M-1,aq^{-N+2M-2k-1}}^+) \\
 & (\otimes_{l=m_2+1}^{m_3} V_{N-1,aq^{-N+2l-1}}^-) \otimes (\otimes_{k=m_2+1}^{m_3}V_{M-1,aq^{-N+2M-2k-1}}^+). 
\end{align*}
Let $T_1,T_2,T_3$ denote the above tensor products of the first, second, and third row at the right-hand side. They are of highest $\ell$-weight. By Equation \eqref{def: fund rep}, 
$$ V_{N,aq^{-N-3}}^- \simeq W_{1,aq^{-1}}^{(M-)},\quad V_{N-1,aq^{-N+1}}^- \simeq W_{1,aq^2}^{(M+1)},\quad V_{M-1,aq^{-N+2M-3}}^- \simeq W_{1,aq^{-2}}^{(M-1)}. $$
By Example \ref{ss: l-weights}, the irreducible quotients of $T_1, T_2, T_3$ are $\simeq N_{m_1,a}^{(M)}, Z_{M,a}^{m_1,m_2}$ and $  Z_{M,a}^{m_2,m_3}$, proving the cyclicity statement.
\end{proof}
Let $0 < m_1 < m_2$. The tensor product $N_{m_1,a}^{(i)} \otimes Z_{i,a}^{m_1,m_2}$ being of highest $\ell$-weight, its irreducible quotient is isomorphic to $N_{m_2,a}^{(i)}$. There exists a unique morphism of $U_q(\Gaff)$-modules $\CF_{m_2,m_1}: N_{m_1,a}^{(i)} \otimes Z_{i,a}^{m_1,m_2} \longrightarrow N_{m_2,a}^{(i)}$ which sends $v^{m_1} \otimes v^{m_1,m_2}$ to $v^{m_2}$. As in \cite[Section 4.2]{HJ}, define 
$$ F_{m_2,m_1}: N_{m_1,a}^{(i)} \longrightarrow N_{m_2,a}^{(i)},\quad w \mapsto \CF_{m_2,m_1}(w \otimes v^{m_1,m_2}). $$
Then $(\{N_{m,a}^{(i)}\}, \{F_{m_2,m_1}\})$ constitutes an inductive system of vector superspaces: $F_{m_3,m_2}F_{m_2,m_1} = F_{m_3,m_1}$ for $0 < m_1 < m_2 < m_3$. The proof is the same as that of \cite[Proposition 4.1 (2)]{Z5}, based on Lemma \ref{lem: cyclicity}. 

\begin{lem} \label{lem: injective inductive system}
Let $0 < m_1 < m_2$. We have $F_{m_2,m_1} x_{j,n}^+ = x_{j,n}^+ F_{m_2,m_1}$ for $j \in I_0$ and $n \in \BZ$. The linear map $F_{m_2,m_1}$ is injective.
\end{lem}
\begin{proof}
This is \cite[Theorem 3.15]{HJ}. For a proof independent of $\ell$-weights, we refer to the first two paragraphs of the proof of \cite[Proposition 4.3]{Z5}; the coproduct $\Delta(e_i^+)$ therein should be replaced by the $\Delta(x_{j,n}^+)$ in Equation \eqref{coproduct: positive}. 
\end{proof}
\begin{lem} \label{lem: F and phi}
Let us write $(h_1(z),h_2(z),\cdots,h_{\kappa}(z);\even) :=  \Bn_{q^{m_2},a}^{(i)} (\Bn_{q^{m_1},a}^{(i)})^{-1}  \in \BR_U$ for $m_2 > m_1 > 0$. Then for $j \in I_0$ we have
$$K_j^{\pm}(z) F_{m_2,m_1} =  h_j(z) \times F_{m_2,m_1} K_j^{\pm}(z) \in \mathrm{Hom}_{\BC}(N_{m_1,a}^{(i)}, N_{m_2,a}^{(i)})[[z^{\pm 1}]]. $$
Here for $\pm$ we take Taylor expansions of $h_j(z)$ at $z = 0, z = \infty$ respectively.
\end{lem}
\begin{proof}
The same as \cite[Proposition 4.2]{HJ} in view of Equation \eqref{coproduct: affine Cartan}. 
\end{proof}
All the $h_j(z) \in \BC[[z]]$ are of the form $A(z) q^{-m_2} + B(z) + C(z) q^{m_2}$ where $A(z), B(z), C(z) \in \BC[[z]]$ are independent of $m_2$. Let $j \in I_0$. If $j \sim i$, then 
$$\phi_j^{\pm}(z) F_{m_2,m_1} = q_{ij}^{m_1-m_2} \frac{1-zaq_{ij}^{1+2m_2}}{1-zaq_{ij}^{1+2m_1}} \times F_{m_2,m_1} \phi_j^{\pm}(z).$$
Otherwise, $F_{m_2,m_1}$ commutes with $\phi_j^{\pm}(z)$ for $|j-i| \neq 1$.

From Lemmas \ref{lem: asym N} and \ref{lem: injective inductive system}--\ref{lem: F and phi}, we conclude that: the normalized $q$-characters $\nqc(N_{m,a}^{(i)})$ for $m \in \BZ_{>0}$ are polynomials in $\BZ[A_{j,b}^{-1}]_{(j,b) \in I_0\times aq^{\BZ}}$, and as $m \rightarrow \infty$ they converge to a formal power series $\lim\limits_{m\rightarrow \infty} \nqc(N_{m,a}^{(i)}) \in \BZ[[A_{j,b}^{-1}]]_{(j,b) \in I_0\times aq^{\BZ}}$, which is bounded above by the normalized $q$-character $\nqc(N_{i,a}^-)$.

\begin{lem} \label{lem: F negative}
For $j \in I_0$ and $m_2-1 > m > 0$ we have
\begin{align*}
x_{j,0}^- F_{m_2,m} &= F_{m_2,m} x_{j,0}^- \quad \mathrm{if}\ |j-i| \neq 1, \\
x_{j,0}^- F_{m_2,m} &= F_{m_2,m+1} (q^{m_2} A_{j,m} + q^{-m_2} C_{j,m}) \quad \mathrm{if}\ |j-i| = 1.
\end{align*}
 Here $A_{j,m}, C_{j,m}: N_{m,a}^{(i)} \longrightarrow N_{m+1,a}^{(i)}$ are linear maps of parity $|\alpha_j|$.
\end{lem}
\begin{proof} 
This corresponds to \cite[Lemma 4.4 \& Proposition 4.5]{HJ}. Here we give a straightforward proof without induction arguments.

By Lemma \ref{lem: cyclicity}, the $U_q(\Gaff)$-module $Z_{i,a}^{m,m+1} \otimes Z_{i,a}^{m+1,m_2}$ is of highest $\ell$-weight with irreducible quotient $Z_{i,a}^{m,m_2}$; let $\CG_{m_2,m}$ be the quotient map sending $v^{m,m+1} \otimes v^{m+1,m_2}$ to $v^{m,m_2}$. We claim that for $v \in N_{m,a}^{(i)},\ v' \in Z_{i,a}^{m,m+1}$ and $j \in I_0$:
\begin{itemize}
\item[(i)] $ \CF_{m_2,m}(v \otimes \CG_{m_2,m}(v' \otimes v^{m+1,m_2})) = F_{m_2,m+1}\CF_{m+1,m}(v\otimes v')$; 
\item[(ii)] $ x_{j,0}^- v^{m,m_2} = [(m_2-m)\delta_{|j-i|,1}]_q \times \CG_{m_2,m}(x_{j,0}^-v^{m,m+1} \otimes v^{m+1,m_2})$.
\end{itemize}
Here $[n]_q := \frac{q^n-q^{-n}}{q-q^{-1}}$ for $n \in \BZ$. Assume the claim for the moment.
For $v \in N_{m,a}^{(i)}$, based on $\Delta(x_{j,0}^-) = 1\otimes x_{j,0}^- + x_{j,0}^- \otimes \phi_{j,0}^-$ we compute $x_{j,0}^- F_{m_2,m}(v)$ 
\begin{align*}
 &= x_{j,0}^- \CF_{m_2,m}(v\otimes v^{m,m_2}) = \CF_{m_2,m}\Delta(x_{j,0}^-)(v\otimes v^{m,m_2}) \\
 &= \CF_{m_2,m}(x_{j,0}^-v \otimes \phi_{j,0}^- v^{m,m_2}) + (-1)^{|v||\alpha_j|} \CF_{m_2,m}(v\otimes x_{j,0}^- v^{m,m_2}) \\
&= q_{ij}^{(m_2-m) \delta_{|j-i|,1}} F_{m_2,m}(x_{j,0}^- v) +   \\
&\quad\quad (-1)^{|v||\alpha_j|}   [(m_2-m)\delta_{|j-i|,1}]_q  \CF_{m_2,m}(v \otimes \CG_{m_2,m}( x_{j,0}^-v^{m,m+1} \otimes v^{m+1,m_2}) ) \\
&= q_{ij}^{(m_2-m) \delta_{|j-i|,1}} F_{m_2,m}(x_{j,0}^- v) + \\
&\quad\quad (-1)^{|v||\alpha_j|}   [(m_2-m)\delta_{|j-i|,1}]_q  F_{m_2,m+1}\CF_{m+1,m}(v \otimes  x_{j,0}^-v^{m,m+1}),
\end{align*}
which proves the lemma. The third and fourth identities used (ii) and (i).

Note that $\CF_{m_2,m}(1_{N_{m,a}^{(i)}} \otimes \CG_{m_2,m})$ and $\CF_{m_2,m+1} (\CF_{m+1,m} \otimes 1_{Z_{i,a}^{m+1,m_2}})$, as $U_q(\Gaff)$-linear maps from the highest $\ell$-weight module $N_{m,a}^{(i)} \otimes Z_{i,a}^{m,m+1} \otimes Z_{i,a}^{m+1,m_2}$ to $N_{m_2,a}^{(i)}$, are identical because they both send the highest $\ell$-weight vector $v^m \otimes v^{m,m+1} \otimes v^{m+1,m_2}$ to $v^{m_2}$. Applying them to $v \otimes v' \otimes v^{m+1,m_2}$ gives (i).

From the proof of Lemma \ref{lem: cyclicity} it follows that $Z_{i,a}^{m,m_2}$ is $\simeq$ irreducible quotient of a tensor product of KR modules associated to $j'\in I_0$ with $j' \sim i$. Let $\mu$ be the weight of $v^{m,m_2}$. If $|j-i| \neq 1$, then by Lemma \ref{lem: KR l-weights refined}, $\mu q^{-\alpha_j} \notin \wt(Z_{i,a}^{m,m_2})$ and $x_{j,0}^- v^{m,m_2} = 0$. Suppose $j \sim i$. Then $ (Z_{i,a}^{m,m_2})_{\mu q^{-\alpha_j}} = \BC x_{j,0}^- v^{m,m_2}$ and $q_{ij} = q^{\pm 1}$.
The equation $x_{j,0}^-\CG_{m_2,m} = \CG_{m_2,m} x_{j,0}^-$ applied to $v^{m,m+1} \otimes v^{m+1,m_2}$ gives
$$x_{j,0}^- v^{m,m_2} =\CG_{m_2,m}(q_{ij}^{m_2-m-1} x_{j,0}^- v^{m,m+1} \otimes v^{m+1,m_2} + v^{m,m+1} \otimes x_{j,0}^- v^{m+1,m_2} ). $$
Consider the following vector in $Z_{i,a}^{m,m+1} \otimes Z_{i,a}^{m+1,m_2}$ of weight $\mu q^{-\alpha_j}$:
$$ w := q_{ij}^{-1}\frac{q_{ij}^{m+1-m_2}-q_{ij}^{m_2-m-1}}{q_{ij}^{-1}-q_{ij}} x_{j,0}^- v^{m,m+1} \otimes v^{m+1,m_2} - v^{m,m+1} \otimes x_{j,0}^- v^{m+1,m_2}. $$
We have $\CG_{m_2,m}(w) \in \BC x_{j,0}^- v^{m,m_2}$ and $x_{j,0}^+ w = 0$. So $x_{j,0}^+ \CG_{m_2,m}(w) = 0$. Now $x_{j,0}^+ x_{j,0}^- v^{m,m_2} \neq 0$ forces $\CG_{m_2,m}(w) = 0$. We express $x_{j,0}^- v^{m,m_2}$
\begin{align*}
&=\CG_{m_2,m}(q_{ij}^{m_2-m-1} x_{j,0}^- v^{m,m+1} \otimes v^{m+1,m_2} + v^{m,m+1} \otimes x_{j,0}^- v^{m+1,m_2})  + \CG_{m_2,m}(w) \\
&= \left(q_{ij}^{m_2-m-1} + \frac{q_{ij}^{m-m_2}-q_{ij}^{m_2-m-2}}{q_{ij}^{-1}-q_{ij}}\right) \times  \CG_{m_2,m}(x_{j,0}^- v^{m,m+1} \otimes v^{m+1,m_2}) \\
&= [m_2-m]_{q_{ij}} \times  \CG_{m_2,m}(x_{j,0}^- v^{m,m+1} \otimes v^{m+1,m_2}),
\end{align*}
which proves (ii) because $[n]_{q_{ij}} = [n]_q$ for $n \in \BZ$. 
\end{proof}

\noindent {\bf Proof of Proposition \ref{prop: asymptotic T}. }
For $r \in \BZ_{\geq 0}$ and $l \in I$ let $K_{l,\pm r}^{\pm}$ be the coefficient of $z^{\pm r}$ in $K_l^{\pm}(z) \in U_q(\Gaff)[[z^{\pm 1}]]$. The superalgebra $U_q(\Gaff)$ is generated by:
$$ \mathcal{S} :=  \{K_{l,\pm r}^{\pm},\ x_{j,0}^-,\ x_{j,n}^+ \ |\ r \in \BZ_{\geq 0},\ n \in \BZ,\ j \in I_0,\ l \in I \}. $$
By Lemmas \ref{lem: injective inductive system}--\ref{lem: F negative},  there are $ \mathrm{Hom}_{\BC}(N_{m,a}^{(i)}, N_{m+1,a}^{(i)})$-valued Laurent polynomials $P_{s;m}(u)$ for $m \in \BZ_{>0}$ and $s \in \mathcal{S}$ such that
 $$s F_{m_2,m} = F_{m_2,m+1} P_{s;m}(q^{m_2}) \in \mathrm{Hom}_{\BC}(N_{m,a}^{(i)}, N_{m_2,a}^{(i)}) \quad \mathrm{for}\ m_2 > m+1.$$ 
These polynomials have non-zero coefficients only at $u, 1, u^{-1}$. Since $q$ is not a root of unity, the generic asymptotic construction of \cite[Section 2]{Z5} can be applied to the inductive system $(\{N_{m,a}^{(i)}\}, \{F_{m_2,m_1}\})$. Let $N_{\infty}$ be its inductive limit. Fix $c \in \BC^{\times}$. There exists a unique representation of $U_q(\Gaff)$ on $N_{\infty}$ on which $s \in \mathcal{S}$ acts as 
$$ \lim\limits_{m\rightarrow \infty} P_{s;m}(c) \in \mathrm{End}(N_{\infty})   $$
Here the $P_{s;m}(c): N_{m,a}^{(i)} \longrightarrow N_{m+1,a}^{(i)}$ for $m \in \BZ_{>0}$ form a morphism of the inductive system, so their inductive limit $\lim\limits_{m\rightarrow \infty} P_{s;m}(c)$ makes sense. As in the proof of \cite[Lemma 6.7]{Z5}, the resulting $U_q(\Gaff)$-module $\CN_{c,a}^{(i)}$ is in category $\BGG$ with $q$-character
$$ \chi_q(\CN_{c,a}^{(i)}) = \Bn_{c,a}^{(i)} \times \lim\limits_{m\rightarrow \infty} \nqc(N_{m,a}^{(i)}).  $$
Let us prove $\lim\limits_{m\rightarrow \infty} \nqc(N_{m,a}^{(i)}) = \nqc(N_{i,a}^-)$. We have seen above Lemma \ref{lem: F negative} that the left-hand side is bounded above by the right-hand side. View $L(\Bn_{c,a}^{(i)})$ as a sub-quotient of $\CN_{c,a}^{(i)}$.  If $c^2 \notin q^{\BZ}$, then by Lemma \ref{lem: simple N}, $\nqc(N_{i,a}^-)$ is bounded above by $\nqc(L(\Bn_{c,a}^{(i)}))$ and so by $(\Bn_{c,a}^{(i)})^{-1} \chi_q(\CN_{c,a}^{(i)})$, which is the left-hand side. This implies the reverse inequality and the irreducibility of $\CN_{c,a}^{(i)}$ for $c^2 \notin q^{\BZ}$. \qed

One can have asymptotic modules $\mathcal{M}_{c,a}^{(i)}$ over $U_q(\Gaff)$ as limits of $M_{q^m,a}^{(i)}$ (as in \cite[Section 7.2]{HJ}, which is slightly different from the limit construction of $\CN_{c,a}^{(i)}$). Then Equation \eqref{equ: positive TQ asy} holds with $M$ replaced by $\mathcal{M}$ for all $c,d \in \BC^{\times}$.

\begin{prop}  \label{prop: simplicity tensor product asymp}
The $U_q(\Gaff)$-module $\CW_{c_1,a_1}^{(i_1)} \otimes \CW_{c_2,a_2}^{(i_2)} \otimes \cdots \otimes \CW_{c_s,a_s}^{(i_s)}$, with $i_l \in I_0$ and $c_l,a_l \in \BC^{\times}$, is irreducible if $a_lc_l^{-2} \notin a_k q^{\BZ}$ for all $1\leq l,k \leq s$. 
\end{prop}
\begin{proof}
Let $L := \otimes_{l=1}^s L_{i_l,a_l}^-$ and $S = L(\prod_{l=1}^s \aBw_{c_l,a_l}^{(i_l)})$, viewed as irreducible $Y_q(\Glie)$-modules by Corollary \ref{cor: simplicity tensor product pre-fund}. $S$ is a sub-quotient of the tensor product $T$ in the proposition. Let $\aBw, \aBw'$ be the highest $\ell$-weights of $L, S$ respectively. Then $\chi_q(T) = \aBw' \nqc(L)$ by Lemma \ref{lem: negative pre char}. It suffices to prove that $\dim L_{\Bn \aBw} \leq \dim S_{\Bn \aBw'}$ for all $\Bn \aBw \in \lwt(L)$. Viewing $L$ as a sub-quotient of $S \otimes D$ where $D \simeq \otimes_{l=1}^s L_{i_l,a_lc_l^{-2}}^-$, we can adapt the proof of Lemma \ref{lem: simple N} to the present situation.
\end{proof}
It follows that the tensor products of the $\CW$ at the right-hand side of Equations \eqref{equ: negative TQ asy}--\eqref{equ: positive TQ asy} are irreducible $U_q(\Gaff)$-modules for $c^2,d^2 \notin q^{\BZ}$.
\section{Proof of extended T-systems: Theorem \ref{thm: Demazure T}}  \label{sec: proof T}
The idea is to provide lower and upper bounds for $\dim (D_{m,a}^{(i,s)})$. We recall from the proof of Corollary \ref{cor: Demazure l-weights} that the $U_q(\Gaff)$-module $W_{m,aq_i^{2m+1}}^{(i)} \otimes W_{m+s,aq_i^{2m-1}}^{(i)}$ has at least two irreducible sub-quotients: $L(\varpi_{m+s+1,aq_i^{2m+1}}^{(i)} \varpi_{m-1,aq_i^{2m-1}}^{(i)})$ and $D_{m,a}^{(i,s)}$.

\begin{lem}
For $i\in I_0\setminus \{M\}$, the $U_q(\Gaff)$-module $W_{m+s,aq_i^{2m+1}}^{(i)} \otimes D_{m,a}^{(i,s)}$ has at least two sub-quotients: $L(\Bd_{m,a}^{(i,s-1)} \varpi_{m+s+1,aq_i^{2m+1}}^{(i)})$ and $L(\Bd_{m+s,aq_i^{-2s}}^{(i,0)} \varpi_{m,aq_i^{2m+1}}^{(i)})$.
\end{lem}
\begin{proof}
Set $ T := W_{m+s,aq_i^{2m+1}}^{(i)} \otimes D_{m,a}^{(i,s)}$ and $S := L(\Bd_{m,a}^{(i,s-1)} \varpi_{m+s+1,aq_i^{2m+1}}^{(i)})$. By Example \ref{ss: l-weights}, $S$ is an irreducible sub-quotient of $T$. By Corollary \ref{cor: Demazure l-weights},
$$ \Bm' := \Bm \prod_{l=1}^s A_{i,aq_i^{2-2l}}^{-1} =  \Bd_{m+s,aq_i^{-2s}}^{(i,0)} \varpi_{m,aq_i^{2m+1}}^{(i)} \in \lwt(T).$$
Viewing $S$ as an irreducible sub-quotient of $W_{m+s+1,aq_i^{2m+1}}^{(i)} \otimes D_{m,a}^{(i,s-1)}$ and using Lemma \ref{lem: KR l-weights refined} and Corollary \ref{cor: Demazure l-weights}, we have $\Bm' \notin \lwt(S)$. Let $\mu := (3m+2s)\varpi_i - m \alpha_i$ so that
$\varpi(\Bm) = q^{\mu}$ and $\varpi(\Bm') = q^{\mu-s\alpha_i}$. Then
$\dim T_{q^{\mu-t\alpha_i}} = t+1$ for $0 \leq t \leq s$.

Let $v_0 \in S$ be a highest $\ell$-weight vector and let $U_i$ be the subalgebra in the proof of Corollary \ref{cor: Demazure l-weights}. Then $U_i v_0$ is an irreducible $U_{q_i}(\widehat{\mathfrak{sl}_2})$-module of highest $\ell$-weight 
$$ \Bm_i := (Y_{aq_i^{-1}} Y_{i,aq_i^{-3}} \cdots Y_{i,aq_i^{1-2s}} ) (Y_{i,aq_i^{2m+1}} Y_{i,aq_i^{2m-1}} \cdots Y_{i,aq_i^{3-2s}})$$
and factorizes as $L^i(Y_{aq_i^{-1}} Y_{i,aq_i^{-3}} \cdots Y_{i,aq_i^{3-2s}}) \otimes L^i(Y_{i,aq_i^{2m+1}} Y_{i,aq_i^{2m-1}} \cdots Y_{i,aq_i^{1-2s}})$; if $s = 1$ then the first tensor factor is trivial. For $1\leq t \leq s$, the weight space $S_{q^{\mu-t\alpha_i}}$ is spanned by the $x_{i,n_1}^-x_{i,n_2}^-\cdots x_{i,n_t}^- v_0 \in U_i v_0$ with $n_l \in \BZ$ for $1\leq l \leq t$ and is therefore of dimension $\min(s,t+1)$. Since $\Bm_i \prod_{l=1}^s (Y_{aq_i^{1-2l}}Y_{aq_i^{3-2l}})^{-1}$ is not an $\ell$-weight of $L^i(\Bm_i)$, we must have $\Bm' \notin \lwt(S)$, as in the proof of Corollary \ref{cor: Demazure l-weights}.

It follows that $\chi_q(T)-\chi_q(S)$ is $\Bm'$ plus terms of the form $\Bm'' \in \BR$ with $\varpi(\Bm'') \notin \varpi(\Bm') q^{\BQ^+}$, forcing $L(\Bm')$ to be an irreducible sub-quotient of $T$.
\end{proof}

\begin{lem}  \label{lem: tensor product T1}
Let $i \in I_0 \setminus \{M\}$. The $U_q(\Gaff)$-modules $W_{m,aq_i^{2m+1}}^{(i)} \otimes W_{m+s,aq_i^{2m-1}}^{(i)}$ and $W_{m+s,aq_i^{2m+1}}^{(i)} \otimes D_{m,a}^{(i,s)}$ are of highest $\ell$-weight, while $W_{m+s+1,aq_i^{2m+1}}^{(i)} \otimes W_{m-1,aq_i^{2m-1}}^{(i)}$, $D_{m+s,aq_i^{-2s}}^{(i,0)} \otimes W_{m,aq_i^{2m+1}}^{(i)}$ and $W_{m+s+1,aq_i^{2m+1}}^{(i)} \otimes D_{m,a}^{(i,s-1)}$ are irreducible.
\end{lem}
\begin{proof}
 Assume $i < M$. Notice that $T_{m,a}^{(i,s)} := W_{m,aq^{2m}}^{(i+1)} \otimes W_{m,aq^{2m}}^{(i-1)} \otimes W_{s,aq^{-1}}^{(i)}$ satisfies Condition \eqref{cond: cyclicty 2} and is of highest $\ell$-weight. By Remark \ref{rem: comparison with HJ}, the irreducible quotient of $T_{m,a}^{(i,s)}$ is $\simeq D_{m,a}^{(i,s)}$. To prove that the five tensor products in the lemma are of highest $\ell$-weight, we can replace $D$ by $T$ and show that the resulting tensor products of KR modules satisfy Condition \eqref{cond: cyclicty 2}. For example the last tensor product corresponds to $W_{m+s+1,aq^{2m+1}}^{(i)} \otimes W_{m,aq^{2m}}^{(i+1)} \otimes W_{m,aq^{2m}}^{(i-1)} \otimes W_{s-1,aq^{-1}}^{(i)}$.
 
 Next, $S_{m,a}^{(i,s)} := W_{s,a^{-1}q^{2s+1}}^{(i)} \otimes W_{m,a^{-1}}^{(i-1)} \otimes W_{m,a^{-1}}^{(i+1)}$ also satisfies Condition \eqref{cond: cyclicty 2} and is of highest $\ell$-weight, the irreducible quotient of which is $\simeq (T_{m,a}^{(i,s)})^{\vee}$. To establish the irreducibility of the last three tensor products in the lemma, we take twisted duals as in Lemma \ref{lem: twisted dual KR even}, replace $D^{\vee}$ by $S$, and check Condition \eqref{cond: cyclicty 2} for the resulting tensor products of KR modules. Take the fourth as an example: $W_{m+s,a^{-1}q^{2s}}^{(i-1)} \otimes W_{m+s,a^{-1}q^{2s}}^{(i+1)} \otimes W_{m,a^{-1}q^{-1}}^{(i)}$ is of highest $\ell$-weight. 
 
 This proves the lemma in the case $i < M$. 
 
 Assume $i > M$. By Lemma \ref{lem: duality by permutation}, $\SG^*(W_{m,a}^{(i)}) \simeq W_{m,aq^{N-M-2+2m}}'^{(M+N-i)}$ as $U_q(\Gafft)$-modules. Applying $\SG^{*-1}$ to the $U_q(\Gafft)$-modules $T_{m,a}'^{(M+N-i,s)}, S_{m,a}'^{(M+N-i,s)}$  we obtain that $D_{m,a}^{(i,s)}$ and $(D_{m,a}^{(i,s)})^{\vee}$ are $\simeq$ the irreducible quotients of the highest $\ell$-weight modules 
  $$W_{m,aq^{-2m}}^{(i+1)} \otimes W_{m,aq^{-2m}}^{(i-1)*} \otimes W_{s,aq}^{(i)}, \quad W_{s,a^{-1}q^{3-2s}}^{(i)} \otimes W_{m,a^{-1}q^4}^{(i-1)*} \otimes W_{m,a^{-1}q^4}^{(i+1)}$$
  respectively. Here $W_{m,a}^{(M)*} := W_{m,a}^{(M-)}$ and $W_{m,a}^{(j)*} = W_{m,a}^{(j)}$ for $j > M$. By replacing $D,D^{\vee}$ with these tensor products, we obtain eight tensor products of KR modules $W_{m,b}^{(j)}, W_{m,b}^{(M-)}$ with $j > M$ and need to show that they are of highest $\ell$-weight. Applying $\SG^*$ gives tensor products of KR modules $W_{m,b}'^{(j)}$ with $j \leq M$ over $U_q(\Gafft)$, which are shown to satisfy Condition \eqref{cond: cyclicity 1}. Consider the last tensor product in the lemma as an example. Let us prove that the $U_q(\Gaff)$-modules
  \begin{align*}
  T_1 &:=  W_{m+s+1,aq^{-2m-1}}^{(i)} \otimes W_{m,aq^{-2m}}^{(i+1)} \otimes W_{m,aq^{-2m}}^{(i-1)*} \otimes W_{s-1,aq}^{(i)}, \\
  T_2 &:=  W_{m+s+1,a^{-1}q^{3-2s}}^{(i)} \otimes W_{s-1,a^{-1}q^{5-2s}}^{(i)} \otimes W_{m,a^{-1}q^4}^{(i-1)*} \otimes W_{m,a^{-1}q^4}^{(i+1)}
  \end{align*}
are of highest $\ell$-weight. Applying $\SG^*$ to $T_1,T_2$ give ($c = q^{N-M-2}, j = M+N-i$):
   \begin{align*}
   T_1' &= W_{s-1,acq^{2s-1}}'^{(j)} \otimes W_{m,ac}'^{(j+1)} \otimes W_{m,ac}'^{(j-1)} \otimes W_{m+s+1,ac^{2s+1}}'^{(j)}, \\
   T_2' &= W_{m,a^{-1}cq^{2m+4}}'^{(j-1)}  \otimes W_{m,a^{-1}cq^{2m+4}}'^{(j+1)} \otimes W_{s-1,a^{-1}cq^3}'^{(j)} \otimes W_{m+s+1,a^{-1}cq^{2m+5}}'^{(j)}.
   \end{align*}
 The $U_q(\Gafft)$-modules $T_1',T_2'$ satisfy Condition \eqref{cond: cyclicity 1}.  
\end{proof}
For $i \in I_0$ and $m \in \BZ_{>0}$ let $d_m^{(i)} := \dim (W_{m,a}^{(i)})$; it is independent of $a \in \BC^{\times}$ because $\Phi_a^*(W_{m,1}^{(i)}) \cong W_{m,a}^{(i)}$ by Equation \eqref{iso:Z-grading}. 

\begin{theorem} \label{thm: Tsuboi}  \cite{Tsuboi}
$(d_m^{(i)})^2 = d_{m+1}^{(i)}d_{m-1}^{(i)} + d_m^{(i-1)}d_m^{(i+1)}$ for $1\leq i < M$.
\end{theorem}
\begin{proof}
For $\mu \in \mathcal{P}$, up to normalization $\mathcal{T}_{\emptyset\subset \mu}(u)$ in \cite[(2.15)]{Tsuboi} can be identified with $\chi_q(V_q^-(\mu;a))$ in Equation \eqref{for: dMA}. The dimension identity is a consequence of \cite[(3.2)]{Tsuboi}, which in turn comes from Jacobi identity of determinants. 
\end{proof}
%Applying $\BD$, we see that the above theorem is true for $M < i < M+N$ and for $i = M+1$ with $d_m^{(i-1)}$ replaced by the dimension of $W_{m,1}^{(M-)}$.

\noindent {\bf Proof of Theorem \ref{thm: Demazure T}.} By Lemma \ref{lem: tensor product T1}, the surjective morphisms of $U_q(\Gaff)$-modules in Theorem \ref{thm: Demazure T} exist (because the third terms are irreducible quotients of the second terms) and their kernels admit irreducible sub-quotients $D_{m,a}^{(i,s)}$ and $D_{m+s,aq_i^{-2s}}^{(i,0)} \otimes W_{m,aq_i^{2m+1}}^{(i)}$ respectively. This gives:
\begin{itemize}
\item[(1)] $\dim (D_{m,a}^{(i,s)}) \leq d_m^{(i)} d_{m+s}^{(i)} - d_{m+s+1}^{(i)}d_{m-1}^{(i)}$; 
\item[(2)] $\dim(D_{m+s,aq_i^{-2s}}^{(i,0)}) d_m^{(i)} \leq d_{m+s}^{(i)} \dim(D_{m,a}^{(i,s)}) - d_{m+s+1}^{(i)} \dim(D_{m,a}^{(i,s-1)})$.
\end{itemize}
We prove the equality in (1)--(2) by induction on $s$. Suppose $s = 0$; (2) is trivial. If $i < M$, then by Example \ref{ss: l-weights} and Corollary \ref{cor: level 0 Demazure},
$$ D_{m,a}^{(i,0)} \simeq W_{m,aq^{2m}}^{(i+1)} \otimes W_{m,aq^{2m}}^{(i-1)}. $$
This together with Theorem \ref{thm: Tsuboi} shows that equality holds in (1). Making use of $\SG^*$, we can remove the assumption $i < M$, as in the proof of Lemma \ref{lem: tensor product T1}.

Suppose $s > 0$. In (2) the induction hypothesis applied to $0, s-1$ indicates that 
\begin{align*}
((d_{m+s}^{(i)})^2 -d_{m+s+1}^{(i)}d_{m+s-1}^{(i)})d_m^{(i)} \leq &\ d_{m+s}^{(i)} \dim(D_{m,a}^{(i,s)}) \\
& - d_{m+s+1}^{(i)}(d_m^{(i)} d_{m+s-1}^{(i)} - d_{m+s}^{(i)}d_{m-1}^{(i)});
\end{align*}
namely, $\dim(D_{m,a}^{(i,s)}) \geq d_m^{(i)} d_{m+s}^{(i)} - d_{m+s+1}^{(i)}d_{m-1}^{(i)}$. This implies that in (1), and henceforth in the above inequality and in (2), $\leq$ can be replaced by $=$. \qed
\begin{rem}
Let $1\leq i < M$. Apply $\SG^{*-1}$ to the second exact sequence in category $\BGG'$ of Theorem \ref{thm: Demazure T} involving $D_{m,a}'^{(M+N-i,1)}$ and take normalized $q$-characters:
\begin{multline*} 
\quad\nqc(N_{m,a}^{(i)}) \nqc(W_{m+1,aq^{-1}}^{(i)}) = \nqc(W_{m+2,aq}^{(i)})\prod_{j\in I_0: j\sim i}\nqc(W_{m,aq^{-2}}^{(j)})\\
 + A_{i,a}^{-1} \times \nqc(W_{m,aq^{-3}}^{(i)}) \prod_{j\in I_0: j \sim i}\nqc(W_{m+1,a}^{(j)}).\quad
\end{multline*}
Setting $m \rightarrow \infty$ recovers the normalized $q$-characters of Equation \eqref{equ: TQ negative}. The second exact sequence of Theorem \ref{thm: Demazure T} is likely to be true for $i = M$.
\end{rem}

Theorem \ref{thm: Demazure T} together with its proof could be adapted to quantum affine algebras, in view of the cyclicity results of \cite{Chari} and T-system \cite{Nakajima,H}. The second and third terms of the first exact sequence appeared in the proof of \cite[Theorem 4.1]{FoH} as $V', V$ by setting $(a, m, s) = (q_i^{-3}, m_2+1, m_1-m_2-2)$. In the context of graded representations of current algebras \cite[Theorem 2]{ChariDemazure} 
by taking $(\ell,\lambda) =(m+s, m \omega_i)$ so that $\nu = (2m+s)\omega_i-m\alpha_i$, the exact sequence therein is an injective resolution of the Demazure module $D(\ell,\nu)$ by fusion products of KR modules. It is natural to expect that $D_{m,1}^{(i,s)}$ admits a classical limit $(q=1)$ as $D(\ell,\nu)$; this is true when $m = s = 1$, as a particular case of \cite[Theorem 1]{Demazure}. 

\section{Transfer matrices and Baxter operators}  \label{sec: Baxter}
Let us fix an integer $\ell \in \BZ_{>0}$ (length of spin chain) and complex numbers $b_j \in \BC^{\times} \setminus q^{\BZ}$ for $1\leq j \leq \ell$ (inhomogeneity parameters). We shall construct an action of $K_0(\BGG)$ on the vector superspace $\BV^{\otimes \ell}$ as in \cite[Section 5]{FZ}. This is the XXZ spin chain with twisted periodic boundary condition, with $\BV^{\otimes \ell}$ referred to as the quantum space and objects of category $\BGG$ auxiliary spaces.

Following Definition \ref{def: category BGG}, let $\CE$ be the subset of $\lCE$ consisting of the $\sum_{\Bf \in \CP} c_{\Bf} \Bf \in \lCE$. Note that $\CE$ is a sub-ring and $\chi(W) \in \CE$ for $W$ in category $\BGG$.

We identify $\underline{i} = i_1i_2\cdots i_{\ell} \in I^{\ell}$, an $I$-string of length $\ell$, with the basis vector $v_{i_1} \otimes v_{i_2} \otimes \cdots \otimes v_{i_{\ell}}$ of $\BV^{\otimes \ell}$. Let $E_{\underline{i}\underline{j}} \in \End(\BV^{\otimes \ell})$ be the elementary matrix $\underline{k} \mapsto \delta_{\underline{j}\underline{k}} \underline{i}$ for $\underline{i}, \underline{j} \in I^{\ell}$, and let $\epsilon_{\underline{i}} := \epsilon_{i_1} + \epsilon_{i_2} + \cdots + \epsilon_{i_{\ell}} \in \BP$. 

To a $Y_q(\Glie)$-module $W$ in category $\BGG$ is by definition attached an matrix $S^W(z)$, a power series in $z$ with values in $\End(W) \otimes \End(\BV)$. We decompose 
$$S^W_{1,\ell+1}(zb_{\ell}) \cdots S^W_{13}(zb_2)S^W_{12}(zb_1) = \sum_{\underline{i},\underline{j} \in I^{\ell}} S_{\underline{i}\underline{j}}^W(z) \otimes E_{\underline{i}\underline{j}} \in \End(W) \otimes \End(\BV)^{\otimes \ell}[[z]].  $$
Then $S_{\underline{i}\underline{j}}^W(z) = \pm s_{i_{\ell}j_{\ell}}^W(zb_{\ell}) \cdots s_{i_2j_2}(zb_2)  s_{i_1j_1}^W(zb_1)$ and it sends one weight space $W_p$ for $p \in \CP$ to another of weight $p q^{\epsilon_{\underline{i}} - \epsilon_{\underline{j}}}$. Its trace over $W_p$ is well-defined: either $0$ if $\epsilon_{\underline{i}} \neq  \epsilon_{\underline{j}}$; or the usual non-graded trace of $S_{\underline{i}\underline{j}}^W(z)|_{W_p} \in \End(W_p)$ if $\epsilon_{\underline{i}} =  \epsilon_{\underline{j}}$ . 
\begin{defi} \label{def: transfer matrix}
Let $W$ be in category $\BGG$. Its associated {\it transfer matrix} is 
$$ t_W(z) := \sum_{\underline{i},\underline{j} \in I^{\ell}} \left( \sum_{p \in \wt(W)} p \times \Tr_{W_p} (S_{\underline{i}\underline{j}}^W(z) ) \right) E_{\underline{i}\underline{j}}, $$
viewed as a power series in $z$ with values in $\End(\BV^{\otimes \ell}) \otimes_{\BZ} \CE$.
\end{defi}
In \cite{BT, Tsuboi2} (for $U_q(\Gaff)$) and \cite{FH} (for an arbitrary non-twisted quantum affine algebra), transfer matrices are partial traces of universal R-matrices $\mathcal{R}(z)$. Since the existence of $\mathcal{R}(z)$ for $U_q(\Gaff)$ is not clear to the author (except the simplest case $\mathfrak{gl}(1|1)$ in \cite{Z2}), we use a different transfer matrix based on RTT. One should imagine $S^W(z)$ as the specialization of $\mathcal{R}(z)$ at $W \otimes \BV$.

As in \cite{FH}, the transfer matrix $t_W(z)$ is a twisted trace of $S^W(z)$ due to the presence of $p \in \wt(W)$. In \cite{BT,Tsuboi2} $p$ is related to an auxiliary field. 

\begin{example} \label{examp: one-dim transfer}
Consider the one-dimensional module $\BC_{\Bf}$ in Example \ref{example one-dim}:
$$ t_{\BC^{\Bf}}(z) \underline{i} = \underline{i} \times p \times \prod_{l=1}^{\ell} h(zb_l) p_{i_l}\quad \mathrm{for}\ \underline{i} \in I^{\ell}.   $$
\end{example}
\begin{prop} \label{prop: transfer matrices}
For $X,Y$ in category $\BGG$ and $a \in \BC^{\times}$, we have: 
$$t_{\Phi_a^*X}(z) = t_X(za),\quad t_X(z) t_Y(z) = t_{X\otimes Y}(z),\quad t_X(z)t_Y(w) = t_Y(w)t_X(z). $$
\end{prop}
\begin{proof}
We mainly prove the second equation; the first one is almost clear from Definition \ref{def: transfer matrix} and Equation \eqref{iso:Z-grading}, and the third one in the same way as \cite[Theorem 5.3]{FH} based on the commutativity of $K_0(\BGG)$.  For $\underline{i}, \underline{j} \in I^{\ell}$:
\begin{align*}
&\quad S_{\underline{i}\underline{j}}^{X\otimes Y}(z) \otimes E_{\underline{i}\underline{j}} = \prod_{r=\ell}^{1} s_{i_rj_r}^{X\otimes Y}(zb_r) \otimes E_{i_1j_1} \otimes E_{i_2j_2} \otimes \cdots \otimes E_{i_{\ell}j_{\ell}} \\
&= \sum_{\underline{k} \in I^{\ell}} \prod_{r=\ell}^1 \left((-1)^{|E_{i_rk_r}||E_{k_rj_r}|} s_{i_rk_r}^X(zb_r)\otimes s_{k_ri_r}^Y(zb_r) \right)\otimes E_{i_1j_1} \otimes E_{i_2j_2} \otimes \cdots \otimes E_{i_{\ell}j_{\ell}} \\
&=  \sum_{\underline{k} \in I^{\ell}} (S_{\underline{i}\underline{k}}^{X}(z) \otimes 1 \otimes  E_{\underline{i}\underline{k}})(1\otimes S_{\underline{k}\underline{j}}^{Y}(z) \otimes E_{\underline{k}\underline{j}}).
\end{align*}
After taking trace over $X_p \otimes Y_{p'}$, only the terms with $\epsilon_{\underline{i}} = \epsilon_{\underline{k}} = \epsilon_{\underline{j}}$ survive and so all the tensor components are of even parity, implying the second equation.
\end{proof}
Let $\varphi: \CP \longrightarrow \BC^{\times}$ be a morphism of multiplicative groups (typical examples are $((p_i)_{i\in I};s) \mapsto (-1)^s$ and $((p_i)_{i\in I};s) \mapsto (-1)^s \times \prod_{i\in I} p_i$). If $W$ is a finite-dimensional $Y_q(\Glie)$-module in category $\BGG$, then the {\it twisted transfer matrix}  is:
\begin{equation} \label{for: twisted transfer matrix}
t_W(z;\varphi) := \sum_{\underline{i},\underline{j} \in I^{\ell}} \left( \sum_{p \in \wt(W)} \varphi(p) \times \Tr_{W_p} (S_{\underline{i}\underline{j}}^W(z) ) \right) E_{\underline{i}\underline{j}} \in \End(\BV^{\otimes \ell})[[z]]. 
\end{equation}
If $W$ is infinite-dimensional and the second summation above converges (for a generic choice of $\varphi$), then $t_W(z;\varphi)$ is still well-defined.

\begin{lem}  \label{lem: polyn asym}
Let $i \in I_0,\ a,c \in \BC^{\times}$. The power series $f_{c,a}^{(i)}(z) s_{jk}(z) \in Y_q(\Glie)[[z]]$ for $j,k \in I$ act on the module $\CW_{c,a}^{(i)}$ as  polynomials in $z$ of degree $\leq 1$, where 
\begin{gather*}  
\begin{tabular}{|c|c|c|} 
  \hline
  &$i \leq M$ & $i > M$  \\
  \hline
 $f_{c,a}^{(i)} $ & $1-zaq^{M-N-i-1} $ & $\frac{(1-zac^{-2}q^{M+N-i-1})(1-zaq^{i-M-N-1})}{1-zaq^{M+N-i-1}} $ \\
  \hline
\end{tabular} 
\end{gather*}
\end{lem}
\begin{proof}
Let us recall  the generic limit construction of $\CW_{c,a}^{(i)}$ in \cite{Z5}.
For $m > 0$ set $V_m := W_{m,aq_i^{-1}}^{(i)} \otimes \BC_{(1,\cdots,1;m|\varpi_i|)}$, so that its highest $\ell$-weights is of even parity.  Let $\mathcal{T} := \{s_{ij}^{(n)},t_{ij}^{(n)}\}$ be the set of RTT generators for $U_q(\Gaff)$. By \cite[Lemma 5.1]{Z5}, their exists an inductive system of vector superspaces $(\{V_m\}, \{F_{m_2,m_1}\})$ with Laurent polynomials $Q_{t;m}(u) \in \mathrm{Hom}_{\BC}(V_m,V_{m+1})[u,u^{-1}]$ for $t \in \mathcal{T}$ and $m > 0$ such that 
$$ (\blacktriangle)\quad t F_{m_2,m} = F_{m_2,m+1} Q_{t;m}(q_i^{m_2}) \in \mathrm{Hom}_{\BC}(V_m,V_{m_2}) \quad \mathrm{for}\ m_2 > m+1. $$ 
Its inductive limit admits a $U_q(\Gaff)$-module structure where $t \in \mathcal{T}$ acts as the inductive limit $\lim\limits_{m \rightarrow \infty} Q_{t;m}(c)$.  This is exactly the module $\CW_{c,a}^{(i)}$.

%If $i \leq M$, then by Equation \eqref{equ: KR +}, $ V_m  \cong V_q^+(m\varpi_i; aq^{M-N-i-1})$.
%It follows from Equation \eqref{homo: evaluation} that $(1-zaq^{M-N-i-1})s_{jk}^{V_m}(z)$ is a polynomial in $z$ of degree $\leq 1$. By Equation $(*)$ above, for $m_2 > m+1$ we have
%$$(1-zaq^{M-N-i-1})s_{jk}(z) F_{m_2,m} = F_{m_2,m+1}(1-zaq^{M-N-i-1}) \sum_{n\geq 0} z^n Q_{m,s_{jk}^{(n)}}(q^{m_2}). $$
%From the injectivity of $F_{m_2,m+1}$, and from the assumption that $q$ is not a root of unity, we obtain that $(1-zaq^{M-N-i-1}) \sum_{n\geq 0} z^n Q_{m,s_{jk}^{(n)}}(u)$ is a polynomial in $z$ of degree $\leq 1$. By taking the inductive limit $m \rightarrow \infty$, the same holds for the action of $(1-zaq^{M-N-i-1}) s_{jk}(z)$ on $W$.

Suppose $i > M$. By comparing the highest $\ell$-weights of the modules in Equation \eqref{equ: KR -} based on \eqref{for: dual iMA}, \eqref{for: dual dMA} and Lemma \ref{lem: JC dual action}, we have:
\begin{align*}
W_{m,aq}^{(i)} &\cong V_q^{-*}(\lambda_m^{(i)}; aq^{M+N-1-i})  \cong \phi_{h_m(z)}^*\left( V_q^{+*}(\lambda_m^{(i)};aq^{i-M-N+2m-1})\right), \\ 
h_m(z) &= \prod_{l=1}^m \prod_{j=1}^{M+N-i} \frac{(1-zaq^{2l-2j+M+N-i-1})^2}{(1-zaq^{2l-2j+M+N-i-3})(1-zaq^{2l-2j+M+N-i+1})} \\
&= \frac{(1-zaq^{2m+i-M-N-1})(1-zaq^{-i+M+N-1})}{(1-zaq^{2m-i+M+N-1})(1-zaq^{i-M-N-1})}. 
\end{align*}
It follows that $h_{m_2}(z)^{-1} (1-zaq^{2m_2+i-M-N-1}) s_{jk}(z) F_{m_2,m}$ is a polynomial in $z$ of degree $\leq 1$ for all $m_2 > m$.  By Equation $(\blacktriangle)$ above, this is equal to 
$$F_{m_2,m+1} \frac{(1-zaq^{2m_2+i-M-N-1})}{h_{m_2}(z)} \sum_{n\geq 0} z^n Q_{s_{jk}^{(n)};m}(q^{-m_2}). $$
Since $h_{m_2}(z)^{-1} (1-zaq^{2m_2+i-M-N-1}) = f_{q^{-m_2},a}^{(i)} (z)$, from the injectivity of $F_{m_2,m+1}$ and the polynomial dependence on $q^{m_2}$, we obtain that $f_{c,a}^{(i)}(z) \sum_{n\geq 0} z^n Q_{s_{jk}^{(n)};m}(c)$ is a polynomial in $z$ of degree $\leq 1$. By taking the inductive limit $m \rightarrow \infty$, the same holds for the action of $f_{c,a}^{(i)}(z) s_{jk}(z)$ on $\CW_{c,a}^{(i)}$.

The case $i \leq M$ is much simpler, since $V_m \cong V_q^+(m\varpi_i; aq^{M-N-i-1})$. We omit the details.
\end{proof}
Based on the lemma, let us define the $Y_q(\Glie)$-module $\BW_{c,a}^{(i)} := \phi_{f_{c,a}^{(i)}(z)}^* (\CW_{c,a}^{(i)})$. (Indeed it can be equipped with a $U_q(\Gaff)$-module structure.) 
\begin{lem} \label{lem: separation of variables}
For $i \in I_0$ and $a,c \in \BC^{\times}$ we have:
\begin{equation}  \label{equ: sov}
[\BW_{c,1}^{(i)} \otimes \BW_{1,a^2}^{(i)}] = [\BW_{ca,a^2}^{(i)} \otimes \BW_{a^{-1},1}^{(i)}] \in K_0(\BGG).
\end{equation}
Let $X$ be a finite-dimensional $U_q(\Gaff)$-module in category $\BGG$. In a fractional ring of $K_0(\BGG)$ we have $[X] = \sum\limits_{l=1}^{\dim X} [D_l] \Bm_l$ where for each $l$, $D_l$ is a one-dimensional $U_q(\Gaff)$-module in category $\BGG$, and $\Bm_l$ is a product of the $\frac{[\BW_{b,a}^{(i)}]}{[\BW_{c,a}^{(i)}]}$ with $i \in I_0,\ a,b,c \in \BC^{\times}$.
\end{lem}
\begin{proof}
For the first statement, by Example \ref{ss: l-weights} and Lemma \ref{lem: polyn asym} we have:
$$ \aBw_{c,1}^{(i)} \aBw_{1,a^2}^{(i)} = \aBw_{ca,a^2}^{(i)}\aBw_{a^{-1},1}^{(i)},\quad f_{c,1}^{(i)}(z) f_{1,a^2}^{(i)}(z) = f_{ca,a^2}^{(i)}(z)f_{a^{-1},1}^{(i)}(z). $$
Together with Lemma \ref{lem: negative pre char}, this implies that the $q$-characters of the two tensor products in Equation \eqref{equ: sov} coincide. For the second statement, we argue as \cite[Theorem 4.8]{FH} based on 
$ \frac{\aBw_{b,a}^{(i)}}{\aBw_{c,a}^{(i)}} = \frac{\qc(\CW_{b,a}^{(i)})}{\qc(\CW_{c,a}^{(i)})} \equiv \frac{\qc(\BW_{b,a}^{(i)})}{\qc(\BW_{c,a}^{(i)})}$; see also \cite[Theorem 6.11]{Z5}. 
\end{proof}
Equation \eqref{equ: sov} is a {\it separation of variables} identity; see also \cite[Theorem 3.11]{FZ}. The same identity holds when replacing $\BW$ by $\CW$. Since $t_{\BW_{c,a}^{(i)}}(z)$ is a polynomial in $z$ of degree $\leq \ell$, the following definition makes sense.
\begin{defi} \label{def: Baxter operators}
For $i \in I_0$ the {\it Baxter operator} is $Q_i(z) := t_{\BW_{z,1}^{(i)}}(1)$.
\end{defi}

Let $p_c^{(i)} = \varpi(\aBw_{c,a}^{(i)})$. Then $\wt(\BW_{c,a}^{(i)}) \subset p_c^{(i)} q^{\BQ^-}$ and $\overline{Q}_i(z) := (p_c^{(i)})^{-1}Q_i(z)$ is a power series in the $q^{-\alpha_j}$ with $j \in I_0$ whose coefficients are in $\End(\BV^{\otimes \ell})[z,z^{-1}]$. Let $\overline{Q}_i^0(z)$ be its leading term. Since $(\BW_{1,1}^{(i)})_{p_1^{(i)}}$ is the one-dimensional simple socle of $\BW_{1,1}^{(i)}$, by Definition \ref{def: transfer matrix}, $\underline{i}$ is an eigenvector of $\overline{Q}_i^0(1)$ with non-zero eigenvalue. (Here we used the overall assumption $b_l \notin q^{\BZ}$.) The formal power series $\overline{Q}_i^0(z)$ and $Q_i(z)$ in the $q^{-\alpha_j}$ can therefore be inverted for $z \in \BC$ generic.

\begin{cor}[generalized Baxter TQ relations] \label{cor: generalized TQ relations}
For $b,c \in \BC^{\times}$, we have:
\begin{equation} \label{equ: generalized TQ}
\frac{t_{\BW_{b,1}^{(i)}}(z^{-2})}{t_{\BW_{c,1}^{(i)}}(z^{-2})} = \frac{Q_i(zb)}{Q_i(zc)},\quad \frac{t_{\CW_{b,1}^{(i)}}(z^{-2})}{t_{\CW_{c,1}^{(i)}}(z^{-2})} = \prod_{l=1}^{\ell} \frac{f_{c,1}^{(i)}(z^{-2}b_l^{-2})}{f_{b,1}^{(i)}(z^{-2}b_l^{-2})} \times \frac{Q_i(zb)}{Q_i(zc)}.
\end{equation} 
If $X$ is a finite-dimensional $U_q(\Gaff)$-module in category $\BGG$, then $t_X(z^{-2})$ is a sum of monomials in the $\frac{Q_i(zb)}{Q_i(zc)} t_D(z^{-2})$ with $i \in I_0,\ b,c \in \BC^{\times}$ and with $D$ one-dimensional $U_q(\Gaff)$-modules in category $\BGG$, the number of terms being $\dim X$.
\end{cor}
\begin{proof}
In Equation \eqref{equ: sov} let us set $(a,c) = (z^{-1}, bz)$:
$$ [\BW_{b,z^{-2}}^{(i)}] [\BW_{z,1}^{(i)}] = [\BW_{zb,1}^{(i)}] [\BW_{1,z^{-2}}^{(i)}]. $$
Taking transfer matrices and evaluating them at $1$ gives the special case $c=1$ of Equation \eqref{equ: generalized TQ}, which in turn implies the general case $c \in \BC^{\times}$. The second statement is a translation of that of Lemma \ref{lem: separation of variables}.
\end{proof}

\begin{example} \label{example: gl(2|2) generalized TQ}
Let $\Glie = \mathfrak{gl}(2|2)$ and $X = W_{1,1}^{(1)} = V_q^+(\epsilon_1;q^{-1})$. By Equation \eqref{for: iMA}:  
$$ \chi_q(X) = \boxed{1}_1 + \boxed{2}_1 + \boxed{3}_1 + \boxed{4}_1. $$
If $s \in \super, g(z) \in \BC[[z]]^{\times}$ and $c \in \BC^{\times}$, for simplicity let $sg(z) := (g(z)^4;s) \in \lCP$, $[s,g(z)] := [L(g(z)^4;s)] \in K_0(\BGG)$ and $\langle s, c\rangle := (c^4;s) \in \CP$. Set $w_{c,a}^{(i)} := f_{c,a}^{(i)}(z) \aBw_{c,a}^{(i)}$. By Definition \ref{def: tableau}, Example \ref{ss: l-weights} and Lemma \ref{lem: polyn asym}:
\begin{gather*}
\boxed{1}_1 = \left(\frac{q-z}{1-zq},1,1,1;\even\right),\quad \boxed{2}_1 = \left(1,\frac{q-zq^2}{1-zq^3},1,1;\even\right),\\
\boxed{3}_1 = \left(1,1,\frac{1-zq^3}{q-zq^2},1;\odd\right),\quad \boxed{4}_1 = \left(1,1,1,\frac{1-zq}{q-z};\odd\right), \\
\frac{w_{c,a}^{(1)}}{w_{1,a}^{(1)}} = \left(\frac{c-zac^{-1}}{1-za},1,1,1;\even\right), \quad \frac{w_{c,a}^{(2)}}{w_{1,a}^{(2)}}= \left(\frac{c-zaqc^{-1}}{1-zaq},\frac{c-zaqc^{-1}}{1-zaq},1,1;\even\right), \\
\frac{w_{c,a}^{(3)}}{w_{1,a}^{(3)}} = \left(\frac{1-zac^{-2}}{1-za},\frac{1-zac^{-2}}{1-za},\frac{1-zac^{-2}}{1-za},c^{-1};\even \right),\quad \boxed{1}_1 = \frac{w_{q,q}^{(1)}}{w_{1,q}^{(1)}},  \\
\boxed{2}_1 = \frac{w_{q^{-1},q}^{(1)}}{w_{1,q}^{(1)}} \frac{w_{q,q^2}^{(2)}}{w_{1,q^2}^{(2)}},\quad \boxed{3}_1 = \odd  q^{-1} \frac{w_{q,q^2}^{(2)}}{w_{1,q^2}^{(2)}} \frac{w_{q^{-1},q}^{(3)}}{w_{1,q}^{(3)}}, \quad \boxed{4}_1 = \odd \frac{1-zq}{1-zq^{-1}} \frac{w_{q,q}^{(3)}}{w_{1,q}^{(3)}}.
\end{gather*}
It follows that in the fractional ring of $K_0(\BGG)$: 
\begin{align*}
[X] = \frac{[\BW_{q,q}^{(1)}]}{[\BW_{1,q}^{(1)}]} + \frac{[\BW_{q^{-1},q}^{(1)}]}{[\BW_{1,q}^{(1)}]} \frac{[\BW_{q,q^2}^{(2)}]}{[\BW_{1,q^2}^{(2)}]}  + [\odd, q^{-1}]\frac{[\BW_{q,q^2}^{(2)}]}{[\BW_{1,q^2}^{(2)}]} \frac{[\BW_{q^{-1},q}^{(3)}]}{[\BW_{1,q}^{(3)}]} + [\odd, \frac{1-zq}{1-zq^{-1}}]\frac{[\BW_{q,q}^{(3)}]}{[\BW_{1,q}^{(3)}]}.
\end{align*}
Let $q^{\frac{1}{2}}$ be a square root of $q$. By Example \ref{examp: one-dim transfer} and Equation \eqref{equ: generalized TQ}:
\begin{align*}
t_X(z^{-2}) &= \frac{Q_1(zq^{\frac{1}{2}})}{Q_1(zq^{-\frac{1}{2}})} + \frac{Q_1(zq^{-\frac{3}{2}})}{Q_1(zq^{-\frac{1}{2}})} \frac{Q_2(z)}{Q_2(zq^{-1})} + \langle\odd,q^{-1}\rangle \times \frac{Q_2(z)}{Q_2(zq^{-1})} \frac{Q_3(zq^{-\frac{3}{2}})}{Q_3(zq^{-\frac{1}{2}})}q^{-\ell} \\
&\quad + \langle\odd,1\rangle \times   \frac{Q_3(zq^{\frac{1}{2}})}{Q_3(zq^{-\frac{1}{2}})}\prod_{l=1}^{\ell} \frac{z^2-b_lq}{z^2-b_lq^{-1}}.
\end{align*}
\end{example} 
\begin{example} \label{example: TQ gl(2)}
Let $\Glie = \mathfrak{gl}(2|0)$ and $X = W_{1,1}^{(1)} = V_q^+(\epsilon_1;q)$. Then 
\begin{gather*}
 \boxed{1}_{q^2} + \boxed{2}_{q^2} = \left(\frac{q-zq^{-2}}{1-zq^{-1}},1;\even\right) + \left( 1,\frac{q-z}{1-zq};\even\right) = \frac{w_{q,q}^{(1)}}{w_{1,q}^{(1)}} + \frac{q-z}{1-zq} \frac{w_{q^{-1},q}^{(1)}}{w_{1,q}^{(1)}}, \\
 t_X(z^{-2}) = \frac{Q_1(zq^{\frac{1}{2}})}{Q_1(zq^{-\frac{1}{2}})} + \langle\even,q \rangle \times \frac{Q_1(zq^{-\frac{3}{2}})}{Q_1(zq^{-\frac{1}{2}})}\prod_{l=1}^{\ell} \frac{qz^2-b_l}{z^2-b_lq} . 
\end{gather*}
\end{example}
\begin{example} \label{example: TQ gl(1|1)}
Let $\Glie = \mathfrak{gl}(1|1)$ and $X = W_{1,1}^{(1)} = V_q^+(\epsilon_1;q^{-1})$. We have 
\begin{gather*}
\chi_q(X) = \boxed{1}_{1} + \boxed{2}_{1} = \left(\frac{q-z}{1-zq},1;\even\right) + \left( 1,\frac{1-zq}{q-z};\odd\right) = \frac{w_{q,q}^{(1)}}{w_{1,q}^{(1)}} \left(1 + \odd \frac{1-zq}{q-z} \right),  \\
t_X(z^{-2}) = \frac{Q_1(zq^{\frac{1}{2}})}{Q_1(zq^{-\frac{1}{2}})}  + \langle \odd,q^{-1}\rangle \times \frac{Q_1(zq^{\frac{1}{2}})}{Q_1(zq^{-\frac{1}{2}})}  \prod_{l=1}^{\ell} \frac{z^2-b_lq}{z^2q-b_l}. 
\end{gather*}
\end{example}
One can view Examples \ref{example: TQ gl(2)}--\ref{example: TQ gl(1|1)} as degenerate cases of Example \ref{example: gl(2|2) generalized TQ}.

We are ready to to deduce three-term functional relations of the Baxter operators $Q_i(z)$. Fix $a = 1$. Let $c,d \in \BC^{\times}$ be such that $c^2 \notin q^{\BZ}$. In Equation \eqref{equ: positive TQ asy} let us evaluate transfer matrices at $z^{-2}$ making use of Proposition \ref{prop: transfer matrices}:
\begin{align*}
 t_{M_{c,1}^{(i)}}(z^{-2}) t_{\CW_{d,d^2}^{(i)}}(z^{-2}) &= t_{\CW_{dq_i,d^2}^{(i)}}(z^{-2}) \prod_{j\in I_0: j\sim i} t_{\CW_{c_{ij}^{-1},q_{ij}^{-1}c_{ij}^{-2}}^{(j)}}(z^{-2}) \\
 &\quad + t_{D_i}(z^{-2}) t_{\CW_{d\hat{q}_i^{-1},d^2}^{(i)}}(z^{-2}) \prod_{j\in I_0: j \sim i}t_{\CW_{c_{ij}^{-1}q_{ij}^{-1},q_{ij}^{-1}c_{ij}^{-2}}^{(j)}}(z^{-2}).
\end{align*}
Dividing both sides by the term at the second row without $t_{D_i}(z^{-2})$ and making use of Equation \eqref{equ: generalized TQ}, we obtain the {\it Baxter TQ relation}:
\begin{equation} \label{equ: Baxter TQ}
X_c^{(i)}(z) \frac{Q_i(z)}{Q_i(z\hat{q}_i^{-1})} = y_i(z) \frac{Q_i(zq_i)}{Q_i(z\hat{q}_i^{-1})} \prod_{j\in I_0: j\sim i} \frac{Q_j(zq_{ij}^{\frac{1}{2}})}{Q_j(zq_{ij}^{-\frac{1}{2}})} + t_{D_i}(z^{-2}),
\end{equation}
where $X_c^{(i)}(z)$ (depending on $c \in \BC^{\times} \setminus q^{\BZ}$) and $y_i(z)$ are given by
\begin{align*}
X_c^{(i)}(z) &= \frac{t_{M_{c,1}^{(i)}}(z^{-2})}{\prod\limits_{j\in I_0: j \sim i}t_{\CW_{c_{ij}^{-1}q_{ij}^{-1},aq_{ij}^{-1}c_{ij}^{-2}}^{(j)}}(z^{-2})} \times \prod_{l=1}^{\ell} \frac{f_{d\hat{q}_i^{-1},d^2}^{(i)}(z^{-2}b_l)}{f_{d,d^2}^{(i)}(z^{-2}b_l)}, \\
y_i(z) &= \prod_{l=1}^{\ell} \left( \frac{f_{d\hat{q}_i^{-1},d^2}^{(i)}(z^{-2}b_l)}{f_{dq_i,d^2}^{(i)}(z^{-2}b_l)} \times \prod_{j\in I_0:j\sim i} \frac{f_{c_{ij}^{-1}q_{ij}^{-1},q_{ij}^{-1}c_{ij}^{-2}}^{(j)}(z^{-2}b_l)}{f_{c_{ij}^{-1},q_{ij}^{-1}c_{ij}^{-2}}^{(j)}(z^{-2}b_l)}  \right).
\end{align*}
Note that $y_i(z), D_i$ are independent of $c,d$ by Lemma \ref{lem: polyn asym} and Theorem \ref{thm: TQ asymptotic}. 

Let us assume that the twisted transfer matrices in Equation \eqref{for: twisted transfer matrix} are well-defined for all the $M_{c,1}^{(i)}$ and $\CW_{c,a}^{(i)}$, upon a generic choice of $\varphi: \lCP \longrightarrow \BC^{\times}$; this corresponds to the convergence assumption in \cite[Remark 5.12 (ii)]{FH}. Then Equation \eqref{equ: Baxter TQ} is an operator equation in $\End(\BV^{\otimes \ell})[[z^{-2}]]$. 

Based on the asymptotic construction of $\CW_{c,a}^{(i)}$, one can show that there exists $n \in \BZ$ such that $z^n Q_i(z)$ is a polynomial in $z$ with values in $\End(\BV^{\otimes \ell})$.

As in \cite[Section 5]{FH2}, we expect that the $t_{M_{c,1}^{(i)}}(z^{-2})$ are polynomials in $z^{-2}$ (up to multiplication by an integer power of $z$). Suppose that $w$ is a zero of $Q_i(z)$ that is neither a zero of $Q_i(z\hat{q}_i^{-1}), Q_j(zq_{ij}^{-\frac{1}{2}})$ nor a pole of $X_c^{(i)}(z)$. Then we have the {\it Bethe Ansatz Equation}: (see \cite[(2.6a)]{Tsuboi} and \cite{B2,TQYangian})
\begin{equation} \label{equ: Bethe Ansatz Equation}
y_i(w) \frac{Q_i(wq_i)}{Q_i(w\hat{q}_i^{-1})} \prod_{j\in I_0: j\sim i} \frac{Q_j(wq_{ij}^{\frac{1}{2}})}{Q_j(wq_{ij}^{-\frac{1}{2}})} = -t_{D_i}(w^{-2}).
\end{equation}
\begin{example} \label{example: BAE gl(2|2)}
Following Example \ref{example: gl(2|2) generalized TQ}, we determine the highest $\ell$-weight (still denoted by $D_i$) of the one-dimensional $U_q(\Gaff)$-module $D_i$ and the $y_i(z)$ in Equation \eqref{equ: Bethe Ansatz Equation} for $\Glie = \mathfrak{gl}(2|2)$. First by Definition \ref{def: tableau} and Example \ref{ss: l-weights}:
\begin{gather*}
\aBw_{c,a}^{(1)} = \left(\frac{c-zac^{-1}}{1-za},1,1,1;\even\right),\quad \aBw_{c,a}^{(2)} = \left(\frac{c-zaqc^{-1}}{1-zaq},\frac{c-zaqc^{-1}}{1-zaq},1,1;\even\right), \\
\aBw_{c,a}^{(3)} = \left(1,1,1,\frac{1-za}{c-zac^{-1}};\even\right),\quad A_{1,a} = \left(\frac{q-zaq^{-1}}{1-za},\frac{1-zaq^2}{q-zaq},1,1 ;\even\right), \\
A_{2,a} = \left(1,\frac{q-za}{1-zaq},\frac{q-za}{1-zaq},1;\odd\right),\quad A_{3,a} =  \left(1,1,\frac{1-zaq^2}{q-zaq},\frac{q-zaq^{-1}}{1-za} ;\even\right).
\end{gather*}
The relations between $A$ and $\aBw$ are as follows: $A_{1,a} = \aBw_{q^2,aq^2}^{(1)} \aBw_{q^{-1},aq^{-1}}^{(2)}$ and
\begin{gather*}
 A_{2,a} = \odd \frac{q-za}{1-zaq}\aBw_{q^{-1},aq^{-1}}^{(1)}\aBw_{q,aq}^{(3)}, \quad A_{3,a} = \frac{1-zaq^2}{q-zaq} \aBw_{q,aq}^{(2)}\aBw_{q^{-2},aq^{-2}}^{(3)}.
\end{gather*}
It it follows that $ D_1 = 1,\ D_2 = \odd \frac{1-zq}{q-z},\ D_3 = \frac{q-zq}{1-zq^2}$ and so ($D_i(z) := t_{D_i}(z^{-2})$)
\begin{gather*}
 D_1(z) = 1,\quad D_2(z) = \langle\odd,q^{-1}\rangle \times \prod_{l=1}^{\ell} \frac{z^2-b_lq}{z^2q-b_l},\quad D_3(z) = \langle \even,q\rangle \times \prod_{l=1}^{\ell} \frac{z^2q-b_lq}{z^2-b_lq^2}, \\
 y_1(z) = 1,\quad y_2(z) = \prod_{l=1}^{\ell} \frac{z^2-b_lq}{z^2-b_lq^{-1}}, \quad y_3(z) = \prod_{l=1}^{\ell} \frac{z^2-b_lq^{-2}}{z^2-b_lq^2}.
\end{gather*}
The Bethe Ansatz Equations become in this case: 
\begin{gather*}
 \frac{Q_1(w_1q)}{Q_1(w_1q^{-1})} \frac{Q_2(w_1q^{-\frac{1}{2}})}{Q_2(w_1q^{\frac{1}{2}})} = -1,\quad \frac{Q_1(w_2q^{-\frac{1}{2}})}{Q_1(w_2q^{\frac{1}{2}})}\frac{Q_3(w_2q^{\frac{1}{2}})}{Q_3(w_2q^{-\frac{1}{2}})} = - \langle\odd,q^{-1}\rangle \times q^{-\ell}, \\
 \frac{Q_3(w_3q^{-1})}{Q_3(w_3q)}\frac{Q_2(w_3q^{\frac{1}{2}})}{Q_2(w_3q^{-\frac{1}{2}})} = -\langle \even,q\rangle \times  \prod_{l=1}^{\ell}  \frac{w_3^{2}q-b_lq}{w_3^{2}-b_lq^{-2}},
\end{gather*}
where $w_i$ is a zero of $Q_i(z)$ for $1\leq i \leq 3$.
\end{example}

The generalized Baxter relations in Lemma \ref{lem: separation of variables} and Bethe Ansatz Equations \eqref{equ: Bethe Ansatz Equation} for the Baxter operators $Q_i(z)$ are based on asymptotic $U_q(\Gaff)$-modules: $\CW_{c,a}^{(i)}, \CN_{c,a}^{(i)}, M_{c,a}^{(i)}$, whereas in recent parallel works \cite{Jimbo1,HL,FH2,Jimbo2} representations of Borel subalgebras ($Y_q(\Glie)$ in our situation) play a key role. 

In \cite{B2,TQYangian}, for the Yangian of $\mathfrak{gl}(M|N)$ the Baxter operators $\BQ_J(z)$ are labeled by the subsets $J$ of $I$. In addition to TQ relations, there are algebraic relations among the $\BQ_J(z)$ called QQ relations.  Our $Q_i(z)$ with $i \in I_0$ seem to be algebraically independent by Proposition \ref{prop: simplicity tensor product asymp}; see also \cite[Theorem 4.11]{FH}.

\begin{rem}
Following \cite{BT,FH} define $\BQ_i(z) := t_{L_{i,1}^+}(z)$ for $i \in I_0$. We have
\begin{equation} \label{for: comp two Q}
t_{L([c]_i)}(z^{-2}) \frac{\BQ_i(z^{-2}c^{-2})}{\BQ_i(z^{-2})} = \prod_{l=1}^{\ell} \frac{f_{1,1}^{(i)}(z^{-2}b_l^{-2})}{f_{c,1}^{(i)}(z^{-2}b_l^{-2})} \times \frac{Q_i(zc)}{Q_i(z)}
\end{equation}
 based on the $q$-character formula $\frac{\chi_q(\CW_{c,1}^{(i)})}{\chi_q(\CW_{1,1}^{(i)})} = [c]_i \frac{\chi_q(L_{i,c^{-2}}^+)}{\chi_q(L_{i,1}^+)}$
and Equation \eqref{equ: generalized TQ}. See \cite[Remark A.7]{FZ} for a similar comparison in the Yangian case. 
\end{rem}

\end{document}